%% file: main.tex
\documentclass[11pt]{article}

\input{Paper/macros}

\input{Paper/math}

\title{On the Efficiency of Fair and Truthful Trade Mechanisms}

\author{
Moshe Babaioff\thanks{Hebrew University of Jerusalem. Email: {\tt moshe.babaioff@mail.huji.ac.il}} \and Yiding Feng\thanks{Hong Kong University of Science and Technology. Email: {\tt ydfeng@ust.hk}} \and Noam Manaker Morag\thanks{Hebrew University of Jerusalem. Email: {\tt noam.manakermorag@mail.huji.ac.il}}
}
\date{}

\begin{document}

\maketitle
\begin{abstract}
    \input{Paper/abstract}
\end{abstract}

\thispagestyle{empty}
\newpage

\section{Introduction}
\input{Paper/intro}

\section{Preliminaries}
\input{Paper/prelim}

\section{KS-Fairness}
\label{sec:ksfairness}
\input{Paper/ks-fairness}

\section{GFT Approximation under KS-Fairness for Bilateral Trade Instances}
\label{sec:general-gft-approx}
\input{Paper/black-box-reduction}

\section{GFT Approximation under KS-Fairness for Zero-Value Seller Instances}
\input{Paper/improved-GFT-result}

\section{Nash Social Welfare Maximization}
\label{sec:nash fairness}
\input{Paper/NSWM-result}

\section{Conclusion and Future Directions}
\label{sec:conclusion}
\input{Paper/conclusion}

\subsection*{Acknowledgments} 
Moshe Babaioff's research is supported in part by a Golda Meir Fellowship and the Israel Science Foundation (grant No.\ 301/24).

\bibliographystyle{apalike}
	\bibliography{refs.bib}

\appendix

\section{Cooperative Bargaining}
\label{appendix:bargaining-and-trade}
\input{Paper/bargaining-and-bilateral-trade}

\subsection{Generalizing \Cref{sec:general-gft-approx} Results to Cooperative Bargaining}
\label{appendix:general-bargaining-results}
\input{Paper/general-bargaining-results}

\section{Equitability}
\label{appendix:equitablity}

\input{Paper/equitable-mechanisms}

\section{Interim KS-Fairness and Ex Post KS-Fairness} 
\label{subsec:interim ks fairness}
\input{Paper/interim-ex-post-ks-fairness}

\section{Omitted Proof for \Cref{lem:GFT program:mhr buyer}}
\input{Paper/apx-lemgftprogrammhrbuyer}

\section{Details about Numerical Evaluation of Program~\ref{program:GFT:regular buyer}}
\label{apx:numerical evaluation:regular buyer}
\input{Paper/apx-numerical-evaluation-regular-buyer}

\section{Details about Numerical Evaluation of Program~\ref{program:GFT:mhr buyer}}
\label{apx:numerical evaluation:mhr buyer}
\input{Paper/apx-numerical-evaluation-mhr-buyer}

\end{document}

%% file: Paper/macros.tex
\usepackage{amsfonts,amsmath,amssymb,amsthm,graphicx}
\usepackage{fullpage}
\usepackage{changepage}
\usepackage{enumerate}
\usepackage{scalerel}
\usepackage{accents}
\usepackage[margin=1in]{geometry}
\usepackage{float}
\usepackage[caption = false]{subfig}
\usepackage{hyperref}
\usepackage{enumitem}
\usepackage[normalem]{ulem}
\usepackage{array}

\usepackage{natbib}
\usepackage{csquotes}
\usepackage{comment}

\usepackage{bbm}

\usepackage{cleveref}
\crefname{claim}{claim}{claims}

\usepackage{tikz}
\usepackage{pgfplots}
\pgfplotsset{compat=1.17} 
\usetikzlibrary{patterns}
\usepgfplotslibrary{fillbetween}
\usetikzlibrary{intersections}

\usepackage[ruled,vlined,linesnumbered]{algorithm2e}
\Crefname{algocf}{Algorithm}{Algorithms}
\allowdisplaybreaks

\usetikzlibrary{arrows, decorations.markings}

\tikzstyle{vecArrow} = [thick, decoration={markings,mark=at position
   1 with {\arrow[semithick]{open triangle 60}}},
   double distance=1.4pt, shorten >= 5.5pt,
   preaction = {decorate},
   postaction = {draw,line width=1.4pt, white,shorten >= 4.5pt}]
\tikzstyle{innerWhite} = [semithick, white,line width=1.4pt, shorten >= 4.5pt]

\allowdisplaybreaks
\theoremstyle{plain}
\newtheorem{theorem}{Theorem}[section]
\newtheorem{lemma}[theorem]{Lemma}

\newtheorem{proposition}[theorem]{Proposition}
\newtheorem{informal}{Informal Theorem}

\theoremstyle{plain}
\newtheorem{definition}{Definition}[section] % definition numbers are dependent on theorem numbers
\newtheorem{example}[definition]{Example}

\theoremstyle{plain}

\usepackage{xfrac}
\usepackage{thm-restate}

\newcommand{\xhdr}[1]{\vspace{2mm} \noindent{\bf #1}}

\newcommand{\Cost}[2][]{\text{\bf COST}\ifthenelse{\not\equal{}{#1}}{_{#1}}{}\!\left[{\def\givenn{\middle|}#2}\right]}
\renewcommand{\d}{\mathrm{d}}   % for integrals
\newcommand{\OPT}{\texttt{OPT}}

\newcommand{\mech}{\mathcal{M}}
\newcommand{\alloc}{x}
\newcommand{\price}{p}
\newcommand{\cost}{c}
\newcommand{\hcost}{\bar\cost}
\newcommand{\val}{v}
\newcommand{\buyerdist}{F}

\newcommand{\sellerdist}{G}

\newcommand{\lval}{\underline{\val}}

\newcommand{\primed}{^\dagger}
\newcommand{\doubleprimed}{^\ddagger}
\newcommand{\lcost}{\underline{\cost}}
\newcommand{\hval}{\bar{\val}}

 \newcommand{\reals}{\mathbb{R}}

\newcommand{\buyerutil}{u}
\newcommand{\sellerutil}{\pi}

\newcommand{\selleralloc}{\tilde{\alloc}}

\newcommand{\GFT}[2][]{\texttt{GFT}\ifthenelse{\not\equal{}{#1}}{_{#1}}{}\!\left[{\def\givenn{\middle|}#2}\right]}

\newcommand{\buyerexanteutil}{U}
\newcommand{\sellerexanteutil}{\Pi}

\newcommand{\RO}{\texttt{RO}}
\newcommand{\quant}{q}

\newcommand{\virtualval}{\psi}

\newcommand{\calC}{\mathcal{C}}
\newcommand{\optquant}{\quant_m}
\newcommand{\reserve}{r}
\newcommand{\optreserve}{\reserve_m}

\usepackage{bm}         % <----- for '\bm' command (optional)
\usepackage{mathtools}  % <----- for '\mathrlap' command (necessary)

\newcommand{\mybar}[3]{%
    \mathrlap{\hspace{#2}\overline{\scalebox{#1}[1]{\phantom{\ensuremath{#3}}}}}\ensuremath{#3}
}

\newcommand{\myunderbar}[3]{%
    \mathrlap{\hspace{#2}\underline{\scalebox{#1}[1]{\phantom{\ensuremath{#3}}}}}\ensuremath{#3}
}

\newcommand{\HUnderBar}{\myunderbar{0.8}{0pt}{H}}
\newcommand{\HOverBar}{\mybar{0.68}{2.5pt}{H}}
\newcommand{\MUnderBar}{\myunderbar{0.8}{0pt}{M}}
\newcommand{\MOverBar}{\mybar{0.68}{3pt}{M}}
\newcommand{\LUnderBar}{\myunderbar{0.8}{0pt}{L}}
\newcommand{\LOverBar}{\mybar{0.68}{1.2pt}{L}}

\newcommand{\NoTrade}{\sf No Trade Mechanism}

\newcommand{\FixPrice}{{\sf Fixed Price Mechanism}}
\newcommand{\FixPrices}{{\sf Fixed Price Mechanisms}}
\newcommand{\FPM}{{\sf FPM}}
\newcommand{\RandomOffer}{{\sf Random Offer Mechanism}}
\newcommand{\ROM}{{\sf ROM}}
\newcommand{\BuyerOffer}{{\sf Buyer Offer Mechanism}}
\newcommand{\BOM}{{\sf BOM}}
\newcommand{\SellerOffer}{{\sf Seller Offer Mechanism}}
\newcommand{\SOM}{{\sf SOM}}
\newcommand{\BiasedRandomOffer}{{\sf $\mixprob$-Biased Random Offer Mechanism}}
\newcommand{\NashSocialWelfareMaximizers}{{\sf NSW-Maximizing Mechanisms}}
\newcommand{\BROM}{{\sf $\mixprob$-ROM}}

\newcommand{\SecondBest}{{\sf Second-Best Benchmark}}
\newcommand{\FirstBest}{{\sf First-Best Benchmark}}

\newcommand{\NashSocialWelfareMaximizer}{{\sf NSW-Maximizing Mechanism}}
\newcommand{\NSWM}{{\sf NSWM}}

\newcommand{\sellerbenchmark}{\sellerexanteutil^*}
\newcommand{\buyerbenchmark}{\buyerexanteutil^*}
\newcommand{\supp}{{\sf supp}}

\newcommand{\revcurve}{R}
\newcommand{\ironrevcurve}{\bar{\revcurve}}
\newcommand{\buyercdf}{\buyerdist}
\newcommand{\buyerpdf}{\buyerdist'}
\newcommand{\sellercdf}{\sellerdist}
\newcommand{\sellerpdf}{\sellerdist'}
\newcommand{\buyerhazardrate}{\phi}

\newcommand{\LambertFunc}{W}

\newcommand{\fixedPriceGFTPercentageRegular}{{85.1\%}}
\newcommand{\fixedPriceGFTPercentageMHR}{{91.3\%}}
\newcommand{\fixedPriceGFTPercentageUBRegular}{{87.7\%}}
\newcommand{\fixedPriceGFTPercentageUBMHR}{{94.4\%}}

\newcommand{\mixprob}{\lambda}

\newcommand{\eps}{\varepsilon}

\newcommand{\constantH}{K}

\newcommand{\cumhazard}{\Phi}

\newcommand{\sellerprice}{\tilde{\price}}

\newcommand{\OPTSB}{\OPT_{\textsc{SB}}}
\newcommand{\OPTFB}{\OPT_{\textsc{FB}}}

\newcommand{\GFTapprox}{\calC}

\newcommand{\ksfairness}{{KS-fairness}}
\newcommand{\ksfair}{{KS-fair}}
\newcommand{\equitable}{{equitable}}
\newcommand{\equitability}{{equitability}}

\newcommand{\ksline}{{KS-line}}

\newcommand{\fprice}{\price_f}
\newcommand{\fquant}{\quant_f}

\newcommand{\exanteutilratio}{\Gamma}
\newcommand{\revratio}{\alpha}
\newcommand{\residuesurplusratio}{\beta}

\newcommand{\mechfamily}{\mathfrak{M}_{\buyerdist,\sellerdist}}
\newcommand{\mechfam}{\mathfrak{M}}

\newcommand{\GFTp}[2][]{\tilde{\texttt{GFT}}\ifthenelse{\not\equal{}{#1}}{_{#1}}{}\!\left[{\def\givenn{\middle|}#2}\right]}

\newcommand{\btinstance}{(\buyerdist,\sellerdist)}

\newcommand{\bargainprob}{S}
\newcommand{\bargainprobBT}{\bargainprob_{\buyerdist,\sellerdist}}
\newcommand{\dispnt}{d}
\newcommand{\bargainsolution}{f}
\newcommand{\RR}{\mathbb{R}}

\newcommand{\solnash}{N}
\newcommand{\solks}{K}
\newcommand{\soleg}{E}

\newcommand{\buyercoord}[1]{\buyerutil(#1)}
\newcommand{\sellercoord}[1]{\sellerutil(#1)}

\newcommand{\idealpnt}{a}
\newcommand{\idealp}[1]{\idealpnt^*(#1)}
\newcommand{\idealbuyer}[1]{\buyercoord{\idealp{#1}}}
\newcommand{\idealseller}[1]{\sellercoord{\idealp{#1}}}

\newcommand{\idealutility}[1]{A^*(#1)}

\newcommand{\buyerbenchmarkpoint}{o^\buyerutil}
\newcommand{\sellerbenchmarkpoint}{o^\sellerutil}

\newcommand{\pnta}{x}
\newcommand{\pntb}{y}
\newcommand{\pntc}{z}
\newcommand{\naturals}{\mathbb{N}}

%% file: Paper/math.tex
	\DeclareMathOperator{\argmax}{argmax}

\newcommand{\prob}[2][]{\text{Pr}\ifthenelse{\not\equal{}{#1}}{_{#1}}{}\!\left[{\def\givenn{\middle|}#2}\right]}
\newcommand{\expect}[2][]{\mathbb{E}\ifthenelse{\not\equal{}{#1}}{_{#1}}{}\!\left[{\def\givenn{\middle|}#2}\right]}

\newcommand{\tparen}{\big}
\newcommand{\tprob}[2][]{\text{Pr}\ifthenelse{\not\equal{}{#1}}{_{#1}}{}\tparen[{\def\given{\tparen|}#2}\tparen]}
\newcommand{\texpect}[2][]{\mathbb{E}\ifthenelse{\not\equal{}{#1}}{_{#1}}{}\tparen[{\def\given{\tparen|}#2}\tparen]}

\newcommand{\sprob}[2][]{\text{Pr}\ifthenelse{\not\equal{}{#1}}{_{#1}}{}[#2]}
\newcommand{\sexpect}[2][]{\mathbb{E}\ifthenelse{\not\equal{}{#1}}{_{#1}}{}[#2]}

\newcommand{\indicator}[1]{{\mathbbm{1}\left\{ #1 \right\}}}
\newcommand{\plus}[1]{\left( #1 \right)^+}

\newcommand{\bigO}{\mathcal{O}}

\newcommand{\Obj}[2][]{\textsc{obj}\ifthenelse{\not\equal{}{#1}}{_{#1}}{}\!\left[{\def\givenn{\middle|}#2}\right]}

%% file: Paper/abstract.tex
We consider the impact of fairness requirements on the social efficiency of truthful mechanisms for trade, focusing on Bayesian bilateral-trade settings. 
Unlike the full information case in which all gains-from-trade can be realized and equally split between the two parties, in the private information setting, equitability has devastating welfare implications (even if only required to hold ex-ante). 
We thus search for an alternative fairness notion and suggest requiring the mechanism to be {\ksfair}: it must ex-ante equalize the fraction of the ideal utilities of the two traders.  
We show that there is always a {\ksfair} (simple) truthful mechanism with expected gains-from-trade that are half the optimum, but always ensuring any better fraction is impossible (even when the seller value is zero). 
We then restrict our attention to trade settings with a zero-value seller and a buyer with valuation distribution that is Regular or MHR, proving that much better fractions can be obtained under these conditions, with simple posted-price mechanisms. 

%% file: Paper/intro.tex
The field of \emph{Mechanism Design} studies the design of mechanisms that obtain good outcomes (as high social welfare or revenue) in the presence of strategic traders. In this paper, we consider imposing a fairness requirement on truthful trade mechanisms\footnote{A trade mechanism must also encourage participation even after an agent knows her value (interim individually rational), and not run a deficit.}, and the impact of such a requirement on the economic efficiency (social welfare) of the mechanism.  That is, a mechanism results in trade which generates gains, and (some of these) gains are allocated to the traders. We take a normative approach and search for simple truthful mechanisms that are constrained to distribute these gains ``fairly'', 
studying their economic efficiency. This problem can be viewed as searching for a fair solution to the bargaining problem between the two traders, but in settings with private valuations. Thus, for trade settings, our work focuses on \emph{the implications of the combination of strategic behavior and fairness requirements on the efficiency of the outcome}.

We focus on the fundamental mechanism-design problem of bilateral trade in the Bayesian setting, that is, a seller that sells a single good to a buyer, where traders' values for the good are private, but are sampled independently from a known Bayesian prior.\footnote{Note that the seller is not the mechanism designer, but rather a trader in the mechanism.} The most basic of these problems is the setting where the seller has no value for the good.\footnote{This model is equivalent to a seller with known value for keeping the item: we can assume that value is $0$ by applying a simple normalization. Thus, we assume a value of $0$ for simplicity. The gains-from-trade (GFT) is invariant to the normalization (but not the welfare), and is harder to approximate than the welfare (so positive results for GFT are better).} We call this basic setting the \emph{zero-value seller setting}.

As an example, consider a zero-value seller setting with a buyer that has a value of $2$ for the good. In this simple setting there is no uncertainty about the values (a full information setting). If the seller can dictate the mechanism, she will use the seller's optimal mechanism ({\SellerOffer}) and post an optimal price of $2$. She will always sell the item,\footnote{Here the buyer is indifferent between buying or not, but for any $\varepsilon>0$, the seller can get revenue of $2-\varepsilon$ which giving a strict incentive to buy. Throughout the paper we ignore this issue and allow ties to be broken as needed.} obtaining utility of $2$, leaving the buyer with utility of $0$. This outcome seems highly unfair to the buyer, since the gains-from-trade (GFT) of $2$ are the result of a trade that cannot take place without the buyer, yet the buyer receives none of these gains. The alternative mechanism in which trade happens at a price of $1$, results with an efficient trade in which each trader gets utility of $1$, and that mechanism seems much more fair. More generally, when the value of the buyer is known to be $\val > 0$ (no private information), the mechanism in which trade happens at the price $\val/2$ maximizes the gains and seems to be perfectly fair to both traders, as it is \emph{equitable} (both get the same utility).\footnote{In fact, in such a full information setting, any symmetric solution to the bargaining problem (e.g., the Nash solution, or the egalitarian solution) will result with the same outcome where the two parties equally split the gains.} We thus see that when there is no private information, there is a mechanism that maximizes the GFT and is ``perfectly fair''.  

This work focuses on the more involved situation where traders do have private information, and are acting strategically. As an example, consider a seller of a digital good (zero cost/value for the good) to a population of buyers with values distributed according to some known distribution $\buyerdist$. For example, $\buyerdist$ might be the uniform distribution over $[0,1]$. As it is a digital good, the setting is equivalent to a single-buyer setting with value sampled from the valuation distribution. Let us consider the following three properties of a trade mechanism: being \emph{truthful}, being \emph{fair}, and \emph{maximizing the gains-from-trade (GFT)}. Without any fairness considerations, all GFT can be realized by a truthful mechanism that always trades at price 0. If we do not care about gains, the mechanism that never trades is truthful and can be considered fair under various definitions (e.g., it equalizes the utilities). Finally, disregarding the issue of strategic behavior (assuming access to the true values), the problem reduces to the full information case, for which we saw that it is possible to realize all gains and split them equally. We thus see that any two out of three properties can be achieved. Can all three be achieved together? Unfortunately, it seems that the answer is no. For example, in the zero-value seller setting and a buyer with value uniformly distributed over $[0,1]$, the only way to truthfully maximize the gains is to trade at 0, but that mechanism leaves the seller with no profit at all, which intuitively seems unfair (and is also formally unfair under essentially any fairness definition we can conceive). 

We thus relax the goal of exactly maximizing the GFT and only ask to approximate it in expectation, where approximation is with respect to {the {\SecondBest}, that is, the} maximum expected GFT {achievable by a mechanism} without fairness constraints, but with incentive constraints as well as 
participation constraints and the constraint of the mechanism not running a deficit. We ask:   
\emph{For which fairness notions can the truthful trade mechanism guarantee a good GFT approximation? For those fairness notions where a good approximation is possible, how good can the approximation be?}

We move to consider the problem {of picking a fair truthful mechanism} from an ex ante perspective (in expectation, before any values are realized). If the seller is a monopolist in the market, she can price at a Myerson price and maximize her own expected revenue, but this might result in very low expected utility for the buyer (such as in the case of full information), an outcome that can be considered very unfair. A market regulator (or a court) might impose the constraint that the mechanism used for trade be fair to both parties (at least ex ante). But, what kind of fairness can be required? Possibly the most natural would be to require equitability (at least ex ante). Unfortunately, such a definition is too stringent and does not allow for good GFT approximation, as we illustrate next.     

Consider a zero-value seller setting with a buyer that has a private value $\val$, sampled from the Equal-Revenue (ER) distribution with support $[1,\constantH]$ (where $\constantH>1$ is some constant).\footnote{The Equal-Revenue distribution (with support $[1,\constantH]$) has probability $1/\price$ of having value at least $\price\in [1,\constantH]$. Thus the expected revenue from any fixed trading price $\price\in [1,\constantH]$ is $1$.} As the buyer value is private information, we now consider the utilities of the parties in mechanisms that are truthful. It is well known \citep{mye-81} that the seller's ex ante utility in any such mechanism is at most $1$, and thus an ex ante equitable mechanism can yield total gains for both traders \footnote{Here we only consider the utilities of the two traders (disregarding any money kept by the mechanism). Our results in the paper allow for weak budget-balanced mechanisms, and our positive results hold with respect to the more challenging benchmark of the entire GFT, including the gains kept by the mechanism.} of at most 2. This is a negligible fraction of the maximal GFT of $\Theta(\log \constantH)$ (obtainable by the mechanism that always trades at price 0, which is very unfair to the seller).

We conclude that the fairness notion of ex ante equitability is too stringent for the setting of trade where the buyer has private information: imposing ex ante equitability in settings where the buyer and seller differ significantly from each other (ex ante) results in devastating implications on the GFT. Therefore, we seek an alternative fairness notion that is better suited for trade settings where the buyer is ex ante very different from the seller and the traders have private information, yet allows for a good GFT approximation.

In spirit, this problem is a cooperative bargaining problem, where the traders need to ex ante agree on a truthful mechanism that is fair. If they fail to reach an agreement, the default outcome of no trade occurs. Yet, the set over which the traders are bargaining is not explicitly given, but rather induced by the traders' valuation distributions. Moreover, in bilateral-trade settings, where the seller also has a non-trivial valuation distribution, it is known that some of these mechanisms (e.g., the one that maximizes GFT subject to truthfulness) are very complicated and unintuitive, even for simple distributions. 

Two prominent solutions to cooperative bargaining problems are the Kalai-Smorodinsky (KS) solution \citep{KS-75}, and the Nash solution \citep{nash-51}. We first ask:
\begin{displayquote}
\emph{How large is the fraction of GFT guaranteed by the Kalai-Smorodinsky solution? By the Nash solution?} 
\end{displayquote}
Our work mainly focuses on studying the GFT of mechanisms that are \emph{\ksfair}, satisfying the ``Kalai-Smorodinsky condition'': mechanisms that ex ante equalize the fraction of the ideal utilities of the two traders. 
We explain {\ksfairness} using the following simple example: Consider the zero-value seller setting where the buyer's value is uniformly distributed over $[0, 1]$. The buyer's ideal ex ante utility is $0.5$, obtained by trading at a price of $0$, while the seller's ideal ex ante utility is $0.25$, obtained by trading at the monopoly reserve (i.e., Myerson price) of $0.5$. If the two traders trade at a price of $0.2$, the seller receives an expected utility of $0.16$, while the buyer receives $0.32$. Since both utilities achieve the same fraction (64\%) of their respective optima, this mechanism is {\ksfair}. Moreover, its GFT is $0.16+0.32 = 0.48$, which is 96\% fraction of the optimal GFT of $0.5$, achieved by trading at price of $0$ (which is unfair to the seller).

Although the Kalai-Smorodinsky solution satisfies {\ksfairness}, the mechanism corresponding to that solution may be complex, making both theoretical analysis and practical implementation challenging. This motivates us to design \emph{simple} mechanisms that are {\ksfair} and guarantee good GFT. Therefore, we pose the following question:
\begin{displayquote}
\emph{How large is the fraction of the {\SecondBest} that can be guaranteed by a simple mechanism satisfying {\ksfairness} (aka., the Kalai-Smorodinsky condition)?} 
\end{displayquote}
As we will explain in detail in the following section, our work presents simple mechanisms that are {\ksfair} and give nearly the best fraction of the {\SecondBest} we can hope for {from any {\ksfair} mechanism} (and split all the gains between the two traders). Furthermore, while these mechanisms may not always be Pareto-optimal (with the Pareto-optimal solution being the Kalai-Smorodinsky solution), they imply that the KS solution is also guaranteed to obtain {at least} the same GFT approximation.

\subsection{Our Contributions} 
\label{sec:intro-our}
\input{Paper/contribution-and-techniques}

\subsection{Related Work}
\label{subsec:related work}
\input{Paper/related-work}

%% file: Paper/contribution-and-techniques.tex
\label{subsec:contribution}
In this work, we study the efficiency of fair and truthful trade mechanisms. Below, we present
an overview of our contributions.

We focus on direct-revelation mechanisms, which specify an allocation (possibly randomized) and payments for each trader, for every valuation profile of the traders.\footnote{This is without loss of generality, by the revelation principle \citep{mye-81}.}
We restrict our attention to mechanisms that are interim individually rational (IIR), Bayesian incentive compatible (BIC) and ex-ante weak budget balance (ex-ante WBB),
\footnote{\label{footnote:reference to prelim}See \Cref{sec:prelim} for formal definitions.} and refer to such mechanisms as \emph{truthful mechanisms}.   
For general bilateral trade instances where both traders have private values sampled from overlapping distributions, the optimal expected GFT (also known as the {\FirstBest}) may not be achievable \citep{MS-83}. Thus, to understand the impact of fairness on truthful mechanisms, we analyze the GFT approximation of our proposed mechanism with respect to the {\SecondBest}, defined as the maximum GFT achievable by any truthful mechanism.\footnote{For zero-value seller instance, the {\FirstBest} and the {\SecondBest} are clearly equal. For general bilateral trade instance, it is known that the {\SecondBest} is at least ${1}/{3.15}$ fraction of the {\FirstBest} \citep{DMSW-22,Fei-22}.}
Truthful mechanisms that maximize the GFT might be complicated, even for simple distributions \citep{MS-83}. 
In contrast, all truthful mechanisms proposed in this work not only satisfy {\ksfairness} and achieve good GFT approximation to the {\SecondBest}, but they are also simple and easily implementable.

\xhdr{Optimal GFT approximation under {\ksfairness}.} 
As the first result of this work, we establish that the optimal GFT approximation of truthful mechanisms under {\ksfairness} is 50\%.

\begin{informal}[\Cref{thm:optimal GFT:general instance} and \Cref{lem:optimal GFT upper bound:irregular}]
\label{infmthm:general}
For every bilateral trade instance (i.e., any pair of seller and buyer distributions), there exists a truthful mechanism that is {\ksfair} and guarantees a GFT of at least 50\% of the {\SecondBest}.

Moreover, for any $\calC > 50\%$, there exists a zero-value seller instance in which no {\ksfair} truthful mechanism can achieve a 
$\calC$-fraction of the {\SecondBest}. 
\end{informal}
To obtain the positive approximation result, we develop a black-box reduction (\Cref{thm:blackbox reduction}) that converts any mechanism (possibly not {\ksfair}) into a {\ksfair} mechanism whose GFT is at least $\calC$-fraction of the sum of the traders’ ideal utilities, where $\calC$ is the smaller ratio between each trader's ex ante utility in the original mechanism, and her own ideal utility.
Our black-box reduction framework is both simple, general, and thus might be of independent interest.\footnote{In \Cref{appendix:bargaining-and-trade}, we generalize this framework to the cooperative bargaining problem.}
 
We apply our black-box reduction to analyze the {\BiasedRandomOffer}. 
This mechanism, for a given parameter $\mixprob\in[0,1]$, runs 
the {\SellerOffer} with probability $\mixprob$, and the {\BuyerOffer}\textsuperscript{\ref{footnote:reference to prelim}} with probability $1 - \mixprob$.
We show that, with an appropriately chosen 
$\mixprob$, this mechanism guarantees at least 50\% of the {\SecondBest} and is {\ksfair}. Notably, the proposed {\BiasedRandomOffer} is both simple and easy to implement. 
In particular, this mechanism is ex post IR and ex post strong budget balance (SBB), ensuring that all gains of the trade are split between the buyer and seller, leaving nothing to the mechanism. The (unbiased) {\RandomOffer}, which sets $\mixprob = 0.5$, already achieves a $\frac{1}{2}$-approximation to the {\SecondBest} \citep{BCWZ-17}. However, the {\RandomOffer} is generally not {\ksfair}. Our result shows that by carefully selecting $\mixprob$, we can preserve the same GFT approximation ratio while ensuring {\ksfairness}. 

To complete the picture, we also show that this GFT approximation ratio of $\frac{1}{2}$ is optimal among all {\ksfair} truthful mechanisms. To {prove} 
this, we present an explicit construction of an example
(\Cref{example:all fair mech:irregular}) with
zero-value seller and a buyer with value sampled from a distribution that we have carefully constructed to obtain
this tight bound.
Notably, this buyer distribution does not satisfy the regularity condition.
Regularity is a common assumption in the mechanism design literature \citep{mye-81,BR-89}, which holds for many classic distributions (e.g., Gaussian, exponential, uniform). In contrast to this bound of 50\%, as we have illustrated above, when the seller has {no value for the item} 
and the buyer's distribution is uniform between $[0, 1]$, posting a fixed trading price of $0.2$ is {\ksfair} and achieves 96\% of the {\SecondBest}. Motivated by this, we next consider settings where additional assumptions are imposed on the traders' distributions, and show that there are simple truthful mechanisms which are {\ksfair} and obtain much better GFT approximations. 

As a starting point, we study the bilateral trade instances where both traders' distributions satisfy the monotone hazard rate (MHR) condition.\footnote{The MHR condition, which is a {strengthening} of regularity, is also widely adopted in the mechanism design literature and satisfied by classic distributions such as exponential or uniform. See \Cref{sec:prelim} for the formal definition.} In this setting, we prove a stronger guarantee of $\frac{1}{e - 1} \geq 58.1\%$ for the {\ksfair} {\BiasedRandomOffer}, based on the results of \citep{Fei-22}. We next move to focus on the case of $0$-value seller, and prove stronger GFT approximation results.

\xhdr{Zero-value seller instances with regular or MHR distributions.}
In the second part of this work, we focus on the special case of the bilateral trade model where the seller has zero value for the item. This case is of particular interest as, while all gains can be realized by the simple and truthful mechanism which always trades at price $0$, such a mechanism is very unfair to the seller. On the other hand, letting the seller set the mechanism may result in arbitrarily small GFT (i.e., in the case of an equal-revenue distribution) or very unfair allocation (in the case of constant-value buyer).

In \Cref{infmthm:general} (\Cref{lem:optimal GFT upper bound:irregular}) we have shown that, even when the seller has no value for the item, the GFT approximation of $\frac{1}{2}$ cannot be improved when the buyer distribution is not regular. We next study the zero-value seller settings where the buyer's valuation distributions are regular or MHR.

We first consider the case that the buyer's valuation distribution is regular, and significantly improve the approximation of $50\%$ to more than $85\%$: 
\begin{informal}[\Cref{thm:improved GFT:regular buyer}]
\label{infmthm:regular buyer}
    For every zero-value seller instance where the buyer has a regular distribution, there exists a {\FixPrice} (which is truthful) that is {\ksfair}, and whose GFT is at least $\fixedPriceGFTPercentageRegular$ of the {\SecondBest}.

    Moreover, there exists a zero-value seller instance in which the buyer has a regular distribution and no {\ksfair} truthful mechanism obtains more than $\fixedPriceGFTPercentageUBRegular$ of the {\SecondBest}. 
\end{informal}

We remark that while the positive result is established by a {\FixPrice} (i.e., posting a trading price to two traders ex-ante),\footnote{{Interestingly, such a result cannot be obtained with a {\BiasedRandomOffer}. We show that} there exists a zero-value seller instance where the buyer has a regular distribution and yet, the {\ksfair} {\BiasedRandomOffer} only obtains 50\% of the {\SecondBest}, and not more.} the negative result holds for all {\ksfair} truthful mechanisms. Note that {\FixPrices} are the most simple mechanisms, and enjoy some excellent properties: not only they deterministic, they are also dominant strategy incentive compatible (DSIC), ex post IR and ex post SBB (so all GFT is split between the two traders, leaving nothing to the mechanism).

The almost-tight negative result of at most $\fixedPriceGFTPercentageUBRegular$ approximation is proven by presenting an explicit construction of an example (\Cref{example:BROM:regular}) with a zero-value seller and a buyer with value sampled from a regular distribution, and analyzing it (\Cref{lem:GFT UB:regular buyer}). For the positive result, we develop a novel argument based on a \emph{revenue curve reduction} analysis (\Cref{lem:GFT program:regular buyer}). As this analysis is our most significant technical contribution we discuss it in more details after presenting the rest of our results.

We next study the zero-value seller instance where the buyer's valuation distribution is MHR.  
\begin{informal}[\Cref{thm:improved GFT:mhr buyer}]
\label{infmthm:mhr buyer}
    For every zero-value seller instance where the buyer has an MHR distribution, there exists a {\FixPrice} (which is  truthful) that is {\ksfair}, and whose GFT is at least $\fixedPriceGFTPercentageMHR$ of the {\SecondBest}.

    Moreover, there exists a zero-value seller instance in which the buyer has an MHR distribution and no {\ksfair} truthful mechanism obtains more than $\fixedPriceGFTPercentageUBMHR$ of the {\SecondBest}. 
\end{informal}
The almost-tight negative result of at most $\fixedPriceGFTPercentageUBMHR$ approximation is proven by presenting an explicit construction of an example (\Cref{example:all fair:mhr buyer}) with a zero-value seller and a buyer with value sampled from an MHR distribution, and analyzing it (\Cref{lem:GFT UB:mhr buyer}). For the positive result, we conduct a similar argument (\Cref{lem:GFT program:mhr buyer}) as the one used for regular distributions, but now the argument centers on the cumulative hazard rate function instead of the revenue curve.

\xhdr{Implication for the KS-solution.}
While all of our results above (Informal Theorems~\labelcref{infmthm:general,infmthm:regular buyer,infmthm:mhr buyer}) are stated for {\ksfair} truthful mechanisms, the negative results trivially apply to the KS-solution as well, since it satisfies {\ksfairness}. Importantly, it is worth noting that all the positive results also hold for the KS-solution. This follows from the fact that the simple {\ksfair} mechanisms we proposed ({\BiasedRandomOffer} and {\FixPrice}) are ex post SBB, and therefore their GFT is at most the GFT of the KS-solution. See \Cref{lem:SBB implication} for the formal statement.\footnote{Though the KS-solution maximizes the sum of the two traders' utilities among all {\ksfair} mechanisms, it does not directly imply that the KS-solution also maximizes the GFT among all {\ksfair} mechanisms, since the GFT {not only includes the utilities of the two  traders, but also includes} the gains left to the mechanism (which could be positive under WBB).}

\xhdr{Implication for market regulation.} 
Our results also shed light on the following connection between efficiency and fairness regulation. Suppose the seller (resp.\ buyer) is a monopolist and can freely decide on the truthful mechanism to maximize her own utility. In this case, the GFT from the mechanism picked by the monopolist can be arbitrarily smaller than the {\SecondBest}. In contrast, consider an alternative scenario where a regulator imposes a regulation so that the monopolist may only choose a truthful mechanism that is {\ksfair}. In \Cref{lem:SBB implication}, we show that the seller-optimal (resp.\ buyer-optimal) {\ksfair} mechanism achieves the same GFT approximation as the ones stated in Informal Theorems~\labelcref{infmthm:general,infmthm:regular buyer,infmthm:mhr buyer}, so high GFT is guaranteed when fairness is imposed.

\xhdr{Alternative fairness definitions.} Besides {\ksfairness}, we also explore alternative fairness definitions for the bilateral trade model.

We first explore another solution concept for the bargaining problem -- the Nash solution \citep{nash-51}. In the context of the bilateral trade model, the Nash solution corresponds to a {\NashSocialWelfareMaximizer}: a truthful mechanism that maximizes the Nash social welfare (NSW), i.e., the product of the two traders' ex ante utilities. Following an argument that is conceptually similar to our black-box reduction framework (for {\ksfair} mechanisms), we obtain a tight bound on the GFT approximation of any {\NashSocialWelfareMaximizer}.

\begin{informal}
    For every bilateral trade instance, a {\NashSocialWelfareMaximizer} guarantees a GFT of at least 50\% of the {\SecondBest}.

    Moreover, for any $\calC > 50\%$, there exists a zero-value seller instance in which no {\NashSocialWelfareMaximizer} can achieve a
    $\calC$-fraction of the {\SecondBest}.
\end{informal}
While our definition of {\ksfairness} aims to explicitly define fairness for bilateral trade, the  fairness properties of the Nash solution are rather implicit (it is more about subscribing a way to trade-off the traders' utilities). We also remark that a {\NashSocialWelfareMaximizer} could have complicated allocation and payment rules. In contrast, both the {\BiasedRandomOffer} and {\FixPrice} which we proposed and analyzed for {\ksfairness} are simple and easy to be implemented. Hence, we believe {\ksfairness} may be more suitable for the bilateral trade problem, and our paper mainly focuses on it.

As we mentioned earlier, we also establish negative results showing that equitability (motivated by the egalitarian solution \citep{Kal-77,Mye-77} to the bargaining problem) and interim or ex post {\ksfairness} (\Cref{def:interim ks fairness,def:ex post ks fairness}) may not be appropriate. Specifically, we show that the GFT of truthful mechanisms that satisfy any of those alternative fairness definitions, can be arbitrary lower than the {\SecondBest} (\Cref{lem:equitable GFT UB}) or imply no trade and thus zero GFT (\Cref{prop:interim/ex post fairness:no trade}), even in settings with a zero-value seller.

\subsection{Our Techniques} 
\label{subsec:intro:techniques}
We next describe the technical framework we put forward for proving our positive results for instances with a zero-value seller and a buyer with either a regular or an MHR valuation distribution (Informal Theorems~\labelcref{infmthm:regular buyer,infmthm:mhr buyer}, respectively), {and the novelty of our technique}. We start by explaining the high-level proof idea behind \Cref{infmthm:regular buyer}. To obtain the almost-tight positive result of at least $\fixedPriceGFTPercentageRegular$, we aim to directly compare a {\ksfair} {\FixPrice} with the {\SecondBest}.\footnote{As a comparison, in the analysis of the positive result in \Cref{infmthm:general}, we compare the GFT of {\ksfair} mechanisms with the summation of two traders' ideal utilities (which upper bounds the {\SecondBest}). Consider a simple example, where the seller and buyer have deterministic value of zero and one, respectively. In this example, the summation of two traders' ideal utilities could  be twice the {\FirstBest} (and thus at least twice the GFT of any truthful mechanism)}. Hence, by comparing with the summation of two traders' ideal utilities, it is impossible to obtain an approximation ratio strictly better than $\frac{1}{2}$. To do this, we develop a novel revenue curve reduction analysis. 

In the mechanism design literature, a revenue-curve-based analysis is a common approach to deriving approximation guarantees \citep[e.g.,][]{FILS-15,DRY-15,AB-18,AHNPY-18,JLQTX-19,BCW-22,FHL-21,ABB-22,JL-23}. Revenue curves \citep{BR-89} give an equivalent representation of the buyer's valuation distribution and enable clean characterization of the revenue of a given mechanism. Though there does not exist an automatic way to conduct a revenue-curve-based analysis, the common high-level idea is to argue that the worst approximation ratio among the class of regular distributions belongs to a subclass of distributions that can be characterized with a small set of parameters (e.g., monopoly reserve, quantile, or revenue). Most prior work focuses on revenue approximation and allows consideration of a class of mechanisms (e.g., anonymous pricing with all possible prices). In this way, the approximation ratio usually depends on a small set of parameters of the revenue curve, and thus the reduction argument can be established.

In contrast, our analysis requires studying the seller's ex-ante utility (revenue), the buyer's ex-ante utility (residual surplus) and the GFT (expected value) -- all three of these \emph{together}. Furthermore, we focus on a single mechanism ({\ksfair} {\FixPrice}). Hence, both the GFT benchmark and the GFT of our mechanism are highly sensitive to the \emph{entire} revenue curve {(e.g., parameters like the monopoly reserve and revenue are not enough)}.

To overcome this challenge, we conduct a three-step argument (see \Cref{lem:GFT program:regular buyer}). We first consider a {\FixPrice} (that is possibly not {\ksfair})  whose trading price $\price$ is selected based on a \emph{small set} of parameters of the revenue curve. 
In the second step, we argue that depending on buyer's ex-ante utility under trading price $\price$, we can modify this price to obtain a {\ksfair} trading price $\fprice$. By comparing the two traders' utility changes from trading price $\price$ to $\fprice$, we obtain a lower bound of the GFT approximation under trading price $\fprice$, which is {\ksfair}. Importantly, this lower bound is a {function of the trading price $\price$ chosen in the first step}. Loosely speaking, using the first and second step, we show that the worst GFT approximation lower bound belongs to a \emph{subclass of regular distributions} parameterized by trading price $\price$ and several other parameters including monopoly quantile and monopoly revenue. Finally, since this GFT lower bound depends on the choice of $\price$ determined in the first step, we formulate the problem of selecting the best price~$\price$ (to maximize the GFT approximation lower bound) as a \emph{two-player zero-sum game}. In this game, the ``min player'' (adversary) selects the revenue curve (represented by finitely many quantities on the revenue curve), while the ``max player'' (ourselves) selects the price $\price$. The game's payoff is the GFT approximation lower bound derived in the second step. Finally, we capture this two-player zero-sum game as an \emph{optimization program~\ref{program:GFT:regular buyer}} {(a minimization problem)} and numerically lower bound the optimal payoff of this game (i.e., optimal objective value of \ref{program:GFT:regular buyer}) by $\fixedPriceGFTPercentageRegular$.\footnote{Formulating the approximation guarantees as optimization program and then solve it numerically is common in the mechanism design literature. See the related work section (\Cref{subsec:related work}) for more discussion.} Also see \Cref{fig:GFT program:regular buyer} for a graphical illustration. 

The positive result of \Cref{infmthm:mhr buyer} relies on a similar reduction argument. However, instead of conducting the reduction on the revenue curve, we analyze the \emph{cumulative hazard rate function}, which is another equivalent representation of the buyer's valuation distribution. In particular, the buyer's valuation distribution is MHR if and only if the induced cumulative hazard rate function is convex. In our analysis (\Cref{lem:GFT program:mhr buyer} and \Cref{fig:GFT program:mhr buyer}), we formulate the GFT approximation of a {\ksfair} {\FixPrice} as a %n
minimization program~\ref{program:GFT:mhr buyer}. We then numerically lower bound its optimal objective value by $\fixedPriceGFTPercentageMHR$.

%% file: Paper/related-work.tex
Our work connects to several strands of existing literature. 

\xhdr{GFT-optimization in the bilateral trade model.} 
\citet{MS-83} introduce and study Bayesian mechanism design for bilateral trade instances. The authors present the characterization of the GFT-optimal truthful mechanism (achieving the {\SecondBest}) and show the strict gap between the First Best and {\SecondBest}. \citet{DMSW-22} shows that the multiplicative gap between the two benchmarks is at most $1/8.23$, by analyzing the GFT approximation of the {\RandomOffer} against the {\FirstBest}. A follow-up work by \citet{Fei-22} further improves the multiplicative gap to $1/3.15$. The author also shows that the approximation of the {\RandomOffer} is at most $1/(e - 1)$ against the {\FirstBest} when both traders have MHR valuation distributions. Since the GFT-optimal truthful mechanism is complicated, various simple mechanisms have been proposed and been proven to achieve good GFT approximation under various conditions \citep[e.g.,][]{mca-92,BNP-09,BM-16,DTR-17,BCWZ-17,CGKLT-17,BCGZ-18,BGG-20,CLMZ-24,HHPS-25}. Our work differs from all those prior work as we study GFT-optimization subject to fairness conditions motivated by the cooperative bargaining problem.

Besides the Bayesian mechanism design model, there is also a recent literature initiated by \citet{CCCFL-24a} that includes \citep{BCCC-24,BCCF-24,CCCFL-24b,AFF-24} which studies
GFT-optimization in an online learning model, where traders' valuation distributions are unknown and need to be learned from repeated interactions. Instead of analyzing the approximation ratios, most prior work uses the optimal-in-hindsight {\FixPrice} (whose GFT could be arbitrarily smaller than the {\SecondBest}) as the benchmark   and focuses on designing sublinear regret online learning algorithm against this weaker benchmark. \citet{BCCC-24} introduce the problem of fair online bilateral trade and propose a new efficiency measure ``fair GFT'' defined as the minimum utility between two traders. Unlike their approach, our work sticks to the classic GFT as the efficiency measure and incorporates fairness as a constraint, motivated by the cooperative bargaining problem. 

\xhdr{Cooperative bargaining with incomplete information.} \label{subsubsec:related-bargaining}
The cooperative bargaining problem was introduced by \citet{N-50}, who also proposed the Nash solution as the unique outcome satisfying four axioms: Pareto optimality, symmetry, scale invariance, and independence of irrelevant alternatives. In the same model, \citet{KS-75} replaced the independence of irrelevant alternatives axiom with the resource monotonicity axiom and introduced the Kalai-Smorodinsky (KS) solution. Meanwhile, the egalitarian solution, which substitutes the scale invariance axiom with the resource monotonicity axiom, was developed by \citet{Kal-77,Mye-77}. For further details, we refer the reader to \citet{AH-92}. There is also a literature about bargaining with private information \citep[e.g.,][]{HS-72,Mye-84,Sam-84,KW-93,ACD-02}, where analog axioms (similar to the original bargaining problem with public information) and solutions (mechanisms) are introduced and studied. While these works focus on axiomatization of the solution concepts for general bargaining problems, we focus on bilateral trade problems and are interested in bounding the GFT approximation of truthful mechanisms satisfying various fairness notions, against the {\SecondBest}.

Importantly, we remark that there is one difference between the bilateral trade problem and the bargaining problem. In the bilateral trade problem, the GFT includes both traders' utilities and the gains left to the mechanism (which may be positive under WBB). In contrast, the bargaining problem only reasons about the utilities of {the} two agents (traders). Hence, if we restrict {the space of mechanisms to only include} SBB mechanisms (i.e., no gain left to mechanisms), then the bilateral trade problem becomes a special case of the bargaining problem. \citet{BCWZ-17} shows that for GFT-optimization, it is without loss of generality to consider ex post SBB mechanisms. In particular, they develop a transformation that converts any ex ante WBB mechanism into an ex post SBB mechanism with the same allocation rule, by adjusting the payment rule for the two traders. Unfortunately, their transformation change the two traders' utilities {in a way that violates the fairness constraint,} and thus cannot be applied to our work.\footnote{It is interesting to study whether every {\ksfair} ex ante WBB mechanism also admits a {\ksfair} ex post SBB mechanism implementation with the same allocation rule. {This, for example, will imply that the {\ksfair} mechanism with highest GFT is the KS-solution.} {One natural attempt is to proportionally ex-ante split the ex-ante gains of the mechanism to two traders (and thus not affect incentive compatibility or participation).} However, this only ensures ex ante SBB, and there might be ex post positive transfer to the buyer, which is impractical and unnatural. We leave this challenge for future work.}

\xhdr{Approximation analysis via optimization programs.} 
In this work, we obtain almost-tight GFT-approximations in \Cref{thm:improved GFT:regular buyer,thm:improved GFT:mhr buyer} by first formulating the GFT approximation guarantees as optimization programs~\ref{program:GFT:regular buyer}, \ref{program:GFT:mhr buyer} and then numerically lower bounding their optimal objective value. Characterizing the approximation guarantees as optimization programs and then obtaining the final approximation ratios by numerical evaluation is a widely adopted approach in the approximate mechanism design literature, e.g., revenue-maximization for zero-value seller problem \citep{AHNPY-18}, prior-independent mechanism design \citep{HJL-19}, price of anarchy \citep{HTW-18}, mechanism design with samples \citep{FHL-21,ABB-22}. In particular, recent work \citep{CW-23,LRW-23} {studies} the welfare-maximization of the {\FixPrice} for the bilateral trade problem. As a comparison, we analyze the {\FixPrice} with an additional fairness constraint, which requires us to analyze the revenue, residual surplus, and GFT -- all three of these together.

%% file: Paper/prelim.tex
\label{sec:prelim}

In this section, we first present the bilateral trade model, and then review some important concepts in Bayesian mechanism design that will be used in our analysis. 

\subsection{Bilateral Trade Model}
There is a single seller and a single buyer, and we refer to them as traders (or agents). 
The seller holds one item, which is demanded by the buyer. 
{Trade will determine an allocation of $\alloc\in[0, 1]$ fraction of the item (or that entire item with probability $\alloc$) that will be traded and transferred to the buyer.} 
The value of each trader for the item is private information, known only to her. 

Specifically, we denote the private value of the buyer by $\val\in\reals_+$, and given an allocation $\alloc\in[0, 1]$ and the payment she pays (monetary transfer from her) $\price\in\reals_+$, her utility $\buyerutil(\alloc,\price)$ is defined to be $\buyerutil(\alloc,\price) = \val\cdot \alloc - \price$.
Similarly, we denote the private value of the seller by $\cost\in\reals_+$, and given an allocation $\alloc\in[0, 1]$ and payment she receives (monetary transfer to her) $\sellerprice\in \reals_+$,\footnote{As it is more convenient  to assume that the seller receives money (rather than paying negative amount), we use $\sellerprice$ to denote the payment the seller receives.} her utility $\sellerutil(\alloc, \sellerprice)$ is defined to be $\sellerutil(\alloc, \sellerprice) = \sellerprice - \cost \cdot \alloc$.

Finally, given allocation $\alloc\in[0, 1]$, the \emph{gains-from-trade} (\emph{GFT}) when there is trade between a seller with value $\cost$ and a buyer with value $\val$ is $(\val - \cost)\cdot \alloc$.

\xhdr{The Bayesian setting.} There is benevolent social planner (mechanism designer) whose goal is to design mechanisms to facilitate trade between the seller and buyer. Though both traders' values are private, we assume each trader's value is independently sampled from some prior distribution (the distributions are not necessarily identical), and the distributions are known to all parties (seller, buyer and social planner). We use $\buyerdist$ to denote the buyer's valuation distribution, and use $\sellerdist$ to denote the seller's valuation distribution, so a bilateral trade instance is a pair $\btinstance$. To simplify the presentation and analysis, we impose the following assumptions on the traders' valuation distributions.
\begin{itemize}
    \item 
    The buyer's valuation distribution $\buyerdist$ has a bounded support $\supp(\buyerdist)$, 
    so it holds that ${\supp(\buyerdist)} \subseteq [\lval, \hval]$ for some $\hval\geq \lval\geq 0$.    
    We assume that $\buyerdist$ has no atoms, except possibly at $\hval$. 
    The cumulative density function (denoted by $\buyercdf$) is left-continuous,\footnote{Following the convention in auction design literature, we define cumulative density function $\buyercdf(t)\triangleq \prob[\val\sim\buyercdf]{\val < t}$ for the buyer's valuation distribution, and $\sellercdf(t) \triangleq \prob[\cost\sim\sellerdist]{\cost \leq t}$ for the seller's valuation distribution.} is differentiable within the interior of the support with measure 1. We denote the corresponding probability density function by $\buyerpdf$.
    
    \item The seller's valuation distribution $\sellerdist$ has a bounded support $\supp(\sellerdist)$, 
    and thus ${\supp(\sellerdist)} \subseteq [\lcost, \hcost]$ for some $\hcost\geq \lcost\geq 0$.    
    We assume that $\sellerdist$ has no atoms, except possibly at $\lcost$. 
    The cumulative density function (denoted by $\sellercdf$) is right-continuous, is differentiable within the interior of the support with measure 1. We denote the corresponding probability density function by $\sellerpdf$.
\end{itemize}
Given a bilateral trade instance $\btinstance$, the maximum (optimal) expected GFT for the instance, 
$\OPTFB\btinstance$, is defined as follows  
\begin{align*}
    \OPTFB\btinstance \triangleq \expect[\val\sim\buyerdist,\cost\sim\sellerdist]{\plus{\val - \cost}}
\end{align*}
where operator $\plus{a} \triangleq \max\{0, a\}$.
We refer to $\OPTFB\btinstance$ as the {\sf First-Best (FB) GFT benchmark}.

\xhdr{Bayesian mechanism design for GFT maximization.} 
A (direct-revelation) mechanism $\mech = (\alloc,\price, \sellerprice)$ first solicits bids from both traders and then decides the allocation and payment based on allocation rule $\alloc$, buyer payment rule $\price$, and seller payment rule $\sellerprice$. Allocation $\alloc:\reals_+^2\rightarrow [0, 1]$ and buyer/seller payments $\price,\sellerprice:\reals_+^{2}\rightarrow\reals_+$ are mappings from the bid profile of the traders to the fraction of item allocated and monetary transfer from the buyer/to the seller, respectively. 

We are interested in the GFT obtainable in some Bayesian Nash equilibrium (BNE) of the designed mechanism. Applying the standard revelation principle \citep{mye-81}, without loss of generality, we focus on mechanisms that satisfy \emph{Bayesian Inventive Compatibility (BIC)}, i.e., reporting private value truthfully forms a BNE. By imposing BIC, we assume that traders report their private values truthfully and view the allocation and payment rules $(\alloc, \price, \sellerprice)$ of a given mechanism as mappings from traders' (true) valuation profile $(\val, \cost)$, to their allocation and payments, respectively. 

With slight abuse of notation, we further define the interim allocation, payment, and utility of a trader in a given mechanism $\mech = (\alloc, \price, \sellerprice)$ as the expected allocation, payment, and utility over the randomness of the other trader's value. Namely,
\begin{align*}
    \text{for the seller:}
    \quad
    &\selleralloc(\cost) = \expect[\val\sim\buyerdist]{\alloc(\val, \cost)},~
    \sellerprice(\cost) = \expect[\val\sim\buyerdist]{\sellerprice(\val, \cost)},~
    \sellerutil(\cost) = \expect[\val\sim\buyerdist]{\sellerprice(\val,\cost) - \cost \cdot \alloc(\val,\cost)}
    \\
    \text{for the buyer:}
    \quad
    &\alloc(\val) = \expect[\cost\sim\sellerdist]{\alloc(\val,\cost)},~
    \price(\val) = \expect[\cost\sim\sellerdist]{\price(\val, \cost)},~
    \buyerutil(\val) = \expect[\cost\sim\sellerdist]{\val\cdot \alloc(\val,\cost) - \price(\val,\cost)}
\end{align*}
Fix any bilateral trade instance $\btinstance$ and any mechanism $\mech = (\alloc,\price,\sellerprice)$, we define the ex ante seller (resp.\ buyer) utility $\sellerexanteutil(\mech,\buyerdist,\sellerdist)$ (resp.\ $\buyerexanteutil(\mech,\buyerdist,\sellerdist)$) as the expected utilities of the seller (resp.\ buyer) over the randomness of the other trader's value. Namely,
\begin{align*}
    \sellerexanteutil(\mech,\buyerdist,\sellerdist) = \expect[\cost\sim\sellerdist]{\sellerutil(\cost)}
    \;\;\mbox{and}\;\;
    \buyerexanteutil(\mech,\buyerdist,\sellerdist) =  \expect[\val\sim\buyerdist]{\buyerutil(\val)}
\end{align*}
Besides Bayesian incentive compatibility, any mechanism under consideration must also satisfy \emph{interim individual rationality} and \emph{ex ante weak budget balance}, formally defined as follows:
\begin{itemize}
    \item \emph{Interim individual rationality (IIR)}: For each trader and each value realization {of that trader}, her interim utility is non-negative.
    \item \emph{Ex ante weak budget balance (ex ante WBB)}: The expected payment collected from the buyer is weakly larger than the expected payment given to the seller, i.e., $\expect[\val\sim\buyerdist]{\price(\val)} \geq \expect[\cost\sim\sellerdist]{\sellerprice(\cost)}$.
\end{itemize}
For every bilateral instance $\btinstance$, we let $\mechfamily$ denote the family of all BIC, IIR  and ex ante WBB mechanisms for these distributions.

The social planner measures the efficiency of mechanisms by their GFT. Specifically, for a bilateral trade instance $\btinstance$ the GFT of mechanism $\mech = (\alloc,\price, \sellerprice)$ is defined as 
\begin{align*}
    \GFT{\mech,\buyerdist,\sellerdist} = \expect[\val\sim\buyerdist,\cost\sim\sellerdist]{(\val - \cost)\cdot \alloc(\val,\cost)}
\end{align*}
We use $\OPTSB$ to denote the maximum (optimal) expected GFT obtainable by any mechanism that is BIC, IIR and ex-ante WBB, i.e.,
\begin{align*}
    \OPTSB\btinstance \triangleq \max_{{\mech\in\mechfamily}}
    \GFT{\mech,\buyerdist,\sellerdist}
\end{align*} 
We also refer to $\OPTSB\btinstance$ as the {\sf Second-Best (SB) GFT benchmark}. It is known that for ``non-degenerate'' bilateral trade instances $\btinstance$, the {\FirstBest} $\OPTFB\btinstance$ is strictly larger than the {{\SecondBest}} $\OPTSB\btinstance$ \citep{MS-83}, but the latter is also a constant approximation to the former, i.e., $\OPTFB\btinstance \leq 3.15 \cdot \OPTSB\btinstance$ \citep{DMSW-22,Fei-22}. 

Throughout the paper, we mainly compare the GFT of a given mechanism $\mech\in\mechfamily$ with the second best benchmark $\OPTSB$. For any $\GFTapprox\in[0, 1]$, we say the GFT of a given mechanism ${\mech\in\mechfamily}$ is a $\GFTapprox$-approximation (to the {{\SecondBest}} $\OPTSB$) if $\GFT{\mech,\buyerdist,\sellerdist} \geq \GFTapprox \cdot \OPTSB\btinstance$.\footnote{Nonetheless, due to the bounded gap between two benchmarks, i.e., $\OPTFB\leq 3.15 \cdot \OPTSB$ for every bilateral trade instances, our results also imply GFT approximations against the {\FirstBest} $\OPTFB$.}

While the {\SecondBest} $\OPTSB$ is with the most relaxed properties (BIC, IIR, and ex ante WBB), all the positive results in this work are obtained with the stronger properties of \emph{ex post individual rationality (ex post IR)} where every trader's ex post utility is non-negative for every realized valuation profile, and \emph{ex post strong budget balance (ex post SBB)} where the buyer's ex post payment is equal to the seller's ex post payment for every realized valuation profile. Some of the positive results even replace BIC with \emph{dominate strategy incentive compatibility (DSIC)} where reporting value truthfully is a dominant strategy.

We make use of four classic mechanisms studied extensively in the bilateral trade model:
\begin{itemize}
    \item {\FixPrice} ({\FPM}): the social planner offers a price $\price$ to both traders ex ante, who individually chooses whether to accept it or not. If both traders accept the price, a trade occurs at price $\price$. Otherwise, no trade occurs and no payments are transferred.\footnote{{We always assume that ties are broken towards trade.}}
     \item {\SellerOffer} ({\SOM}): given her realized value $\cost$, the seller optimally picks a take-it-or-leave-it price $\price$ and offers to sell the item at that price.\footnote{Specifically, the seller picks a monopoly reserve $\optreserve \in\argmax_{\price}(\price -\cost)\cdot (1 - \buyercdf(\price))$ given her realized value $\cost$.} Trade occurs at price $\price$ if the buyer accepts the offer (Otherwise, no trade occurs and no payments are transferred).
    \item {\BuyerOffer} ({\BOM}): given her realized value $\val$, the buyer optimally picks a take-it-or-leave-it price $\price$ and offers to buy the item at that price. Trade occurs at price $\price$ if the seller accepts the offer (Otherwise, no trade occurs and no payments are transferred).
    \item {\RandomOffer} ({\ROM}): 
    the social planner implements the {\RandomOffer} and {\BuyerOffer} with identical ex ante probability $\frac{1}{2}$.
\end{itemize}
All four mechanisms are ex post IR and ex post SBB. The {\FixPrice} is DSIC, while the other three mechanisms are BIC.\footnote{In fact, 
in the {\SellerOffer} (resp.\ {\BuyerOffer}), the buyer (resp.\ seller) has a dominant strategy.} By \citet{mye-81}, the {\SellerOffer} (resp.\ {\BuyerOffer}) maximizes the seller's (resp.\ buyer's) ex ante utility among all BIC, IIR, and ex ante WBB mechanisms.

When traders' distributions $\buyerdist,\sellerdist$ are clear from the context, we sometimes simplify notations and omit them, e.g., writing $\GFT{\mech}$ instead of $\GFT{\mech,\buyerdist,\sellerdist}$. This includes the notations $\mechfam$, $\GFT{\mech}$, $\sellerexanteutil(\mech)$, $\buyerexanteutil(\mech)$, $\OPTFB$, and $\OPTSB$, among others.

\subsection{Necessary Mechanism Design Notations and Concepts}

We next establish some notations and standard properties that we use in the paper. 

\xhdr{Regularity and monotone hazard rate.}
Two important distribution subclasses have been introduced and commonly studied in the mechanism design literature \citep{mye-81,MS-83}.
\begin{definition}[Regularity and MHR for the buyer]
\label{def:regularity buyer}
    A buyer's valuation distribution $\buyerdist$ is \emph{regular} if its virtual value function 
    $
        \virtualval(\val) \triangleq \val -
        \frac{1 - \buyercdf(\val)}{\buyerpdf(\val)}
    $
    is non-decreasing in $\val$. 
    
    A buyer's valuation distribution $\buyerdist$ satisfies the \emph{monotone hazard rate (MHR) condition} if its hazard rate function 
    $
        \buyerhazardrate(\val) \triangleq
        \frac{\buyerpdf(\val)}{1 - \buyercdf(\val)}
        $
    is non-decreasing in $\val$.
\end{definition}
MHR is a stronger condition: any distribution that is MHR is also regular. Symmetrically, we say the seller's valuation distribution $\sellerdist$ is regular (resp.\ MHR) if the random variable $\hcost - \cost$ satisfies the regularity (resp.\ MHR) condition defined for the buyer in \Cref{def:regularity buyer}, where $\hcost$ is the largest value in support $\supp(\sellerdist)$. 

\xhdr{Revenue curve and cumulative hazard rate.}
We first present the revenue curve, which is useful for the revenue analysis.

\begin{definition}[Revenue curve]
    Fix a valuation distribution $\buyerdist$ of the buyer. The \emph{revenue curve} $\revcurve:[0, 1]\rightarrow\reals_+$ is a mapping from a quantile $\quant\in[0, 1]$ to the expected revenue from posting a price equal to $\val(\quant)= \buyercdf^{-1}(1 - \quant)  \triangleq \sup\{\val:\buyercdf(\val) \leq 1 - \quant\}$, the price for which sell happens with probability~$\quant$. Namely, for every quantile $\quant\in[0, 1]$, $\revcurve(\quant) \triangleq \quant\cdot \buyercdf^{-1}(1 - \quant)$.
\end{definition}
We define the \emph{monopoly revenue} as the largest revenue $\max_{\quant\in[0, 1]}\revcurve(\quant)$ and denote every quantile (resp.\ price) attaining this largest revenue as a \emph{monopoly quantile} $\optquant$ (resp.\ \emph{monopoly reserve $\optreserve$}).
The regularity of a valuation distribution has the following nice geometric interpretation on the induced revenue curve. 
\begin{lemma}[\citealp{BR-89}]
\label{lem:concave revenue curve}
    A valuation distribution $\buyerdist$ of the buyer is regular if and only if the induced revenue curve $\revcurve$ is weakly concave (i.e., $\revcurve \equiv \ironrevcurve$). Moreover, for every value $\val$, the virtual value $\virtualval(\val)$ is equal to the right derivative of revenue curve $\revcurve$ at quantile $\quant = 1  - \buyercdf(\val)$, i.e., $\virtualval(\val) = \partial_+\revcurve(\quant)$. 
\end{lemma} 

Similar to the connection between the virtual value and revenue curve, we also introduce the \emph{cumulative hazard rate function} $\cumhazard:\supp(\buyerdist)\rightarrow \reals_+$ where $\cumhazard(\val) = \int_{0}^\val \buyerhazardrate(t)\cdot \d t \equiv -\ln(1 - \buyercdf(\val))$. Note that buyer distribution $\buyerdist$ is MHR if and only if the induced cumulative hazard rate function $\cumhazard$ is weakly convex. 

The virtual value and hazard rate admit the following characterization on the buyer's ex ante utility and payment for the zero-value seller instances.
\begin{proposition}[{\citealp{mye-81,HR-08}}]
\label{prop:revenue equivalence}
\label{prop:buyer surplus equivalence}
\begin{flushleft}
{Fix any zero-value seller instance where the buyer has a valuation distribution $\buyerdist$. For every any BIC, IIR, ex ante WBB mechanism $\mech \in\mechfam$,} the buyer's ex ante utility and payment satisfy 
$\buyerexanteutil(\mech)
= 
\expect[\val\sim\buyerdist]{\buyerhazardrate(\val)\cdot \alloc(\val)}$
and 
$\expect[\val\sim\buyerdist]{\price(\val)}
= 
\expect[\val\sim\buyerdist]{\virtualval(\val)\cdot \alloc(\val)}$,
where $\alloc(\val)$ is the interim allocation of the buyer with value $\val$.
\end{flushleft}
\end{proposition}

%% file: Paper/ks-fairness.tex
In this section, we define {\ksfairness}, the main notion of fairness for bilateral trade that we study in this paper. Before presenting the formal definition, we introduce $\sellerbenchmark$ and $\buyerbenchmark$, which are the seller and buyer's ideal (ex ante)  utilities (aka., the seller and buyer benchmarks) for instance $\btinstance$, computed as\footnote{In defining the ideal utilities $\sellerbenchmark$, $\buyerbenchmark$, as well as in some other definitions later in the paper, we use $\max$ instead of $\sup$ because the set $\mechfam$ is compact, ensuring that the maximum is always attainable.}
\begin{align*}
{\sellerbenchmark= 
    \sellerbenchmark\btinstance} \triangleq \max\limits_{{\mech\in\mechfamily}}\sellerexanteutil(\mech) 
    \;\;
    \mbox{and}
    \;\; 
    {\buyerbenchmark = \buyerbenchmark\btinstance} \triangleq \max\limits_{\mech\in\mechfamily}\buyerexanteutil(\mech)
\end{align*}
As a {reminder}, for {every} bilateral trade instance $\btinstance$, the ideal utilities $\sellerbenchmark$ and $\buyerbenchmark$ are achievable by the {\SellerOffer} and {\BuyerOffer} \citep{mye-81}, respectively. For zero-value seller instances, the seller and buyer's ideal utilities $\sellerbenchmark$ and $\buyerbenchmark$ are achieved by posting a deterministic take-it-or-leave it price at the monopoly reserve and at zero, respectively \citep{mye-81}.

In the spirit of the Kalai-Smorodinsky (KS) solution to bargaining problems \citep{KS-75}, we study \emph{{\ksfairness}} for bilateral trade, defined as follows: 
\begin{definition}[{\ksfairness}]
\label{def:ks fairness} {Fix any bilateral trade instance $\btinstance$.} 
    A mechanism $\mech\in\mechfamily$ is \emph{{\ksfair}} if two players' ex ante utilities achieve the same fraction of
    each player's own ideal utility, i.e.,
    \begin{align*}
        \frac{\sellerexanteutil(\mech,\buyerdist,\sellerdist)}{\sellerbenchmark\btinstance}
        = \frac{\buyerexanteutil(\mech, \buyerdist,\sellerdist)}{\buyerbenchmark\btinstance}
    \end{align*}
\end{definition}

We remark that for every instance, there is at least one mechanism that is (trivially) {\ksfair}: the {\NoTrade} (one in which there is never trade) satisfies {\ksfairness}, since both players receive zero ex ante utility. However, its GFT is also zero, which leads to a zero-approximation to the {\SecondBest} $\OPTSB$ for any non-trivial instance.

In \Cref{appendix:bargaining-and-trade} we discuss a way to view the bilateral trade problem as a bargaining problem, where each mechanism corresponds to a point in the two dimensional plane of ex ante utilities obtained by the two traders. With that perspective, the Kalai-Smorodinsky solution to the bargaining problem that corresponds to the bilateral trade instance, is the {\ksfair} mechanism which is Pareto optimal (maximizes $\buyerexanteutil(\mech)+\sellerexanteutil(\mech)$). In \Cref{appendix:bargaining-and-trade} we also discuss other bargaining solutions (as the Nash solution and the Egalitarian solution) and their applicability to the bilateral trade problem.

Note that {\ksfairness} is defined on the \emph{ex ante} utilities of the traders. In \Cref{subsec:interim ks fairness} we justify the focus on ex ante utilities, by demonstrating that the corresponding interim and ex post fairness notions are too stringent, as both imply that trade must never happen.

For a given bilateral trade instance $\btinstance$, a {\ksfair} mechanism with the highest gains-from-trade $\mech$ is a solution to a linear optimization problem.\footnote{Note that for non-discrete distributions that program is infinite.} The variables of the linear optimization problem are  the {ex post} allocation probability, buyer payment and seller payment $\alloc(\val,\cost),\price(\val,\cost),\sellerprice(\val,\cost)$ for every $\val\in \supp(\buyerdist), \cost\in\supp(\sellerdist)$. It is well known that the requirements of BIC, IIR and ex ante WBB can all be enforced with linear constraints. Additionally, {\ksfairness} can also be enforced {as linear constraints}, as the ex ante utilities $\buyerexanteutil(\mech), \sellerexanteutil(\mech)$ of both traders are linear functions in the variables and the benchmark values $\buyerbenchmark, \sellerbenchmark$ can be computed prior to running this linear optimization problem. Finally, the program maximizes the (expected) GFT obtained by the mechanism, which is a linear objective in the variables.

Notice that maximizing the GFT is not the only option; by changing the objective, we can obtain a  {\ksfair} mechanism which maximizes any linear objective. For example:
\begin{enumerate}
    \item a seller-optimal (resp.\ buyer-optimal) {\ksfair} mechanism can be obtained by maximizing $\sellerexanteutil(\mech)$ (resp. $\buyerexanteutil(\mech)$).
     \item a KS-solution\footnote{\label{footnote:ref-to-ks-solution}See \Cref{subsec:ks-solution} for the formal definition of the KS-solution.} can be obtained by maximizing the objective $\sellerexanteutil(\mech)+\buyerexanteutil(\mech)$.%\footnote{\nmme{Note that this not equivalent to maximizing the GFT since the mechanism may be WBB, i.e. the buyer may pay more then the seller receives. See \Cref{subsec:ks-solution} for a deeper discussion. }}  %\nmmc{There may be more than one, with different GFTs} MB: this is true for the seller optimal as well. I have made some edits. 
\end{enumerate}
We remark that the requirement of {\ksfairness} imposes a constant ratio between $\buyerexanteutil(\mech)$ and $\sellerexanteutil(\mech)$. Thus, the KS-solution, the seller-optimal {\ksfair} mechanism, and the buyer-optimal {\ksfair} mechanism, all have the same seller ex ante utility (and also the same buyer ex ante utility).

A significant issue is that the solution to the linear optimization problem might be a very complex mechanism, that is hard to describe, and  thus not practical. Another issue is that it solves a particular instance and does not give us any insights about the guarantees on the approximation to the {\SecondBest} achieved by such mechanisms for particular sets of instances that are of special interest (as when the buyer distribution is regular). Throughout the paper we overcome both issues by presenting {\ksfair} mechanisms which are simple (ex post SBB, ex post IIR, and sometimes even DSIC) and yet achieve some good approximation to the second best GFT. The following lemma formalizes how any such result implies a lower bound for the fraction of the second best GFT achievable by the above linear optimization problems. 

\begin{restatable}{lemma}{lemsbbimplication}
\label{lem:SBB implication}
    Fix any bilateral trade instance $\btinstance$, if there exists a mechanism $\mech\in \mechfamily$ that is ex ante SBB, {\ksfair}, and guarantees a GFT of at least $\calC$ fraction of the {\SecondBest} $\OPTSB$, then each of the following mechanisms also obtains a GFT that is at least $\calC$ fraction of the {\SecondBest}:
\begin{enumerate}
    \item any GFT-maximizing {\ksfair} mechanism $\mech\primed\in\mechfamily$,
    \item {any mechanism $\mech\doubleprimed\in\mechfamily$ that is a} KS-solution\textsuperscript{\ref{footnote:ref-to-ks-solution}}{(and thus also any seller-optimal {\ksfair} mechanism, and any buyer-optimal {\ksfair} mechanism)}. 
\end{enumerate}
\end{restatable}
\begin{proof}
    We show the GFT approximation for each mechanism separately. For a GFT-maximizing {\ksfair} mechanism $\mech\primed$, its GFT is at least the GFT of mechanism $\mech$ by definition, and thus is a $\calC$-approximation {to the {\SecondBest}}. 
    
    For a KS-solution mechanism $\mech\doubleprimed$, since it is also a seller-optimal {\ksfair} mechanism, the seller's ex ante utility satisfies $\sellerexanteutil(
    {\mech\doubleprimed}) \geq \sellerexanteutil(\mech)$. Combining with the fact that both mechanisms $\mech$ and $\mech\doubleprimed$ are {\ksfair}, we also obtain $\buyerexanteutil(
    {\mech\doubleprimed}) \geq \buyerexanteutil(\mech)$. Thus,
    \begin{align*}
        \GFT{
        {\mech\doubleprimed}} \geq \buyerexanteutil(
        {\mech\doubleprimed}) + \sellerexanteutil(
        {\mech\doubleprimed}) \geq \buyerexanteutil(\mech) + \sellerexanteutil(\mech)
        = \GFT{\mech} \geq \calC \cdot \OPTSB
    \end{align*}
    where the first inequality holds since mechanism $\mech\doubleprimed$ is ex ante WBB, and the equality holds since mechanism $\mech$ is ex ante SBB. 
\end{proof}

As all our positive results prove the existence of ex post SBB {(and thus ex ante SBB as well)} mechanisms $\mech\in \mechfamily$ that {are} {\ksfair} and achieve a good approximation to the {\SecondBest} $\OPTSB$, the lemma applies to all our positive results (\Cref{thm:optimal GFT:general instance,thm:improved GFT:regular buyer,thm:improved GFT:mhr buyer} and \Cref{thm:improved GFT:mhr traders}). Moreover, since all these results also include upper bounds on the fraction of the second best GFT achievable by \emph{any} {\ksfair} truthful mechanism in different settings, these upper bounds also hold for each of the above mechanisms.

%% file: Paper/black-box-reduction.tex
In this section, we show that for general bilateral instances, there exists a simple mechanism (the {\BiasedRandomOffer}) that is {\ksfair} and has GFT approximation of $\frac{1}{2}$, and that fraction can be improved to $\frac{1}{e-1}$ when the distributions are MHR. The {\BiasedRandomOffer} that we use is BIC, ex post IR, and ex post SBB. 
\begin{definition}[{\BROM}]
    \label{def:biased random offer}
    In the {\BiasedRandomOffer}, the social planner {fixes} a probability $\mixprob\in[0, 1]$ {ex ante}. Then, the {\SellerOffer} is implemented with probability $\mixprob$, and the {\BuyerOffer} is implemented with probability $1 - \mixprob$.\footnote{As a sanity check, with $\mixprob= 0, 0.5, 1$, the {\BiasedRandomOffer} recovers the classic {\BuyerOffer}, (unbiased) {\RandomOffer}, and {\SellerOffer}, respectively.}
\end{definition}
With start for our main result in this section, the results for general bilateral instances.

\begin{theorem}
\label{thm:optimal GFT:general instance}
\label{cor:biased random offer}
    For every bilateral trade instance $\btinstance$, there exists probability $\mixprob\in[0, 1]$ such that the {\BiasedRandomOffer} is {\ksfair} and its GFT is at least at least $\frac{1}{2}$ fraction of the {\SecondBest} $\OPTSB$. 
\end{theorem}
We remark that GFT approximation of $\frac{1}{2}$ in the above theorem does not require any assumption on the traders' valuation distributions. Moreover, this GFT approximation is \emph{optimal}: we show in \Cref{sec:improved GFT} that there exists a zero-value seller instance (\Cref{example:all fair mech:irregular}) where every BIC, IIR, ex ante WBB mechanism that is {\ksfair}, cannot obtain more than half of the {\SecondBest}.

To prove \Cref{thm:optimal GFT:general instance}, we develop a black-box reduction framework: it converts an arbitrary mechanism $\mech$ into a {\ksfair} mechanism $\mech\primed$ with provable GFT approximation guarantee, which is equal to the smaller ratio between each trader's ex ante utility in mechanism $\mech$ with her own benchmark. We then apply this framework on the (unbiased) {\RandomOffer} to prove the GFT approximation of $\frac{1}{2}$ 
as claimed.

\begin{restatable}[Black-box reduction]{theorem}{thmblackboxreduction}
    \label{thm:blackbox reduction}
    {Fix any bilateral trade instance $\btinstance$.}
    Fix any mechanism $\mech$ {that is BIC, IIR and ex ante WBB} (possibly not {\ksfair}). 
    Define the constant $\GFTapprox\in[0,1]$ as
    \begin{align*}
        \GFTapprox \triangleq \min\left\{
        \frac{\sellerexanteutil(\mech)}{\sellerbenchmark},
        \frac{\buyerexanteutil(\mech)}{\buyerbenchmark}
        \right\}
    \end{align*}
    Then, there exists BIC, IIR and ex ante WBB mechanism $\mech\primed$ that is {\ksfair} and guarantees a GFT of at least $\GFTapprox$ fraction of the {\SecondBest} $\OPTSB$. Specifically, mechanism $\mech\primed$ is constructed as follows:
    \begin{itemize}
        \item If $\frac{\buyerexanteutil(\mech)}{\buyerbenchmark} \geq \frac{\sellerexanteutil(\mech)}{\sellerbenchmark}$, mechanism $\mech$ is implemented with probability $\mixprob$, and the {\SellerOffer} is implemented with probability $1-\mixprob$, for some  $\mixprob\in[0, 1]$.
        \item If $\frac{\buyerexanteutil(\mech)}{\buyerbenchmark} \leq \frac{\sellerexanteutil(\mech)}{\sellerbenchmark}$, mechanism $\mech$ is implemented with probability $\mixprob$, and the {\BuyerOffer} is implemented with probability $1-\mixprob$, for some  $\mixprob\in[0, 1]$.
    \end{itemize}
\end{restatable}
We remark that if the original mechanism $\mech$ is ex post IR (resp.\ ex post SBB), the constructed mechanism $\mech\primed$ from \Cref{thm:blackbox reduction} is also ex post IR (resp.\ ex post SBB). Moreover, our black-box reduction (\Cref{thm:blackbox reduction}) and the positive result from \Cref{thm:optimal GFT:general instance} can be extrapolated to obtain similar results for the more abstract model of cooperative bargaining, with an almost identical proof. In \Cref{appendix:general-bargaining-results} we detail the formulations of these results in the bargaining model. Therefore, though the black-box reduction framework is conceptually simple, it is general and could be of independent interest. 
\begin{proof}[Proof of \Cref{thm:blackbox reduction}]
    Without loss of generality, we assume $\frac{\buyerexanteutil(\mech)}{\buyerbenchmark} \geq \frac{\sellerexanteutil(\mech)}{\sellerbenchmark} = %\geq
    \GFTapprox$. (The other case follows a symmetric argument.)
    For every $\mixprob\in[0, 1]$, consider $\mech\primed_{\mixprob}$ constructed as follows: original mechanism $\mech$ is implemented with probability $\mixprob$ and the {\SellerOffer} is implemented with probability $1-\mixprob$. By construction, we have
    \begin{align*}
        \frac{\sellerexanteutil(\mech\primed_\mixprob)}{\sellerbenchmark} 
        =
        \frac{\mixprob\cdot \sellerexanteutil(\mech) + (1-\mixprob)\cdot \sellerexanteutil(\SOM)}{\sellerbenchmark} 
        =
        \frac{\mixprob\cdot \sellerexanteutil(\mech) + (1-\mixprob)\cdot \sellerbenchmark}{\sellerbenchmark} 
    \end{align*}
    which is weakly decreasing linearly in $\mixprob\in[0, 1]$, since $\sellerbenchmark(\mech) \leq \sellerbenchmark$ by definition. Similarly, 
    \begin{align*}
        \frac{\buyerexanteutil(\mech\primed_\mixprob)}{\buyerbenchmark} 
        =
        \frac{\mixprob\cdot \buyerexanteutil(\mech) + (1-\mixprob)\cdot \buyerexanteutil(\SOM)}{\buyerbenchmark} 
    \end{align*}
    which is also linear in $\mixprob\in[0, 1]$. Moreover, due to the case assumption, we know 
    \begin{align*}
        \frac{\sellerexanteutil(\mech\primed_1)}{\sellerbenchmark} 
        \leq 
        \frac{\buyerexanteutil(\mech\primed_1)}{\buyerbenchmark}
        \;\;
        \mbox{and}
        \;\;
        \frac{\sellerexanteutil(\mech\primed_0)}{\sellerbenchmark} = 1
        \geq 
        \frac{\buyerexanteutil(\mech\primed_0)}{\buyerbenchmark}
    \end{align*}
    Thus, there exists $\mixprob^*\in[0, 1]$ such that 
    \begin{align*}
        \frac{\buyerexanteutil(\mech\primed_{\mixprob^*})}{\buyerbenchmark} 
        \overset{(a)}{=} 
        \frac{\sellerexanteutil(\mech\primed_{\mixprob^*})}{\sellerbenchmark} 
        \overset{(b)}{\geq}
        \frac{\sellerexanteutil(\mech\primed_{1})}{\sellerbenchmark} 
        \overset{(c)}{\geq}
        \GFTapprox
    \end{align*}
    which implies that mechanism $\mech\primed_{\mixprob^*}$ is {\ksfair} and its GFT is a $\GFTapprox$-approximation to the {\SecondBest} $\OPTSB$. (Recall that the {\SecondBest} $\OPTSB$ is upper bounded by $\sellerbenchmark + \buyerbenchmark$.) Here equality~(a) holds due to the intermediate value theorem, inequality~(b) holds due to the monotonicity of $
    {\sellerexanteutil(\mech\primed_{\mixprob})}/{\sellerbenchmark}$ as a function of $\mixprob$ argued above,
    and 
    inequality~(c) holds since mechanism $\mech\primed_1$ is equivalent to original mechanism $\mech$.
\end{proof}

Utilizing this black-box reduction,  we immediately prove \Cref{thm:optimal GFT:general instance} as follows. 
\begin{proof}[Proof of \Cref{thm:optimal GFT:general instance}]
    Consider the (unbiased) {\RandomOffer}. Note that both traders' ex ante utilities satisfy
    \begin{align*}
        \frac{\buyerexanteutil(\ROM)}{\buyerbenchmark}
        =
        \frac{1}{2} \left(\frac{\buyerexanteutil(\SOM)}{\buyerbenchmark}
        +
        \frac{\buyerexanteutil(\BOM)}{\buyerbenchmark}\right)
        \geq 
        \frac{1}{2}
        \;
        \mbox{and}
        \;
        \frac{\sellerexanteutil(\ROM)}{\sellerbenchmark}
        =
        \frac{1}{2} \left(\frac{\sellerexanteutil(\SOM)}{\sellerbenchmark}
        +\frac{\sellerexanteutil(\BOM)}{\sellerbenchmark}
        \right)
        \geq 
        \frac{1}{2}
    \end{align*}    
    where we use the fact that the buyer's (resp.\ seller's) benchmark $\buyerbenchmark$ ($\sellerbenchmark$) is achieved by the {\BuyerOffer} (resp.\ {\SellerOffer}), i.e., $\buyerexanteutil(\BOM) = \buyerbenchmark$ (resp.\ $\sellerexanteutil(\SOM) = \sellerbenchmark$). Invoking \Cref{thm:blackbox reduction} for the {\RandomOffer}, we obtain {\ksfair} mechanism $\mech\primed$ whose GFT is at least $\frac{1}{2}$ fraction of the {\SecondBest} $\OPTSB$. Moreover, it can be verified by the construction in \Cref{thm:blackbox reduction}, mechanism $\mech\primed$ is the {\BiasedRandomOffer}. 
\end{proof}

We next consider the bilateral trade instances where both traders' valuation distributions are MHR. In this setting, prior work \citep{Fei-22} shows that both the {\BuyerOffer} and {\SellerOffer} (which might not be {\ksfair}) guarantee a GFT approximation of $\frac{1}{e - 1}\approx 58.1\%$. Combining this with our black-box reduction, an improved GFT approximation of $\frac{1}{e - 1}$ for the {\ksfair} {\BiasedRandomOffer} is obtained. For completeness, we also include its proof below.

\begin{theorem}[Theorem~3.1 in \citealp{Fei-22}]
\label{thm:BO SO GFT:mhr traders}
    For every bilateral trade instance, if the buyer (resp.\ seller) has MHR valuation distribution, then the {\SellerOffer} (resp.\ {\BuyerOffer}) obtains at least $\frac{1}{e - 1}$ fraction of the {\FirstBest} $\OPTFB$.
\end{theorem}
\begin{restatable}{corollary}{thmmhrtraders}
\label{thm:improved GFT:mhr traders}
For every bilateral trade instance $\btinstance$ where both traders have MHR valuation distributions, there exists probability $\mixprob\in[0, 1]$ such that the {\BiasedRandomOffer} is {\ksfair} and its GFT is at least at least $\frac{1}{e-1}$ fraction of the {\SecondBest} $\OPTSB$.\footnote{The $\frac{1}{e - 1}$ approximation in this theorem also holds for the {\FirstBest} $\OPTFB$, i.e., $\GFT{\mech} \geq \frac{1}{e - 1}\cdot \OPTFB$.}
\end{restatable}
\begin{proof}[Proof of \Cref{thm:improved GFT:mhr traders}]
    Applying \Cref{thm:blackbox reduction} with the (unbiased) {\RandomOffer}, there exists $\mixprob\in[0, 1]$ such that the {\BiasedRandomOffer} is {\ksfair}. Since in the {\BiasedRandomOffer}, either {\BuyerOffer} or {\SellerOffer} are implemented in every realized execution, invoking \Cref{thm:BO SO GFT:mhr traders} finishes the proof.
\end{proof}

%% file: Paper/improved-GFT-result.tex
\label{sec:improved GFT}

In this section, we focus on the special case where the seller's valuation distribution is a single point mass (known to all). This is equivalent to assuming that the seller has zero value, which is publicly known by all parties. This special case has been studied extensively in the mechanism design literature \citep[e.g.,][]{mye-81,BR-89,HR-09}, mainly from the perspective of a monopolist seller aiming to maximize her revenue. In this special case of a zero-value seller, as trading at price of $0$ is always efficient, the {\FirstBest} and {\SecondBest} become identical, and are equal to the expected value of the buyer, i.e., $\OPTFB=\OPTSB = \expect[\val\sim\buyerdist]{\val}$. 

In \Cref{subsec:optimal GFT upper bound:general instance}, we first present a negative result that matches the GFT approximation of $\frac{1}{2}$ established in \Cref{thm:optimal GFT:general instance}, and thus prove the tightness of the result. Notably, this negative result is derived from a zero-value seller instance  where the buyer's valuation distribution is not regular (\Cref{example:all fair mech:irregular}). Motivated by this, we then study the settings where the buyer's valuation distribution is regular (\Cref{subsec:improved GFT:regular buyer}), and then consider settings that the distribution is MHR (\Cref{subsec:improved GFT:mhr buyer}). As we will show, (significantly) improved GFT-approximations are achievable for such instances.

\subsection{Matching Negative Result for Zero-Value Seller Instances}
\label{subsec:optimal GFT upper bound:general instance}
In this section we show that even for the simple setting where the seller has zero value, there exists an instance (\Cref{example:all fair mech:irregular}) for which every BIC, IIR, ex ante WBB mechanism that is {\ksfair}, cannot obtain more than half of the {\SecondBest} $\OPTSB$. We note that the buyer's valuation distribution in this instance is not regular. 

\begin{example}
\label{example:all fair mech:irregular}
Fix any $\constantH \geq 1$. Let $\val\primed \triangleq \frac{\constantH}{\sqrt{\ln\constantH} + 1}$. The buyer has a valuation distribution $\buyerdist$ (that is not regular) defined as follows: % Specifically, 
distribution~$\buyerdist$ 
has support $\supp(\buyerdist)=[1, \constantH]$ and cumulative density function $\buyercdf(\val) = \frac{\val - 1}{\val}$ for $\val \in[1,\val\primed]$, $\buyercdf(\val) = 1- \frac{{\ln\constantH}}{(\val/\constantH+\sqrt{\ln\constantH}-1)\cdot \constantH}$ for $\val\in[\val\primed, \constantH]$, and $\buyercdf(\val) = 1$ for $\val \in(\constantH, \infty)$.  
(Namely, there is an atom at $\constantH$ with probability mass of ${\sqrt{\ln\constantH}}/{\constantH}$.)
The seller has a deterministic value of 0. See \Cref{fig:all fair mech:irregular} for an illustration of the revenue curve and the bargaining set. 
\end{example}

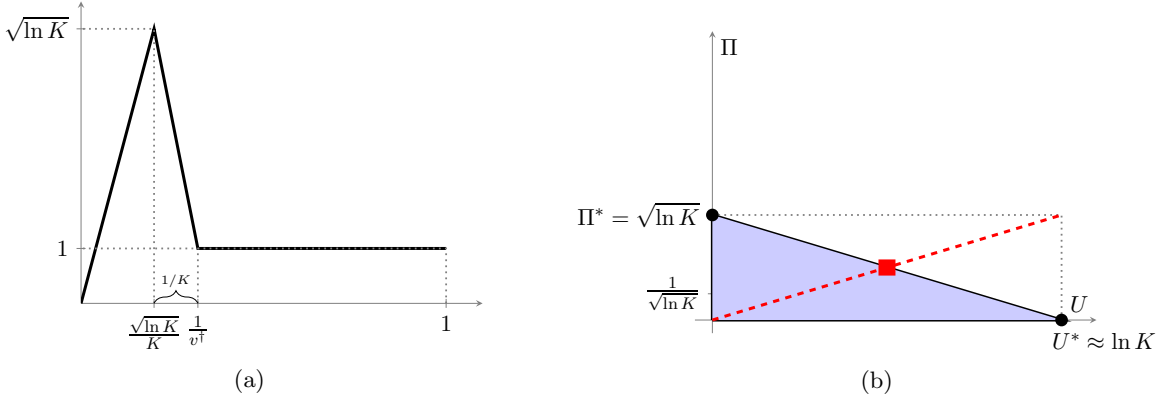
\begin{figure}
    \centering
    \subfloat[]{
\input{Figures/fig-fair-mech-SB-example}
\label{fig:all fair mech:irregular:revenue curve}
}~~~~
    \subfloat[]{
\input{Figures/fig-fair-mech-SB-example-util-pair}
\label{fig:all fair mech:irregular:ex ante utility pair}
}
    \caption{Graphical illustration of \Cref{example:all fair mech:irregular}. In \Cref{fig:all fair mech:irregular:revenue curve}, the black curve is the revenue curve of the buyer. In \Cref{fig:all fair mech:irregular:ex ante utility pair}, the shaded region corresponds to all pairs of buyer and seller's ex ante utilities $(\buyerexanteutil,\sellerexanteutil)$ that are achievable by some BIC, IIR, ex ante WBB mechanism. The two black points represent the traders' ex ante utility pairs $(\buyerbenchmark,0)$ and $(0,\sellerbenchmark)$, induced by the {\BuyerOffer} and {\SellerOffer}, respectively. The red dashed line represents utility pairs that are {\ksfair} (i.e., ones with $\buyerexanteutil/\buyerbenchmark = \sellerexanteutil/\sellerbenchmark$). The red square corresponds to a truthful {\ksfair} mechanism whose GFT is at least $\frac{1}{2}$ fraction of the {\SecondBest} $\OPTSB$.}
    \label{fig:all fair mech:irregular}
\end{figure}

\begin{restatable}{proposition}{lemoptimalGFTUB}
\label{lem:optimal GFT upper bound:irregular}
    Fix any $\eps > 0$.
    For any instance in \Cref{example:all fair mech:irregular} with sufficiently large $\constantH$ (as a function of $\eps$), every BIC, IIR, ex ante WBB mechanism $\mech$ that is {\ksfair} obtains at most $(\frac{1}{2} + \eps)$ fraction of the {\SecondBest} $\OPTSB$, i.e., $\GFT{\mech} \leq (\frac{1}{2} + \eps) \cdot \OPTSB$. Moreover, if mechanism $\mech$ further has a deterministic integral allocation rule (i.e., mechanism $\mech$ is a {\FixPrice}), its GFT is at most $\eps$ fraction of the {\SecondBest} $\OPTSB$, i.e., $\GFT{\mech}\leq \eps\cdot \OPTSB$. 
\end{restatable}
Recall that the {\ksfair} {\BiasedRandomOffer} (from \Cref{thm:optimal GFT:general instance}) is a randomized mechanism.\footnote{Namely, if the item is indivisible, the allocation should be randomized. Otherwise (if the item is divisible), the allocation should be fractional.} \Cref{lem:optimal GFT upper bound:irregular} also shows the necessity of the randomization to achieve constant GFT approximation under {\ksfairness}. 

\begin{proof}[Proof of \Cref{lem:optimal GFT upper bound:irregular}]
    In \Cref{example:all fair mech:irregular}, the seller has zero-value, so the {\SecondBest} $\OPTSB$ is equal to the {\FirstBest} $\OPTFB$, and it  satisfies
    \begin{align*}
        \OPTSB = \expect[\val]{\val} =
        \displaystyle\int_1^{\constantH} \val\cdot \d \buyercdf(\val)
        + \constantH\cdot (1 - \buyercdf(\constantH))
        =
        (1 - o_{}(1))\cdot \ln \constantH
    \end{align*}
    Moreover, two traders' benchmark can be computed as
    \begin{align*}
        \sellerbenchmark = \sqrt{\ln \constantH}
        \;\;
        \mbox{and}
        \;\;
        \buyerbenchmark = \OPTSB = (1 - o_{}(1))\cdot \ln \constantH
    \end{align*}
    We now prove the first claim in the lemma statement by contradiction. Suppose there exists BIC, IIR, ex ante WBB, and {\ksfair} mechanism $\mech = (\alloc,\price,\sellerprice)$ whose GFT is a $(\frac{1}{2} + \Omega_{}(1))$-approximation to the {\SecondBest} $\OPTSB$.
    Note that the GFT of mechanism $\mech$ can be upper bounded as
    \begin{align*}
        \GFT{\mech} & = \expect[\val]{\val\cdot \alloc(\val)} 
        \leq
        \displaystyle\int_1^{\val\primed} \val\cdot \alloc(\val)\cdot \d \buyercdf(\val)
        +
        \displaystyle\int_{\val\primed}^{\constantH} \val\cdot \d \buyercdf(\val)
        + \constantH\cdot (1 - \buyercdf(\constantH))
        \\
        &{}=
        \displaystyle\int_1^{\val\primed} \val\cdot \alloc(\val)\cdot \d \buyercdf(\val)
        +
        o_{}(\ln\constantH)
    \end{align*}
    Invoking the assumption that $\GFT{\mech} \geq (\frac{1}{2} + \Omega_{}(1))\cdot \OPTSB$, we have 
    \begin{align*}
        \displaystyle\int_1^{\val\primed} \val\cdot \alloc(\val)\cdot \d \buyercdf(\val) 
        \geq
        \left(\frac{1}{2} + \Omega_{}(1)\right)
        \cdot \ln\constantH
    \end{align*}
    Moreover, since mechanism $\mech$ is BIC, the buyer's interim allocation $\alloc$ is weakly increasing and thus
    \begin{align*}
        \displaystyle\int_1^{\val\primed} \val\cdot \alloc(\val)\cdot \d \buyercdf(\val)
        \leq 
        \displaystyle\int_1^{\val\primed} \val\cdot \alloc(\val\primed)\cdot \d \buyercdf(\val)
        =
        (\alloc(\val\primed) - o_{}(1))\cdot \ln\constantH
    \end{align*}
    which further implies 
    \begin{align*}
        \alloc(\val\primed) \geq  \frac{1}{2} + \Omega_{}(1)
    \end{align*}
    We next upper bound the seller's ex ante utility $\sellerexanteutil(\mech)$ as follows
    \begin{align*}
        \sellerexanteutil(\mech) &{} = 
        \sellerprice(0) \overset{(a)}{\leq} 
        \expect[\val]{\price(\val)}
        \overset{(b)}{=}
        \expect[\val]{\virtualval(\val)\cdot \alloc(\val)}
        \overset{(c)}{\leq} 
        \virtualval(\constantH)\cdot (1 - \buyercdf(\constantH))
        +
        \virtualval(\val\primed)\cdot (\buyercdf(\constantH) - \buyercdf(\val\primed)) \cdot \alloc(\val\primed)
        \\
        &{}
        =
        \left(1 - \alloc(\val\primed) + o_{}(1)\right)\cdot \sqrt{\ln\constantH}
        \overset{(d)}{\leq}
        \left(\frac{1}{2} - \Omega_{}(1)\right) \cdot \sqrt{\ln\constantH}
    \end{align*}
    where inequality~(a) holds since mechanism $\mech$ is ex ante WBB, 
    equality~(b) holds due to \Cref{prop:revenue equivalence},
    inequality~(c) holds since the virtual value function $\virtualval$ satisfies $\virtualval(\constantH) = \constantH \geq 0$, $\virtualval(\val) = \virtualval(\val\primed) < 0$ for $\val\in[\val\primed,\constantH)$, $\virtualval(\val) = 0$ for $\val\in[1, \val\primed)$, and the buyer's interim allocation $\alloc$ is weakly increasing due to BIC, and inequality~(d) holds since $\alloc(\val\primed) \geq \frac{1}{2} + \Omega_{}(1)$ argued above.
    
    Meanwhile, the buyer's ex ante utility $\buyerexanteutil(\mech)$ can be lower bounded as 
    \begin{align*}
        \buyerexanteutil(\mech) &{}= \expect[\val]{\val\cdot \alloc(\val) - \price(\val)}
        \overset{(a)}{\geq} \expect[\val]{\val\cdot \alloc(\val)} - \sqrt{\ln\constantH}
        \overset{(b)}{\geq}
        \left(\frac{1}{2} + \Omega_{}(1)\right)
        \cdot \ln\constantH
    \end{align*}
    where the inequality~(a) holds since the expected payment of the buyer is upper bounded by the monopoly revenue $\sqrt{\ln\constantH}$, and inequality~(b) holds since $\int_1^{\val\primed} \val\alloc(\val)\cdot \d \buyercdf(\val) \geq (\frac{1}{2} + \Omega_{}(1)) \cdot \ln\constantH$ argued above.

    Putting all the pieces together, both traders' ex ante utilities satisfy
    \begin{align*}
        \frac{\sellerexanteutil(\mech)}{\sellerbenchmark} &\leq 
        \frac{\left(\frac{1}{2} - \Omega_{}(1)\right) \cdot \sqrt{\ln\constantH}}{\ln\constantH}
        =
        \frac{1}{2} - \Omega_{}(1)
        \;\;
        \mbox{and}
        \;\;-
        \frac{\buyerexanteutil(\mech)}{\buyerbenchmark} &\geq 
        \frac{\left(\frac{1}{2} + \Omega_{}(1)\right)
        \cdot \ln\constantH}{(1 - o_{}(1))\cdot \ln \constantH}
        =
        \frac{1}{2} + \Omega_{}(1)
    \end{align*}
    which implies mechanism $\mech$ is not {\ksfair}. This is a contradiction.
    
    We now prove the second claim in the lemma statement. Fix any BIC, IIR, ex ante WBB, {\ksfair} mechanism $\mech= (\alloc,\price,\sellerprice)$ that has a deterministic integral allocation rule $\alloc$. We consider two cases depending on allocation $\alloc(\val\primed)$ at value $\val\primed$. 
    
    \xhdr{Case (a) $\alloc(\val\primed) = 1$:} In this case, we can upper bound the seller's ex ante utility as 
    \begin{align*}
        \sellerexanteutil(\mech) = \sellerprice(0) \overset{(a)}{\leq} \expect[\val]{\price(\val)} \overset{(b)}{=} \expect[\val]{\virtualval(\val) \cdot \alloc(\val)}
        \overset{(c)}{\leq} \expect[\val]{\virtualval(\val)\cdot \indicator{\val \geq \val\primed}} = 1
    \end{align*}
    where inequality~(a) holds due to ex ante WBB, equality~(b) holds due to \Cref{prop:revenue equivalence}, and inequality~(c) holds since $\virtualval(\val) = 0 $ for every value $\val<\val\primed$ and $\alloc(\val) = \alloc(\val\primed) = 1$ for every value $\val \geq \val\primed$ (implied by BIC). Therefore, the GFT of mechanism can be upper bounded by
    \begin{align*}
        \GFT{\mech} &= \sellerexanteutil(\mech) + \buyerexanteutil(\mech) + \expect[\val,\cost]{\price(\val) - \sellerprice(\cost)}
        \overset{}{\leq } \frac{\sellerexanteutil(\mech)}{\sellerbenchmark}\cdot \left(\sellerbenchmark + \buyerbenchmark\right) + \expect[\val]{\price(\val)}
        \\
        &\leq 
        (1 - o_{}(1))\cdot \sqrt{\ln \constantH}
        = \bigO_{}\left(\frac{1}{\sqrt{\ln\constantH}}\right)\cdot \OPTSB
    \end{align*}
    where the first inequality holds due to the {\ksfairness} and the second inequality holds due to $\sellerexanteutil(\mech)\leq 1$ and $\expect[\val]{\price(\val)} \leq 1$ argued above.

    \xhdr{Case (b) $\alloc(\val\primed) = 0$:} In this case, we can upper bound the buyer's ex ante utility as 
    \begin{align*}
        \buyerexanteutil(\mech) \leq \expect[\val]{\val\cdot \indicator{\val\geq \val\primed}} = (1 + o_{}(1))\cdot \sqrt{\ln \constantH}
    \end{align*}
    and thus the GFT can be upper bounded as
    \begin{align*}
        \GFT{\mech} &= \sellerexanteutil(\mech) + \buyerexanteutil(\mech) + \expect[\val,\cost]{\price(\val) - \sellerprice(\cost)}
        \leq 2 \sellerbenchmark + \buyerexanteutil(\mech) 
        \\
        &\leq (3 +o_{}(1))\cdot \sqrt{\ln \constantH}
        = \bigO_{}\left(\frac{1}{\sqrt{\ln\constantH}}\right)\cdot \OPTSB
    \end{align*}
    where the first inequality holds since $\sellerexanteutil(\mech) \leq \sellerbenchmark$ and $\expect[\val,\cost]{\price(\val) - \sellerprice(\cost)}\leq\sellerbenchmark$.
    We now finish the proof of \Cref{lem:optimal GFT upper bound:irregular}.
\end{proof}

\subsection{Zero-Value Seller and Buyer with a Regular Distribution}
\label{subsec:improved GFT:regular buyer}
\input{Paper/zero-seller-regular-buyer}

\subsection{Zero-Value Seller and Buyer with an MHR Distribution}
\label{subsec:improved GFT:mhr buyer}
\input{Paper/zero-seller-mhr-buyer}

%% file: Figures/fig-fair-mech-SB-example.tex
\begin{tikzpicture}[scale=0.8, transform shape]
\begin{axis}[
axis line style=gray,
axis lines=middle,
xtick={0, 0.2, 0.32, 1},
ytick={0, 1, 5},
xticklabels={0, $\frac{\sqrt{\ln\constantH}}{\constantH}$, $\frac{1}{\val\primed}$, 1},
yticklabels={0, 1, ${\sqrt{\ln\constantH}}$},
xmin=0,xmax=1.1,ymin=-0.0,ymax=5.5,
width=0.5\textwidth,
height=0.4\textwidth,
samples=1000]

\addplot[black!100!white, line width=0.5mm] (0, 0) -- (0.2, 5) -- (0.32, 1) -- (1, 1);

\addplot[dotted, gray, line width=0.3mm] (1, 0) -- (1, 1) -- (0, 1);
\addplot[dotted, gray, line width=0.3mm] (0.2, 0) -- (0.2, 5) -- (0, 5);
\addplot[dotted, gray, line width=0.3mm] (0.32, 0) -- (0.32, 1);

\draw[decorate,decoration={brace,amplitude=6pt}] (0.2,0) -- (0.32,0) node[midway,above=4pt] {\tiny $1/\constantH$};

\end{axis}

\end{tikzpicture}

%% file: Figures/fig-fair-mech-SB-example-util-pair.tex
\begin{tikzpicture}[scale=0.8, transform shape]
\begin{axis}[
axis line style=gray,
axis lines=middle,
xlabel = $\buyerexanteutil$,
ylabel = $\sellerexanteutil$,
xtick={0, 1},
ytick={0, 0.1, 0.4},
xticklabels={0, ~~~~~~~~~~~$\buyerbenchmark\approx\ln\constantH$},
yticklabels={0, $\frac{1}{\sqrt{\ln\constantH}}$, $\sellerbenchmark = \sqrt{\ln\constantH}$},
xmin=-0.05,xmax=1.1,ymin=-0.05,ymax=1.1,
width=0.5\textwidth,
height=0.4\textwidth,
samples=1000]

\addplot[black!100!white, line width=0.5mm] (0, 0) -- (0, 0.4) -- (1, 0) -- (0, 0);

\fill[blue!20] (0, 0) -- (0, 0.4) -- (1, 0) -- (0, 0);

\addplot[dotted, gray, line width=0.3mm] (1, 0) -- (1, 0.4) -- (0, 0.4);

\addplot[dashed, red, line width = 0.5mm] (0, 0) -- (1, 0.4);

\draw[black, fill=black, line width=0.5mm] (axis cs:1, 0) circle[radius=0.08cm];
\draw[black, fill=black, line width=0.5mm] (axis cs:0, 0.4) circle[radius=0.08cm];

\draw[red, fill=red, line width=0.5mm] (0.48, 0.175) rectangle (0.52, 0.225);

\end{axis}

\end{tikzpicture}

%% file: Paper/zero-seller-regular-buyer.tex
In this section we consider zero-value seller settings when the buyer's valuation distribution is assumed to be regular. Note that due to \Cref{example:all fair mech:irregular} and \Cref{lem:optimal GFT upper bound:irregular}, when no assumption is made about the valuation distribution, there are instances in which any {\ksfair} mechanism has GFT that is at most half the {\SecondBest}.

By imposing regularity on the buyer's valuation distribution and restricting the seller's value to be deterministically equal to zero, a natural first attempt is to explore whether the {\BiasedRandomOffer} mechanism, as developed in \Cref{cor:biased random offer}, can achieve an improved GFT approximation ratio strictly greater than $\frac{1}{2}$. However, the following example, along with its corresponding lemma, demonstrates the failure of the {\BiasedRandomOffer} to achieve a better GFT approximation.

\begin{example}
\label{example:BROM:regular}
Fix any $\constantH \geq 1$. The buyer has a regular valuation distribution $\buyerdist$, which has support $[0, \constantH]$ and cumulative density function $\buyercdf(\val) = \frac{(\constantH - 1)\cdot \val}{(\constantH - 1)\cdot \val + \constantH}$ for $\val \in[0,\constantH]$ and $\buyercdf(\val) = 1$ for $\val \in(\constantH, \infty)$.
(Namely, there is an atom at $\constantH$ with probability mass of $\frac{1}{\constantH}$.)
The seller has a deterministic value of 0.
See \Cref{fig:BROM:regular} for an illustration.
\end{example}

\begin{restatable}{lemma}{lemBROMregular}
\label{lem:BROM:regular} 
    Fix any $\eps > 0$. In \Cref{example:BROM:regular} with sufficiently large $\constantH$ (as a function of $\eps$), the {\ksfair} {\BiasedRandomOffer} $\mech$ (from \Cref{cor:biased random offer}) obtains less than $(\frac{1}{2} + \eps)$ of the {\SecondBest} $\OPTSB$, i.e., $\GFT{\mech} < (\frac{1}{2} + \eps)\cdot \OPTSB$.
\end{restatable}
\begin{proof}
    In \Cref{example:BROM:regular}, the {\SecondBest} $\OPTSB = \expect[\val]{\val} = (1 - o_{}(1))\cdot \ln\constantH$. The {\ksfair} {\BiasedRandomOffer} $\mech$ is exactly the (unbiased) {\RandomOffer}. Consequently, the GFT of the {\BiasedRandomOffer} $\mech$ is $\GFT{\mech} = (\frac{1}{2} + o_{}(1))\cdot \ln\constantH$. This completes the proof as desired. 
\end{proof}

\Cref{lem:BROM:regular} shows that to obtain a larger fraction of GFT by a {\ksfair} mechanism we cannot use to the {\BiasedRandomOffer}, so next we consider alternative mechanisms. While \Cref{lem:optimal GFT upper bound:irregular} showed that the {\FixPrice} cannot obtain good GFT approximation for general buyer distributions, we move to study it when the distribution is regular. For the setting of \Cref{example:BROM:regular} it can be verified that there exists a particular price $\fprice\in(0, \constantH)$ such that the {\FixPrice} with trading price $\fprice$ is {\ksfair}. To see this, note that for this example, in the {\FixPrice}, as trading price $\price$ increases, the buyer's ex ante utility is continuous and strictly decreasing, while the seller's ex ante utility is continuous and strictly increasing up to the monopoly reserve $\optreserve = \constantH$ (since the buyer's valuation distribution is regular and thus revenue curve is concave). Moreover, at the two extreme cases where price $\price = 0, \constantH$, two traders' ex ante utilities become zero and their benchmarks $\sellerbenchmark, \buyerbenchmark$, respectively. Hence, the existence of {\ksfair} trading price $\fprice\in(0, \constantH)$ is guaranteed by the intermediate value theorem. Finally, for this example, by simple numerical calculation, it can be checked that the GFT approximation of the {\FixPrice} with trading price $\fprice$ is also large. For example, with $\constantH = 25$, the GFT approximation is $\fixedPriceGFTPercentageUBRegular$. (Analytical expression can be found in Eqn.~\eqref{eq:BROM:regular:fair price} and \eqref{eq:BROM:regular:fair price GFT approx}.) 
Also see \Cref{fig:BROM:regular} for an illustration.

As the main result of this subsection, we extend the above observation regarding \Cref{example:BROM:regular} to general zero-value seller instances where the buyer has a regular valuation distribution.\footnote{{We remind that when no assumption (such as regularity) is made for the buyer's valuation distribution, even if the seller is zero-value, any {\ksfair} {\FixPrice} can only achieve $\eps$ fraction of the {\SecondBest} $\OPTSB$ for every $\eps > 0$ (\Cref{lem:optimal GFT upper bound:irregular}).}} Although arguing the existence of such a {\ksfair} {\FixPrice} is relatively easy due to the intermediate value theorem, proving its almost-tight GFT approximation is \emph{technically highly non-trivial}. 

\begin{figure}
    \centering
    \input{Figures/fig-GFT-regular-buyer}
    \caption{Graphical illustration of \Cref{example:BROM:regular} when $\constantH = 25$. The x-axis is quantile $\quant$. 
    The black solid line is the revenue curve of the buyer.
    Consider {\FixPrice} $\mech$ with trading price $\price = \val(\quant)$ for every quantile $\quant$. The black solid (resp.\ dashed) curve also represents   ${\sellerexanteutil(\mech)}/{\sellerbenchmark}$ (resp.\ ${\buyerexanteutil(\mech)}/{\buyerbenchmark}$) for the seller (resp.\ buyer). The red curve is the GFT approximation ratio ${\GFT{\mech}}/{\OPTSB}$. 
    A {\ksfair} {\FixPrice} is achieved at trading price $\fprice = \val(\fquant)$ with GFT approximation ratio of $0.877$ (when $\constantH = 25$).
    Finally, the blue curve is used in the proof of the negative result in \Cref{lem:GFT UB:regular buyer}.}
    \label{fig:BROM:regular}
\end{figure}
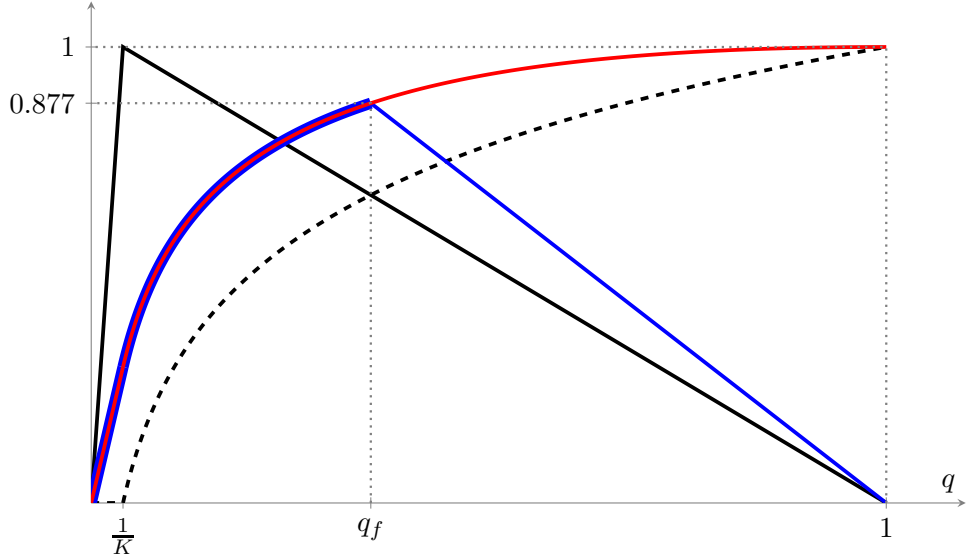

\begin{restatable}{theorem}{thmGFTregularbuyer}
\label{thm:improved GFT:regular buyer}
    For every bilateral trade instance of a zero-value seller and a buyer with a regular valuation distribution $\buyerdist$, there exists price $\fprice$ smaller than the smallest monopoly reserve $\optreserve$ for distribution $\buyerdist$, such that the {\FixPrice} with trading price $\fprice$ is {\ksfair} and its GFT is at least $\calC$ fraction of the {\SecondBest} $\OPTSB$. Here $\calC$ is the solution to minimization program~\ref{program:GFT:regular buyer},
    and we observe that 
    $\calC\geq \fixedPriceGFTPercentageRegular$ by a numerical computation.

    Moreover, there exists an instance of a zero-value seller and a  buyer with a regular valuation distribution, in which any BIC, IIR, ex ante WBB mechanism that is {\ksfair} has GFT that is less than  $\fixedPriceGFTPercentageUBRegular$ of the {\SecondBest} $\OPTSB$.
\end{restatable}

\Cref{thm:improved GFT:regular buyer} follows directly from two technical lemmas (\Cref{lem:GFT program:regular buyer,lem:GFT UB:regular buyer}) which establish the positive and negative results in the theorem statement respectively. The technical overview of our argument are given under each lemma.

\subsubsection{GFT Approximation by KS-Fair {\FixPrice}}
We characterize the GFT approximation of a {\ksfair} {\FixPrice} as follows.
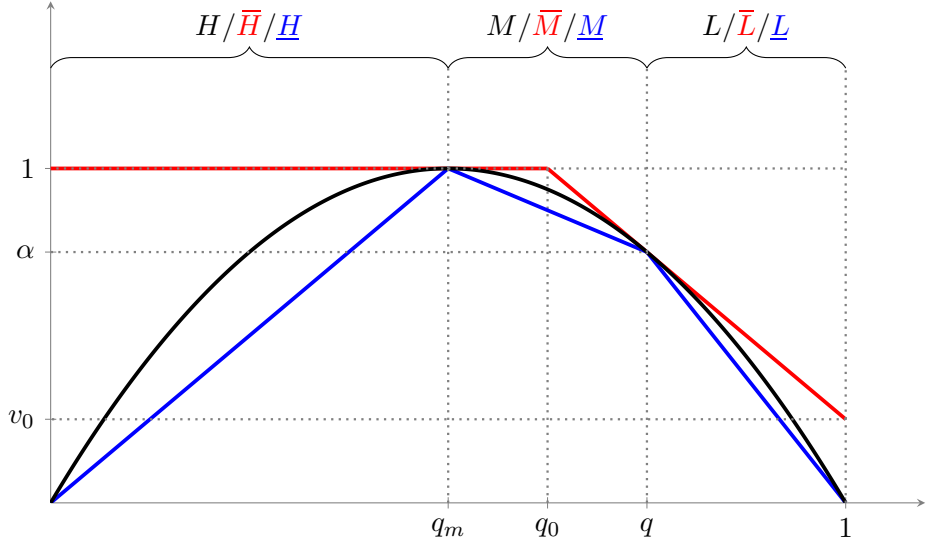
\begin{figure}
    \centering
    \input{Figures/fig-base-model-regular-revenue-curve-analysis}
    \caption{Graphical illustration of the analysis for \Cref{lem:GFT program:regular buyer}. The black curve is the concave revenue curve $\revcurve$ of the buyer. The red and blue revenue curves $\revcurve_1, \revcurve_2$ (defined in the analysis) sandwich the original revenue curve $\revcurve$.}
    \label{fig:GFT program:regular buyer}
\end{figure}
\begin{lemma}
\label{lem:GFT program:regular buyer}
    For every bilateral trade instance of a zero-value seller and buyer with regular valuation distribution $\buyerdist$, there exists price $\fprice$ smaller than the smallest monopoly reserve $\optreserve$ for distribution $\buyerdist$, such that the {\FixPrice} with trading price $\fprice$ is {\ksfair} and its GFT is at least $\calC$ fraction of the {\SecondBest} $\OPTSB$. Here $\calC$ is the solution to the following program~\ref{program:GFT:regular buyer}:
    \begin{align}
    \label{program:GFT:regular buyer}
    \tag{$\mathcal{P}_{\mathrm{REG}}$}
    \arraycolsep=5.4pt\def\arraystretch{1}
        \begin{array}{llll}
          \calC \triangleq~~&\min\limits_{\optquant, H}
          ~\max\limits_{\revratio}
          ~\min\limits_{\substack{\quant, \val_0, M, L}}   & 
          \displaystyle\exanteutilratio + \frac{\exanteutilratio}{H + M + L} &
          \vspace{10pt}
          \\
          \vspace{10pt}
          &\text{s.t.}
          & \optquant\in[0, 1]~,  & 
          \\
          \vspace{10pt}
          && \revratio\in(0, 1)~,  & 
          \\ 
          \vspace{10pt}
          && \quant\in \left[\optquant + (1 - \revratio)(1 - \optquant), 1\right]~,  & 
          \\
          \vspace{10pt}
          && \val_0\in \left[0, 1 - \displaystyle\frac{1 - \revratio}{\quant - \optquant}(1 - \optquant)\right]~,  & 
          \\
          \vspace{10pt}
          && H\in\left[\HUnderBar, \HOverBar\right]~, 
          M\in\left[\MUnderBar, \MOverBar\right]~,
          L\in\left[\LUnderBar, \LOverBar\right]~,& 
        \end{array}
    \end{align}
    where $\exanteutilratio$, $\quant_0$, $\HUnderBar, \HOverBar, \MUnderBar, \MOverBar, \LUnderBar, \LOverBar$ are auxiliary variables constructed as 
    \begin{align*}
        \exanteutilratio &\triangleq  \displaystyle
        \revratio - \plus{\frac{H + M + L}{H + M + L + 1}\left(\revratio - \frac{H + M - \revratio}{H + M + L}\right)}~,
        \\
        \quant_0 &\triangleq 1 - \frac{(1 - \val_0)(1 - q)}{\revratio - \val_0}~,
        \\
        \HUnderBar &\triangleq 1~, \HOverBar \triangleq \infty~, 
        \\
        \MUnderBar &\triangleq \ln\left(\frac{q}{\optquant}\right)\left(1 + \frac{1 - \revratio}{q - \optquant}\optquant\right) - 1 + \revratio ~,
        \\
        \MOverBar &\triangleq  \ln\left(\frac{\quant_0}{\optquant}\right) 
        +
        \ln\left(\frac{q}{\quant_0}\right)\left(\val_0 + \frac{\revratio - \val_0}{1 - q}\right) 
        -
        \frac{q - \quant_0}{1 - q}(\revratio - \val_0)~,
        \\
        \LUnderBar &\triangleq \frac{\revratio}{1 - \quant}\ln\left(\frac{1}{\quant}\right)~,
        % \\
        \LOverBar \triangleq \ln\left(\frac{1}{q}\right)\left(\val_0 + \frac{\revratio - \val_0}{1 - q}\right) - \revratio + \val_0~.
    \end{align*}
    By a numerical computation (see \Cref{apx:numerical evaluation:regular buyer}), we observe that 
    $\calC\geq \fixedPriceGFTPercentageRegular$.
\end{lemma}

\xhdr{Proof overview of \Cref{lem:GFT program:regular buyer}.} We first explain the idea of proving the lemma above and the optimization program in its statement. As we explained before, to obtain an GFT approximation better than $\frac{1}{2}$, we cannot compare the GFT with $\sellerbenchmark + \buyerbenchmark$, since it may be twice larger than the {\SecondBest}.\footnote{Since the seller has zero value, the first and {\SecondBest}s both equal to the expected value of the buyer, i.e., $\OPTFB = \OPTSB = \expect[\val]{\val}$.} Hence, unlike the proof of \Cref{thm:blackbox reduction} in the previous section, our analysis of \Cref{lem:GFT program:regular buyer} aims to directly compare the GFT of a {\ksfair} {\FixPrice} with the {\SecondBest}. Loosely speaking, {for a given regular distribution} we establish a lower bound on the GFT approximation. We then search over all regular distributions and argue that the worst regular distribution that minimizes the established GFT lower bound belongs to a distribution subclass, whose revenue curve can be characterized by finite many parameters (i.e., variables in the optimization program in the statement of \Cref{lem:GFT program:regular buyer}, also see \Cref{fig:GFT program:regular buyer}). For this distribution subclass, we lower bound its GFT approximation as program~\ref{program:GFT:regular buyer}.

Reducing the class of regular distributions to a subclass with finitely many parameters is a common technique for deriving approximation bounds in the mechanism design literature \citep[e.g.,][]{AHNPY-18, JLQTX-19, FHL-21, JL-23}. However, most prior work focuses on revenue approximation and allows consideration of a class of mechanisms (e.g., anonymous pricing with all possible prices). This makes it relatively easy to argue that the approximation ratio depends on a \emph{small set} of parameters (e.g., monopoly quantile, monopoly price) of the revenue curve. In contrast, our analysis requires studying both revenue and GFT (expected value). Furthermore, we focus on a single mechanism ({\ksfair} {\FixPrice}). Hence, both the GFT benchmark and the GFT of our mechanism are highly sensitive to the \emph{entire} revenue curve.

To overcome this challenge, our technical lemma (\Cref{lem:GFT program:regular buyer}) develops a three-step argument. We first consider a (possibly not {\ksfair}) {\FixPrice} whose trading price $\price$ is selected based on a \emph{small set} of parameters of the revenue curve. In particular, depending on the monopoly quantile $\optquant$ and the expected value $H = \expect[\val]{\val\cdot \indicator{\val \geq \optreserve}}$ above the monopoly reserve $\optreserve$, we select a constant $\revratio\in(0, 1)$ and let $\price$ be the trading price smaller than the monopoly reserve such that by posting price $\price$, the seller's ex ante utility (revenue) is $\revratio$ fraction of her benchmark (monopoly revenue) $\sellerbenchmark$, i.e., $\price\cdot(1 - \buyercdf(\price)) = \revratio\cdot \sellerbenchmark$. Our second step is conceptually similar to our previous black-box reduction (\Cref{thm:blackbox reduction}). In particular, we argue that depending on buyer's ex ante utility under trading price $\price$, we can decrease (or increase) this price to obtain a {\ksfair} trading price $\fprice$. This step allows us to lower bound the GFT approximation of a {\ksfair} {\FixPrice} mechanism with trading price $\fprice$ as a \emph{function of the price $\price$ chosen in the first step}. Finally, since this GFT lower bound depends on the choice of $\price$ determined in the first step, we formulate the problem of selecting the optimal price $\price$ (to maximize the GFT approximation lower bound) as a two-player zero-sum game. In this game, the ``min player'' (adversary) selects the revenue curve (represented by finitely many quantities on the revenue curve), while the ``max player'' (ourselves) selects the price $\price$. The game's payoff is the GFT approximation lower bound derived in the second step. Finally, we numerically solve this two-player zero-sum game and obtain $\fixedPriceGFTPercentageRegular$ stated in the lemma.

\begin{proof}[Proof of \Cref{lem:GFT program:regular buyer}]
    Let $\revcurve$ be the revenue curve of the buyer and $\optquant\in \argmax_{\quant}\revcurve(\quant)$ be the monopoly quantile. If there exist multiple monopoly quantiles, we let $\optquant$ be the largest one. Without loss of generality, we normalize the monopoly revenue to be equal to one, i.e., $\revcurve(\optquant) = 1$ and thus monopoly reserve $\optreserve = \frac{1}{\optquant}$. Since the buyer's valuation distribution is regular, revenue curve $\revcurve$ is concave (\Cref{lem:concave revenue curve}).
    
    \xhdr{Step 0- Introducing necessary notations.} We introduce $\revratio\in(0, 1)$ as a constant whose value will be pined down at the end of this analysis. Given constant $\revratio$, let $\quant\in[\optquant, 1]$ be the largest quantile such that $\revcurve(\quant) \geq \revratio$. Moreover, let $\val_0 \triangleq \revcurve(\quant) + (1-\quant)\cdot \revcurve'(\quant)$ and $\quant_0 \triangleq \quant + (1 - \revratio) / \revcurve'(\quant)$. Finally, we also define 
    \begin{align*}
        H &\triangleq \expect[\val]{\val\cdot \indicator{\val \geq \optreserve}},
        \;\;
        M \triangleq \expect[\val]{\val\cdot \indicator{\revratio/\quant \leq \val < \optreserve}},
        \;\;
        L \triangleq \expect[\val]{\val\cdot \indicator{ \val < \revratio/\quant}}
    \end{align*}
    which partition the {\SecondBest} into three pieces, i.e., $\OPTSB = \expect[\val]{\val} = H + M + L$. All notations are illustrated in \Cref{fig:GFT program:regular buyer}.

    \xhdr{Step 1- Characterizing of a (possibly) not {\ksfair} {\FixPrice}.}
    First, we consider a (possibly not {\ksfair}) {\FixPrice} ({$\FPM_{\price}$}) with trading price $\price = \frac{\revratio}{\quant}$:
    \begin{itemize}
        \item For the seller, her ex ante utility (aka., revenue) is $\sellerexanteutil(\FPM_{\price}) = \price\cdot (1 - \buyercdf(\price)) = \price\cdot \quant = \revratio$. Since her benchmark $\sellerbenchmark$ (aka., monopoly revenue) is normalized to one, her ex ante utility is an $\alpha$ fraction of her benchmark $\sellerbenchmark$.
        \item For the buyer, her ex ante utility can be computed as 
        \begin{align*}
            \buyerexanteutil(\FPM_{\price}) &= 
            \expect[\val]{\plus{\val - \price}} = 
            H + M - \revratio
        \end{align*}
        Meanwhile, the buyer's benchmark $\buyerbenchmark$ is 
        \begin{align*}
            \buyerbenchmark = H + M + L
        \end{align*}
    \end{itemize}
    Putting the two pieces together, we conclude that in this (possibly not {\ksfair}) {\FixPrice} ({$\FPM_{\price}$}), both traders' ex ante utilities satisfy
    \begin{align*}
        \frac{\sellerexanteutil(\FPM_{\price})}{\sellerbenchmark} = \revratio
        \;\;
        \mbox{and}
        \;\;
        \frac{\buyerexanteutil(\FPM_{\price})}{\buyerbenchmark} = \frac{H + M - \revratio}{H + M + L}
    \end{align*}
    
    \xhdr{Step 2- Characterizing of {\ksfair} {\FixPrice}.} Define auxiliary notation $\exanteutilratio\in[0, 1]$ as 
    \begin{align*}
        \exanteutilratio \triangleq \displaystyle
        \revratio - \plus{\frac{H + M + L}{H + M + L + 1}\left(\revratio - \frac{H + M - \revratio}{H + M + L}\right)}
    \end{align*}
    We next show that there exists price $\fprice\in[0,\optreserve]$ such that the {\FixPrice} with trading price $\fprice$ is {\ksfair} and both traders' ex ante utilities are at least $\exanteutilratio$ fraction of their benchmarks $\sellerbenchmark,\buyerbenchmark$, respectively. To see this, consider the following two cases separately.
    \begin{itemize}
        \item Suppose that $\revratio < \frac{H + M - \revratio}{H + M + L}$ and thus $\exanteutilratio = \revratio$. In this case, by increasing trading price $\price$ in the {\FixPrice}, the seller's ex ante utility increases continuously (due to the concavity of revenue curve $\revcurve$) and the buyer's ex ante utility decreases continuously. Invoking the intermediate value theorem, there exists price $\fprice \in (\price, \optreserve)$ such that both traders' ex ante utilities is at least $\exanteutilratio$ fraction of their benchmarks $\sellerbenchmark,\buyerbenchmark$, respectively.
        \item Suppose that $\revratio \geq \frac{H + M - \revratio}{H + M + L}$ and thus $\exanteutilratio = \revratio - {\frac{H + M + L}{H + M + L + 1}\left(\revratio - \frac{H + M - \revratio}{H + M + L}\right)}$. In this case, let $\Delta \triangleq {\frac{H + M + L}{H + M + L + 1}\left(\revratio - \frac{H + M - \revratio}{H + M + L}\right)}\geq 0$. By decreases trading price $\price$ in the {\FixPrice}, the seller's ex ante utility decreases continuously (due to the concavity of revenue curve $\revcurve$). Let $\price\primed < \price$ be the trading price such that the seller's ex ante   utility (aka., revenue) is equal to $\revratio - \Delta$. Under trading price $\price\primed$, the buyer's ex ante utility is at least $H + M - \revratio + \Delta$. (To see this, note that by decreasing trading price from $\price$ to $\price\primed$, the GFT weakly increases and the seller's ex ante utility decreases by $\Delta$. Thus, the buyer's ex ante utility increases by at least $\Delta$.) Due to the definition of $\Delta$, two traders' ex ante utilities in the {\FixPrice} ($\FPM_{\price\primed}$) with trading price $\price\primed$ satisfy
        \begin{align*}
            \frac{\sellerexanteutil(\FPM_{\price\primed})}{\sellerbenchmark} = \revratio - \Delta = \frac{H + M - \revratio + \Delta}{H + M + L}
            \leq 
            \frac{\buyerexanteutil(\FPM_{\price\primed})}{\buyerbenchmark}
        \end{align*}
        If the inequality above holds with equality, the {\FixPrice} with trading price $\price\primed$ is {\ksfair} and both traders' ex ante utilities are at least $\exanteutilratio$ fraction of their benchmarks $\sellerbenchmark,\buyerbenchmark$, respectively. Otherwise, we can invoke argument in the previous case.
    \end{itemize}
    Summarizing the analysis above, the GFT-approximation of the {\ksfair} {\FixPrice} ($\FPM_{\fprice}$) with trading price $\fprice$ can be computed as 
    \begin{align*}
        \frac{\GFT{\FPM_{\fprice}}}{\OPTSB}
        \geq
        \frac{\exanteutilratio \cdot (\sellerbenchmark + \buyerbenchmark)}{\expect[\val]{\val}}
        =
        \exanteutilratio + \frac{\exanteutilratio}{H + M + L}
    \end{align*}

    \xhdr{Step 3- Formulating GFT approximation as two-player game.} Putting all the pieces together, the optimization program in the lemma statement can be viewed as a two-player zero-sum game between a min player (adversary) and a max player (ourself as GFT-approximation prover). The payoff in this game is the GFT approximation lower bound $\exanteutilratio + \frac{\exanteutilratio}{H + M + L}$ shown above. As a reminder, quantities $\exanteutilratio, H, M, L$ depend on both the buyer's valuation distribution and constant $\revratio$ used in the analysis. The min player chooses the worst regular distribution (equivalently, concave revenue curve) of the buyer, and the max player chooses constant $\revratio$. Importantly, the choice of constant $\revratio$ can depend on the buyer's valuation distribution. To capture this, we formulate this two-player zero-sum game in three stages: 
    \begin{itemize}
        \item (Stage 1) The min player (adversary) chooses monopoly quantile $\optquant$ and $H$.
        \item (Stage 2) The max player (ourself as GFT-approximation prover) chooses constant $\revratio$ for the analysis.
        \item (Stage 3) The min player (adversary) chooses $\val_0, \quant, M, L$.
    \end{itemize}
    It remains to verify that all constraints in the optimization program capture the feasibility condition for the both min player and max player's actions. We next verify the non-trivial constraints individually. 
    \begin{itemize}
        \item (Lower bound for quantile $\quant$) Recall that $\quant\in[\optquant, 1]$ is the largest quantile such that $\revcurve(\quant) \geq \revratio$. 
        The concavity of revenue curve $\revcurve$ implies that $\revcurve(\quant) \geq \frac{1-\quant}{1-\optquant}\revcurve(\optquant)$. Since $\revcurve(\quant) = \revratio$ and $\revcurve(\optquant) = 1$, we obtain $\quant\geq \optquant + (1 - \revratio)(1 - \optquant)$ as stated in the constraint.
        \item (Upper bound for value $\val_0$) Recall that $\val_0 \triangleq \revcurve(\quant) + (1-\quant)\cdot \revcurve'(\quant)$. The concavity of revenue curve $\revcurve$ implies that $\revcurve'(\quant) \leq \frac{\revcurve(\quant) - \revcurve(\optquant)}{\quant - \optquant} = \frac{\revratio - 1}{\quant - \optquant}$. After rearranging, we obtain $\val_0 \leq 1 - \frac{1 - \revratio}{\quant - \optquant}(1 - \optquant)$ as stated in the constraint.
        \item (Bounds for truncated GFT $L$) Recall that $L \triangleq \expect[\val]{\val\cdot \indicator{ \val < \revratio/\quant}}$. The concavity of the revenue curve implies that for every quantile $\quant\primed \in [\quant, 1]$,
        \begin{align*}
            \frac{1 - \quant\primed}{1 - \quant}\cdot \revcurve(\quant)
            \leq 
            \revcurve(\quant\primed) 
            \leq
            \frac{1 - \quant\primed}{1 - \quant}
            \cdot (\revcurve(\quant) - \val_0)
            +
            \val_0
        \end{align*}
        where both inequalities bind at $\quant\primed = \quant$. The left-hand side and right-hand side can be viewed as two revenue curves $\revcurve_1, \revcurve_2$ that sandwich the original revenue curve $\revcurve$ (see blue and red revenue curves illustrated in \Cref{fig:GFT program:regular buyer}). It can be verified that $\LUnderBar$ and $\LOverBar$ are $\expect[\val]{\val\cdot \indicator{ \val < \revratio/\quant}}$ where the random value $\val$ is realized from valuation distribution induced by those two revenue curves $\revcurve_1, \revcurve_2$, respectively. Invoking \Cref{lem:revenue curve monotonicity}, we obtain $\LUnderBar\leq L \leq \LOverBar$ as stated in the constraint, since these two revenue curves $\revcurve_1, \revcurve_2$ sandwich the original revenue curve $\revcurve$.
        \item (Bounds for truncated GFT $M$) The argument is similar to the argument above for truncated GFT $L$. Recall that $M \triangleq \expect[\val]{\val\cdot \indicator{ \revratio/\quant \leq \val <\optreserve }}$. The concavity of the revenue curve implies that for every quantile $\quant\primed \in [\optquant, \quant]$,
        \begin{align*}
            \frac{\quant\primed - \optquant}{\quant - \optquant}\cdot (\revcurve(\optquant) - \revcurve(\quant)) + \revcurve(\quant)
            \leq 
            \revcurve(\quant\primed) 
            \leq
            \min\left\{ \revcurve(\optquant), 
            \frac{\quant\primed - \optquant}{\quant - \quant_0}\cdot (\revcurve(\optquant) - \revcurve(\quant)) + \revcurve(\quant)
            \right\}
        \end{align*}
        where both inequalities bind at $\quant\primed = \optquant$ and $\quant\primed = \quant$. The left-hand side and right-hand side can be viewed as two revenue curves $\revcurve_1, \revcurve_2$ that sandwich the original revenue curve $\revcurve$ (see blue and red revenue curves illustrated in \Cref{fig:GFT program:regular buyer}). It can be verified that $\MUnderBar$ and $\MOverBar$ are $\expect[\val]{\val\cdot \indicator{ \revratio/\quant \leq \val <\optreserve }}$ where the random value $\val$ is realized from valuation distribution induced by those two revenue curves $\revcurve_1, \revcurve_2$, respectively. Invoking \Cref{lem:revenue curve monotonicity}, we obtain $\MUnderBar\leq M \leq \MOverBar$ as stated in the constraint, since these two revenue curves $\revcurve_1, \revcurve_2$ sandwich the original revenue curve $\revcurve$.
        \item (Bounds for truncated GFT $H$) The argument is similar to the argument above for truncated GFT $L$. Recall that $H \triangleq \expect[\val]{\val\cdot \indicator{\val \geq \optreserve}}$. The concavity of the revenue curve and monopoly revenue $\revcurve(\optreserve) = 1$ imply that for every quantile $\quant\primed \in [0, \optquant]$,
        \begin{align*}
            \frac{\quant\primed}{\optquant}\cdot \revcurve(\optquant)
            \leq 
            \revcurve(\quant\primed) 
            \leq
             \revcurve(\optquant)
        \end{align*}
        where both inequalities bind at $\quant\primed = \optquant$. The left-hand side and right-hand side can be viewed as two revenue curves $\revcurve_1, \revcurve_2$ that sandwich the original revenue curve $\revcurve$ (see blue and red revenue curves illustrated in \Cref{fig:GFT program:regular buyer}). It can be verified that $\HUnderBar$ and $\HOverBar$ are $\expect[\val]{\val\cdot \indicator{  \val \geq\optreserve }}$ where the random value $\val$ is realized from valuation distribution induced by those two revenue curves $\revcurve_1, \revcurve_2$, respectively. Invoking \Cref{lem:revenue curve monotonicity}, we obtain $\HUnderBar\leq H \leq \HOverBar$ as stated in the constraint, since these two revenue curves $\revcurve_1, \revcurve_2$ sandwich the original revenue curve~$\revcurve$.
        \item (Equation for quantile $\quant_0$) Recall that  $\quant_0 \triangleq \quant + (1 - \revratio) / \revcurve'(\quant)$ and $\val_0 \triangleq \revcurve(\quant) + (1-\quant)\cdot \revcurve'(\quant)$. Combining both equations with $\revcurve(\quant) = \revratio$, we obtain $\quant_0 = 1 - \frac{(1 - \val_0)(1 - q)}{\revratio - \val_0}$ as stated in the optimization program.
    \end{itemize}
    Finally, we numerically evaluation the optimization program and obtain $\fixedPriceGFTPercentageRegular$. We present more details of this numerical evaluation in \Cref{apx:numerical evaluation:regular buyer}. This completes the proof of \Cref{lem:GFT program:regular buyer}.
\end{proof}

\begin{lemma}
\label{lem:revenue curve monotonicity}
    Given any two distributions $\buyerdist_1,\buyerdist_2$ and any two value $\val\primed, \val\doubleprimed$ with $\val\primed \leq \val\doubleprimed$. 
    Suppose $1-\buyercdf_1(\val\primed) = 1-\buyercdf_2(\val\primed)\triangleq \quant\primed$ and $1-\buyercdf_1(\val\doubleprimed) = 1-\buyercdf_2(\val\doubleprimed)\triangleq\quant\doubleprimed$. If the induced revenue curves $\revcurve_1,\revcurve_2$ satisfy that for every quantile $\quant \in[\quant\doubleprimed,\quant\primed]$,
    $\revcurve_1(\quant) \leq \revcurve_2(\quant)$,
    then 
    \begin{align*}
        \expect[\val\sim\buyerdist_1]{\val\cdot \indicator{\val\primed \leq \val\leq \val\doubleprimed}}
        \leq 
        \expect[\val\sim\buyerdist_2]{\val\cdot \indicator{\val\primed \leq \val\leq \val\doubleprimed}}
    \end{align*}
\end{lemma}
\begin{proof}
    Since for every quantile $\quant \in[\quant\doubleprimed,\quant\primed]$,
    $\revcurve_1(\quant) \leq \revcurve_2(\quant)$ and the inequality is binding at $\quant = \quant\primed$ and $\quant = \quant\doubleprimed$, it is guaranteed that for every value $\val\in[\val\primed,\val\doubleprimed]$, $\buyercdf_1(\val) \geq \buyercdf_2(\val)$
    and the inequality is binding at $\val = \val\primed$ and $\val = \val\doubleprimed$. This is sufficient to prove the lemma statement:
    \begin{align*}
        &\expect[\val\sim\buyerdist_1]{\val\cdot \indicator{\val\primed \leq \val\leq \val\doubleprimed}} 
        =
        \displaystyle\int_{\val\primed}^{\val\doubleprimed}
        \val\cdot \d \buyercdf_1(\val) 
        \overset{(a)}{=}
        \val\doubleprimed \buyercdf_1(\val\doubleprimed)
        -
        \val\primed\buyercdf_1(\val\primed)
        -
        \displaystyle\int_{\val\primed}^{\val\doubleprimed}
        \buyercdf_1(\val) \cdot \d \val
        \\
        &\qquad\overset{(b)}{\leq}
        \val\doubleprimed \buyercdf_2(\val\doubleprimed)
        -
        \val\primed\buyercdf_2(\val\primed)
        -
        \displaystyle\int_{\val\primed}^{\val\doubleprimed}
        \buyercdf_2(\val) \cdot \d \val
        \overset{(c)}{=}  
        \displaystyle\int_{\val\primed}^{\val\doubleprimed}
        \val\cdot \d \buyercdf_2(\val) 
        =
        \expect[\val\sim\buyerdist_2]{\val\cdot \indicator{\val\primed \leq \val\leq \val\doubleprimed}} 
    \end{align*}
    where equalities~(a) (c) hold due to the integration by parts, and inequality~(b) holds as we argued above. This complete the proof of \Cref{lem:revenue curve monotonicity}.
\end{proof}

\subsubsection{Negative Result for Zero-Value Seller and Buyer with a Regular Distribution}

In this part, we establish the negative result (\Cref{lem:GFT UB:regular buyer}) for a zero-value seller and a buyer with regular valuation distribution by analyzing \Cref{example:BROM:regular}.

\begin{lemma}
    \label{lem:GFT UB:regular buyer}
    In \Cref{example:BROM:regular} with $\constantH = 25$, any BIC, IIR, ex ante WBB mechanism $\mech$ that is {\ksfair} has GFT that is less than $\fixedPriceGFTPercentageUBRegular$ of the {\SecondBest} $\OPTSB$. Here $\fixedPriceGFTPercentageUBRegular$ is the numerical evaluated upper bound of the following expression with $\constantH = 25$:\footnote{The \emph{Lambert $W$ function}, also called the \emph{omega function} or \emph{product logarithm} in mathematics, is the converse relation of the function $y(x) = x\cdot e^x$. The principal branch $x = \LambertFunc(y)$ is the inverse relation when $x \geq -1$.} 
    \begin{align*}
    \frac{\constantH\ln\constantH- (\constantH - 1)\cdot \LambertFunc\left(\frac{\constantH^{\frac{\constantH}{\constantH - 1}}\ln (\constantH)}{\constantH-1}\right)}{\constantH\ln \constantH}\cdot \left(\frac{\constantH}{\constantH - 1}+\frac{1}{\ln\constantH}\right)
    \end{align*}
\end{lemma} 
The main difficulty to prove the lemma above is to establish such a GFT approximation bound for \emph{all} BIC, IIR, ex ante WBB, {\ksfair} mechanisms, which may be more complicated than a {\ksfair} {\FixPrice}. Loosely speaking, we first construct an upper bound of GFT (see \eqref{eq:fixed price mech suffices:regular UB}) for every BIC, IIR, ex ante WBB, and {\ksfair} mechanisms. We then argue that for this example, a {\ksfair} {\FixPrice} (which is DSIC, ex post IR, ex post SBB) optimizes this upper bound among all BIC, IIR, ex ante WBB, but (possibly not {\ksfair}) mechanisms.

\begin{proof}[Proof of \Cref{lem:GFT UB:regular buyer}]
Fix any BIC, IIR, ex ante WBB mechanism $\mech=(\alloc,\price,\sellerprice)$ that is {\ksfair}.  
Let $\revratio(\mech)$ be the ratio between the buyer's expected payment over the monopoly revenue, and $\residuesurplusratio(\mech)$ be the ratio between the buyer's ex ante utility and her expected value. Namely,
\begin{align*}
    \revratio(\mech) \triangleq  
    \frac{\expect[\val]{\price(\val,0)}}{\sellerbenchmark}
    =
    \frac{\sellerexanteutil(\mech) + \expect[\val]{\price(\val,0) - \sellerprice(\val,0)}}{\sellerbenchmark} 
    \;\;
    \mbox{and}
    \;\;
    \residuesurplusratio(\mech) \triangleq \frac{\buyerexanteutil(\mech)}{\buyerbenchmark}
\end{align*}
Since mechanism $\mech$ is ex ante WBB and {\ksfair}, $\revratio(\mech) \geq \residuesurplusratio(\mech)$. Consequently, the GFT of mechanism $\mech$ can be expressed as follows:
\begin{align*}
    \GFT{\mech} 
    &= \sellerexanteutil(\mech) + \buyerexanteutil(\mech) + \expect[\val]{\price(\val,0) - \sellerprice(\val,0)}
    \\
    &=
    \revratio(\mech) \cdot \sellerbenchmark
    +
    \residuesurplusratio(\mech) \cdot \buyerbenchmark
    \\
    &=
    \revratio(\mech) \cdot \sellerbenchmark
    +
    \min\{\revratio(\mech), \residuesurplusratio(\mech)\} \cdot \buyerbenchmark
\end{align*}
where the last equality holds since $\revratio(\mech) \geq \residuesurplusratio(\mech)$ argued above. Thus, the optimal GFT among all BIC, IIR, ex ante WBB, {\ksfair} mechanisms can be upper bounded as
\begin{align}
\nonumber
    \max_{\substack{\mech\in\mechfam:~\text{$\mech$ is {\ksfair}}}} \GFT{\mech}
    & =
    \max_{\substack{\mech\in\mechfam:~\text{$\mech$ is {\ksfair}}}}  \revratio(\mech) \cdot \sellerbenchmark
    +
    \min\{\revratio(\mech), \residuesurplusratio(\mech)\} \cdot \buyerbenchmark
    \\
\label{eq:fixed price mech suffices:regular UB}
    &\leq 
    \max_{\substack{\mech\in\mechfam}}  \revratio(\mech) \cdot \sellerbenchmark
    +
    \min\{\revratio(\mech), \residuesurplusratio(\mech)\} \cdot \buyerbenchmark
\end{align}
where we drop the {\ksfairness} requirement in the last step. 

Next we show that in \Cref{example:BROM:regular}, the optimal mechanism of the optimization program~\eqref{eq:fixed price mech suffices:regular UB} is a {\FixPrice}. Fix an arbitrary BIC, IIR, ex ante WBB mechanism $\mech\primed = (\alloc\primed,\price\primed,\sellerprice\primed)$. Let $\val\primed$ be the unique solution such that
\begin{align*}
    \displaystyle\int_0^{\val\primed} \alloc\primed(\val) \cdot \d\buyercdf(\val) 
    =
    \displaystyle\int_{\val\primed}^{\constantH} \left(1 - \alloc\primed(\val)\right)\cdot \d\buyercdf(\val)
\end{align*}
The existence and uniqueness of value $\val\primed$ hold, since mechanism $\mech\primed$ is BIC and thus interim allocation $\alloc\primed(\val)$ of the buyer is weakly increasing, and buyer's valuation distribution $\buyerdist$ has no point mass on $[0, \constantH)$. Now consider the {\FixPrice} $\mech\doubleprimed$ with trading price $\price\doubleprimed \triangleq \val\primed$. We claim that $\revratio(\mech\doubleprimed) \geq \revratio(\mech\primed)$ and $\residuesurplusratio(\mech\doubleprimed) \geq \residuesurplusratio(\mech\primed)$.
To see this, note that 
\begin{align*}
    \revratio(\mech\doubleprimed) &\overset{(a)}{=} \frac{1}{\sellerbenchmark}\cdot \expect[\val]{\virtualval(\val)\cdot \indicator{\val\geq \price\doubleprimed}}
    =
    \frac{1}{\sellerbenchmark}\cdot 
    \left(\expect[\val]{\virtualval(\val)\cdot \indicator{\val= \constantH}}
    +
    \displaystyle\int_{\val\primed}^{\constantH} \virtualval(\val)\cdot \d\buyercdf(\val)
    \right)
    \\
    &\overset{(b)}{=}
    \frac{1}{\sellerbenchmark}\cdot 
    \left(\expect[\val]{\virtualval(\val)\cdot \indicator{\val= \constantH}}
    +
    \displaystyle\int_{\val\primed}^{\constantH} \virtualval(\val)\alloc\primed(\val)\cdot \d\buyercdf(\val)
    +
    \displaystyle\int_0^{\val\primed} \virtualval(\val)(1-\alloc\primed(\val))\cdot \d\buyercdf(\val)
    \right)
    \\
    &\overset{(c)}{\geq}
    \revratio(\mech\primed)
\end{align*}
where equality~(a) holds due to \Cref{prop:revenue equivalence},
equality~(b) holds since virtual value $\virtualval(\val)$ is identical for every value $\val\in[0, \constantH)$ and the construction of value $\val\primed$,
and inequality~(c) holds due to \Cref{prop:revenue equivalence} and the fact that virtual value $\virtualval(\constantH) = \constantH > 0$. Similarly,
\begin{align*}
    \residuesurplusratio(\mech\doubleprimed) &\overset{(a)}{=} \frac{1}{\buyerbenchmark}\cdot \expect[\val]{\buyerhazardrate(\val)\cdot \indicator{\val\geq \price\doubleprimed}}
    =
    \frac{1}{\buyerbenchmark}\cdot 
    \left(\expect[\val]{\buyerhazardrate(\val)\cdot \indicator{\val= \constantH}}
    +
    \displaystyle\int_{\val\primed}^{\constantH} \buyerhazardrate(\val)\cdot \d\buyercdf(\val)
    \right)
    \\
    &\overset{(b)}{\geq}
    \frac{1}{\buyerbenchmark}\cdot 
    \left(\expect[\val]{\buyerhazardrate(\val)\cdot \indicator{\val= \constantH}}
    +
    \displaystyle\int_{\val\primed}^{\constantH} \buyerhazardrate(\val)\alloc\primed(\val)\cdot \d\buyercdf(\val)
    +
    \displaystyle\int_0^{\val\primed} \buyerhazardrate(\val)(1-\alloc\primed(\val))\cdot \d\buyercdf(\val)
    \right)
    \\
    &\overset{(c)}{=}
    \residuesurplusratio(\mech\primed)
\end{align*}
where equality~(a) holds due to \Cref{prop:buyer surplus equivalence},
inequality~(b) holds since hazard rate $\buyerhazardrate(\val)$ is increasing in $\val\in[0, \constantH)$ and the construction of value $\val\primed$,
and equality~(c) holds due to \Cref{prop:buyer surplus equivalence} and the fact that hazard rate $\buyerhazardrate(\constantH) = 0$.

Putting the two pieces together, we know that {\FixPrice} $\mech\doubleprimed$ has weakly higher objective value in program~\eqref{eq:fixed price mech suffices:regular UB}. Hence, the optimal mechanism of the optimization program~\eqref{eq:fixed price mech suffices:regular UB} is a {\FixPrice}. 

Finally, we evaluate the objective value of optimization program~\eqref{eq:fixed price mech suffices:regular UB} among {\FixPrice} $\mech$ with every trading price $\price\in[0, \constantH]$ for \Cref{example:BROM:regular} (see the blue curve in \Cref{fig:BROM:regular}). In particular, for every quantile $\quant\in[0, 1]$, {\FixPrice} with trading price $\price = \buyercdf^{-1}(1 - \quant)$ satisfies
\begin{align*}
    \revratio(\mech) = \frac{\constantH}{\constantH - 1}\cdot (1 - \quant)
    \;\;
    \mbox{and}
    \;\;
    \residuesurplusratio(\mech) = \frac{\ln\quant}{\ln \constantH} + 1 
\end{align*}
Combining with the fact that $\sellerbenchmark = 1$, $\buyerbenchmark = \frac{\constantH\ln\constantH}{\constantH - 1}$, we know the optimal objective value of program~\eqref{eq:fixed price mech suffices:regular UB} is obtained at\footnote{It can be verified that $\quant^*$ is also the quantile of the {\ksfair} trading price $\fprice$ and its induced objective value of program~\ref{eq:fixed price mech suffices:regular UB} is equal to the GFT approximation of the {\ksfair} {\FixPrice} with trading price $\fprice$.}
\begin{align}
\label{eq:BROM:regular:fair price}
    \quant^* = \frac{(\constantH - 1)\cdot \LambertFunc\left(\frac{\constantH^{\frac{\constantH}{\constantH-1}}\ln\constantH}{\constantH - 1}\right)}{\constantH\ln\constantH}
\end{align}
with objective value
\begin{align}
\label{eq:BROM:regular:fair price GFT approx}
    \frac{\constantH\ln\constantH - (\constantH - 1)\cdot \LambertFunc\left(\frac{\constantH^{\frac{\constantH}{\constantH-1}}\ln\constantH}{\constantH - 1}\right)}{\constantH\ln\constantH}
    \cdot \left(
    \frac{\constantH}{\constantH - 1
    }
    +\frac{1}{\ln\constantH}
    \right)
\end{align}
Setting $\constantH = 25$, we obtain $\text{\eqref{eq:BROM:regular:fair price GFT approx}} \leq \fixedPriceGFTPercentageUBRegular$ as the upper bound of the optimal GFT approximation among all BIC, IIR, ex ante WBB, {\ksfair} mechanisms. This completes the proof of \Cref{lem:GFT UB:regular buyer}.
\end{proof}

%% file: Figures/fig-GFT-regular-buyer.tex
\begin{tikzpicture}[scale=1, transform shape]
\begin{axis}[
axis line style=gray,
axis lines=middle,
xlabel = $\quant$,
xtick={0, 0.04, 0.35168, 1},
ytick={0, 0.8767514, 1},
xticklabels={0, $\frac{1}{\constantH}$, $\fquant$, 1},
yticklabels={0, $0.877$, $1$},
xmin=0,xmax=1.1,ymin=-0.0,ymax=1.1,
width=0.8\textwidth,
height=0.5\textwidth,
samples=500]

\addplot[black!100!white, line width=0.5mm] (0, 0) -- (0.04, 1) -- (1, 0);

\addplot[domain=0:0.04, dashed, line width=0.5mm] (x, {0});
% \addplot[domain=0.04:1, red, line width=0.5mm] (x, {(1 + 25 / 24 * (ln(x) + ln(25) - x + 0.04) - 25 / 24 * (1 - x) / 3.3529956509});
\addplot[domain=0.04:1, dashed, line width=0.5mm] (x, {(1 + 25 / 24 * (ln(x) + ln(25) - x + 0.04) - 25 / 24 * (1 - x)) / 3.3529956509});

\addplot[blue, line width=1.5mm] (0, 0) -- (0.04, 0.2982407);
\addplot[domain=0.04:0.35168, blue, line width=1.5mm] (x, {(1 + 25 / 24 * (ln(x) + ln(25) - x + 0.04)) / 3.3529956509});
\addplot[blue, line width=0.5mm] (0.35168, 0.8767514) -- (1, 0);

\addplot[red, line width=0.5mm] (0, 0) -- (0.04, 0.2982407);
\addplot[domain=0.04:1, red, line width=0.5mm] (x, {(1 + 25 / 24 * (ln(x) + ln(25) - x + 0.04)) / 3.3529956509});

\addplot[dotted, gray, line width=0.3mm] (0.35168, 0) -- (0.35168, 0.8767514) -- (0, 0.8767514);

\addplot[dotted, gray, line width=0.3mm] (1, 0) -- (1, 1) -- (0, 1);

\end{axis}

\end{tikzpicture}

%% file: Figures/fig-base-model-regular-revenue-curve-analysis.tex
\begin{tikzpicture}[scale=1, transform shape]
\begin{axis}[
axis line style=gray,
axis lines=middle,
% xlabel = $\quant$,
% ylabel = $\revcurve$,
xtick={0, 0.5, 0.625, 0.75, 1},
ytick={0, 0.25, 0.75, 1},
xticklabels={0, $\optquant$, $\quant_0$, $\quant$,  1},
yticklabels={0, $\val_0$, $\alpha$, $1$},
xmin=0,xmax=1.1,ymin=-0.0,ymax=1.5,
width=0.8\textwidth,
height=0.5\textwidth,
samples=500]

% \fill[blue!20] (0,0) -- (0.642857, 0.229591877551) -- (0.725, 0.199375) -- cycle;

\addplot[domain=0:0.625, red, line width=0.5mm] (x, {1});
\addplot[domain=0.625:1, red, line width=0.5mm] (x, {-2*(x-0.75)+0.75});
\addplot[dotted, gray, line width=0.3mm] (0.625, 0) -- (0.625, 1);

\addplot[line width=0.5mm, blue] (0, 0) -- (0.5, 1) -- (0.75, 0.75) -- (1, 0);

\addplot[domain=0:1, black!100!white, line width=0.5mm] (x, {x * (1-x) / 0.25}); 

\addplot[dotted, gray, line width=0.3mm] (1, 0) -- (1, 0.25) -- (0, 0.25);

% \addplot[line width=0.5mm, dashed] (0, 0) -- (0.6, 1.2);
\addplot[dotted, gray, line width=0.3mm] (0.5, 0) -- (0.5, 1) -- (0, 1);
% \draw (0.63, 1.25) node {$\optreserve$};
\addplot[dotted, gray, line width=0.3mm] (0.5, 1) -- (1, 1);

% \addplot[line width=0.5mm, dashed] (0, 0) -- (0.9, 0.9);
\addplot[dotted, gray, line width=0.3mm] (0.75, 0) -- (0.75, 0.75) -- (0, 0.75);
% \draw (0.93, 0.92) node {$\price$};

% \addplot[line width=0.5mm, dashed] (0, 0) -- (0.7, 0.25);

% \addplot[dotted, gray, line width=0.3mm] (0.725, 0) -- (0.725, 0.199375);% -- (0, 0.199375);

% \addplot[dotted, gray, line width=0.3mm] (0.642857, 0.176785675) -- (0, 0.176785675);

% \addplot[dotted, gray, line width=0.3mm] (0.642857, 0) -- (0.642857, 0.229591877551) -- (0, 0.229591877551);

\draw[decorate,decoration={brace,amplitude=8pt}] (0.0,1.3) -- (0.5,1.3) node[midway,above=6pt] {$H$/{{\color{red}$\HOverBar$}}/{{\color{blue}$\HUnderBar$}}};
\draw[decorate,decoration={brace,amplitude=8pt}] (0.5,1.3) -- (0.75,1.3) node[midway,above=6pt] {$M$/{{\color{red}$\MOverBar$}}/{{\color{blue}$\MUnderBar$}}};
\draw[decorate,decoration={brace,amplitude=8pt}] (0.75,1.3) -- (1,1.3) node[midway,above=6pt] {$L$/{{\color{red}$\LOverBar$}}/{{\color{blue}$\LUnderBar$}}};
\addplot[dotted, gray, line width=0.3mm] (0.5, 1) -- (0.5, 1.3);
\addplot[dotted, gray, line width=0.3mm] (0.75, 0.75) -- (0.75, 1.3);
\addplot[dotted, gray, line width=0.3mm] (1, 0.25) -- (1, 1.3);

\end{axis}

\end{tikzpicture}

%% file: Paper/zero-seller-mhr-buyer.tex
In this section we consider zero-value seller settings when the buyer's valuation distribution is assumed to be MHR. With this stronger assumption that the distribution is MHR rather than regular, we are able to improved the GFT approximation and show that the approximation for \emph{all} instances with MHR distributions is strictly larger than the approximation guaranteed by regularity alone.

\begin{restatable}{theorem}{thmGFTmhrbuyer}
\label{thm:improved GFT:mhr buyer} 
    For every bilateral trade instance of a zero-value seller and buyer with MHR valuation distribution $\buyerdist$, there exists price $\fprice$ smaller than the monopoly reserve $\optreserve$ for distribution $\buyerdist$,\footnote{When distribution $\buyerdist$ is MHR, its induced monopoly reserve $\optreserve$ is unique.} 
    such that the {\FixPrice} mechanism with trading price $\fprice$ is {\ksfair} and its GFT is at least $\calC$ fraction of the {\SecondBest} $\OPTSB$. Here $\calC$ is the solution to minimization program~\ref{program:GFT:mhr buyer},  
    and we observe that 
    $\calC\geq \fixedPriceGFTPercentageMHR$ by a numerical computation.

    Moreover, there exists an instance of a zero-value seller and buyer with MHR valuation distribution, in which any BIC, IIR, ex ante WBB mechanism that is {\ksfair} has GFT that is less than  $\fixedPriceGFTPercentageUBMHR$ of the {\SecondBest} $\OPTSB$.
\end{restatable}

\Cref{thm:improved GFT:mhr buyer} directly follows \Cref{lem:GFT program:mhr buyer,lem:GFT UB:mhr buyer} which are presented below. Their proofs are similar to the ones for regular distributions in \Cref{subsec:improved GFT:regular buyer}.

\subsubsection{GFT Approximation by KS-Fair {\FixPrice}}

In this section we show that {\FixPrice} that is {\ksfair} can get very good GFT approximation. We characterize the GFT approximation of a {\ksfair} {\FixPrice} in the following lemma. Its proof is similar to the one for \Cref{lem:GFT program:regular buyer} in \Cref{subsec:improved GFT:regular buyer}. Specifically, we search over all MHR distributions and argue that the worst GFT can only be induced by a subclass of them, whose \emph{cumulative hazard rate function} (rather than the revenue curve, as was the case for regular buyer distributions) can be characterized by finite many parameters (see \Cref{fig:GFT program:mhr buyer}). We defer its formal proof to \Cref{apx:lemgftprogrammhrbuyer}.

\begin{figure}
    \centering
    \input{Figures/fig-base-model-mhr-cumhazard-analysis}
    \caption{Graphical illustration of the analysis for \Cref{lem:GFT program:mhr buyer}. The black curve is the convex cumulative hazard rate function $\cumhazard$ of the buyer. The red and blue cumulative hazard rate functions $\cumhazard_1, \cumhazard_2$ (defined in the analysis) sandwich the original cumulative hazard rate function $\cumhazard$. The black dashed line is $\ln(\val)$, which touches the cumulative hazard rate function $\cumhazard$ at monopoly reserve $\optreserve$.}
    \label{fig:GFT program:mhr buyer}
\end{figure}
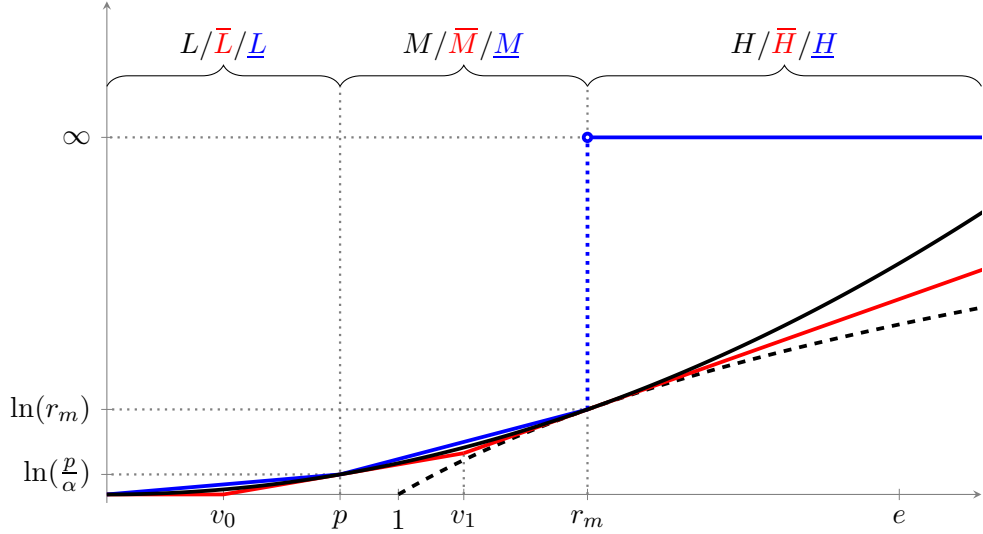

\begin{restatable}{lemma}{lemgftprogrammhrbuyer}
\label{lem:GFT program:mhr buyer}
    For every bilateral trade instance of a zero-value seller and buyer with MHR valuation distribution, there exists price $\fprice$ smaller than the monopoly reserve $\optreserve$ such that the {\FixPrice} with trading price $\fprice$ is {\ksfair} and its GFT is a $\GFTapprox$-approximation to the {\SecondBest} $\OPTSB$. Here $\calC$ is the solution to the following program~\ref{program:GFT:mhr buyer}:
    \begin{align}
    \label{program:GFT:mhr buyer}
    \tag{$\mathcal{P}_{\mathrm{MHR}}$}
    \arraycolsep=5.4pt\def\arraystretch{1}
        \begin{array}{llll}
          \calC \triangleq~~&\min\limits_{\optreserve, H}
          ~\max\limits_{\revratio}
          ~\min\limits_{\substack{\price, \val_0, M, L}}   & 
          \displaystyle\exanteutilratio + \frac{\exanteutilratio}{H + M + L} &
          \vspace{10pt}
          \\
          \vspace{10pt}
          &\text{s.t.}
          & \optreserve\in[1, e]~,  & 
          \\
          \vspace{10pt}
          && \revratio\in(0, 1)~,  & 
          \\ 
          \vspace{10pt}
          && \displaystyle\price\in \left[\max\left\{
    -\optreserve\cdot \LambertFunc\left(-\frac{\revratio}{e}\right), \revratio\right\},
    - \optreserve\cdot \LambertFunc\left(-\displaystyle\frac{\revratio\ln(\optreserve)}{\optreserve}\right)\cdot \frac{1}{\ln(\optreserve)} \right]~,  & 
          \\
          \vspace{10pt}
          && \displaystyle\val_0\in \left[0, \optreserve - \frac{\ln(\optreserve)}{\ln(\optreserve) - \ln\left(\frac{\price}{\revratio}\right)}\cdot (\optreserve - \price)\right]~,  & 
          \\
          \vspace{10pt}
          && H\in[\HUnderBar, \HOverBar]~, 
          M\in[\MUnderBar, \MOverBar]~,
          L\in[\LUnderBar, \LOverBar]~,& 
        \end{array}
    \end{align}
    where $\exanteutilratio$, $\val_1$, $\HUnderBar, \HOverBar, \MUnderBar, \MOverBar, \LUnderBar, \LOverBar$ are auxiliary variables constructed as 
    \begin{align*}
        \exanteutilratio &\triangleq  \displaystyle
        \revratio - \plus{\frac{H + M + L}{H + M + L + 1}\left(\revratio - \frac{H + M - \revratio}{H + M + L}\right)}~,
        \\
        \val_1 &\triangleq \left(\frac{1}{\optreserve} - \frac{1}{\price - \val_0}\ln\left(\frac{\price}{\revratio}\right)\right)^{-1}\left(
        \ln\left(\frac{\price}{\revratio}\right) - \ln(\optreserve)
        + 1 
        - \frac{\price}{\price-\val_0}\ln\left(\frac{\price}{\revratio}\right) 
        \right)~,
        \\
        \HUnderBar &\triangleq 1~, \HOverBar \triangleq 2~, 
        \\
        \MUnderBar &\triangleq \left(\price\cdot \optreserve\cdot \ln\left(\frac{\revratio\optreserve}{\price}\right)\right)^{-1}(\optreserve-\price)(\revratio\cdot \optreserve-\price) - 1 + \revratio ~,
        \\
        \MOverBar &\triangleq  \left({\ln\left(\frac{\price}{\revratio}\right)}\right)^{-1}\left(\left(\frac{\price}{\revratio}\right)^{\frac{\val_1 - \price}{\price - \val_0}} - 1\right)(\price - \val_0) + e^{1 - \frac{\val_1}{\optreserve}} - 2 + \revratio~,
        \\
        \LUnderBar &\triangleq \left({\ln\left(\frac{\price}{\revratio}\right)}\right)^{-1}(\price - \revratio) - \revratio~,
        % \\
        \LOverBar \triangleq \left({\ln\left(\frac{\price}{\revratio}\right)}\right)^{-1}\left(1 - \frac{\revratio}{\price}\right)(\price - \val_0) + \val_0 - \revratio ~.
    \end{align*}
    By a numerical computation (see \Cref{apx:numerical evaluation:mhr buyer}), we observe that 
    $\calC\geq \fixedPriceGFTPercentageMHR$.
\end{restatable}

\subsubsection{Negative Result for Zero-Value Seller and a Buyer with MHR Distribution}

In this section, we establish the negative result (\Cref{lem:GFT UB:mhr buyer}) for a zero-value seller and a buyer with MHR valuation distribution by analyzing \Cref{example:all fair:mhr buyer}. 

\begin{example}
\label{example:all fair:mhr buyer}
The buyer has an MHR valuation distribution $\buyerdist$, which has support $[0, e]$ and cumulative density function $\buyercdf(\val) = 1 - e^{-\val/e}$ for every $\val\in[0, e]$ and $\buyercdf(\val) = 1$ for $\val \in(e, \infty)$. (Namely, there is an atom at $e$ with probability mass of $\frac{1}{e}$.) The seller has a deterministic value of 0. See \Cref{fig:all fair:mhr buyer} for an illustration. 
\end{example}

\begin{figure}[ht]
    \centering
    \input{Figures/fig-GFT-mhr-buyer}
    \caption{Graphical illustration of \Cref{example:all fair:mhr buyer}. The x-axis is value $\val$. 
    Consider {\FixPrice} $\mech$ with trading price $\price$ for every $\price\in[0,e]$. The black solid (resp.\ dashed) curve also represents   ${\sellerexanteutil(\mech)}/{\sellerbenchmark}$ (resp.\ ${\buyerexanteutil(\mech)}/{\buyerbenchmark}$) for the seller (resp.\ buyer). The red curve is the GFT approximation ratio ${\GFT{\mech}}/{\OPTSB}$. 
    A {\ksfair} {\FixPrice} is achieved at trading price $\fprice \approx 0.80$ with GFT approximation ratio of $0.944$. Finally, the blue curve is used in the proof of the negative result in \Cref{lem:GFT UB:mhr buyer}.}
    \label{fig:all fair:mhr buyer}
\end{figure}

\begin{lemma}
    \label{lem:GFT UB:mhr buyer}
    In \Cref{example:all fair:mhr buyer}, any BIC, IIR, ex ante WBB mechanism $\mech$ that is {\ksfair} has GFT that is less than $\fixedPriceGFTPercentageUBMHR$ of the {\SecondBest} $\OPTSB$. Here $\fixedPriceGFTPercentageUBMHR$ is the numerical evaluated upper bound of 
    \begin{align*}
        \max_{\price\in[0, e]}\frac{1}{e - 1}\cdot \left(\price\cdot e^{-\frac{\price}{e}} + \min\{(e-1)\cdot \price\cdot e^{-\frac{\price}{e}},e^{1 - \frac{\price}{e}} - 1\}\right)
    \end{align*}
\end{lemma}
The proof of \Cref{lem:GFT UB:mhr buyer} is similar to the one for \Cref{lem:GFT UB:regular buyer} in \Cref{subsec:improved GFT:regular buyer}. We first construct an upper bound of GFT (see \eqref{eq:fixed price mech suffices:mhr UB}) for every BIC, IIR, ex ante WBB, and {\ksfair} mechanisms. We then argue that in \Cref{example:all fair:mhr buyer}, a \emph{\ksfair} {\FixPrice} (which is DSIC, ex post IR, ex post SBB) optimizes this upper bound among all BIC, IIR, ex ante WBB, but (possibly not {\ksfair}) mechanisms.
\begin{proof}[Proof of \Cref{lem:GFT UB:mhr buyer}]
Fix any BIC, IIR, ex ante WBB mechanism $\mech =(\alloc,\price,\sellerprice)$ that is {\ksfair}. Let $\revratio(\mech)$ be the ratio between the buyer's expected payment over the monopoly revenue, and $\residuesurplusratio(\mech)$ be the ratio between the buyer's ex ante utility and her expected value. Namely,
\begin{align*}
    \revratio(\mech) \triangleq  
    \frac{\expect[\val]{\price(\val,0)}}{\sellerbenchmark}
    =
    \frac{\sellerexanteutil(\mech) + \expect[\val]{\price(\val,0) - \sellerprice(\val,0)}}{\sellerbenchmark} 
    \;\;
    \mbox{and}
    \;\;
    \residuesurplusratio(\mech) \triangleq \frac{\buyerexanteutil(\mech)}{\buyerbenchmark}
\end{align*}
Since mechanism $\mech$ is ex ante WBB and {\ksfair}, $\revratio(\mech) \geq \residuesurplusratio(\mech)$. Consequently, the GFT of mechanism $\mech$ can be expressed as follows:
\begin{align*}
    \GFT{\mech} 
    &= \sellerexanteutil(\mech) + \buyerexanteutil(\mech) + \expect[\val]{\price(\val,0) - \sellerprice(\val,0)}
    \\
    &=
    \revratio(\mech) \cdot \sellerbenchmark
    +
    \residuesurplusratio(\mech) \cdot \buyerbenchmark
    \\
    &=
    \revratio(\mech) \cdot \sellerbenchmark
    +
    \min\{\revratio(\mech), \residuesurplusratio(\mech)\} \cdot \buyerbenchmark
\end{align*}
where the last equality holds since $\revratio(\mech) \geq \residuesurplusratio(\mech)$ argued above. Thus, the optimal GFT among all BIC, IIR, ex ante WBB, {\ksfair} mechanisms can be upper bounded as
\begin{align}
\nonumber
    \max_{\substack{\mech\in\mechfam:~\text{$\mech$ is {\ksfair}}}} \GFT{\mech}
    & =
    \max_{\substack{\mech\in\mechfam:~\text{$\mech$ is {\ksfair}}}}  \revratio(\mech) \cdot \sellerbenchmark
    +
    \min\{\revratio(\mech), \residuesurplusratio(\mech)\} \cdot \buyerbenchmark
    \\
\label{eq:fixed price mech suffices:mhr UB}
    &\leq 
    \max_{\substack{\mech\in\mechfam}}  \revratio(\mech) \cdot \sellerbenchmark
    +
    \min\{\revratio(\mech), \residuesurplusratio(\mech)\} \cdot \buyerbenchmark
\end{align}
where we drop the {\ksfairness} requirement in the last step. 

Next we show that in \Cref{example:all fair:mhr buyer}, the optimal mechanism of the optimization program~\eqref{eq:fixed price mech suffices:mhr UB} is a {\FixPrice}. Fix an arbitrary BIC, IIR, ex ante WBB mechanism $\mech\primed = (\alloc\primed,\price\primed,\sellerprice\primed)$. Let $\val\primed$ be the unique solution such that
\begin{align*}
    \displaystyle\int_0^{\val\primed} \alloc\primed(\val) \cdot \d\buyercdf(\val) 
    =
    \displaystyle\int_{\val\primed}^{e} \left(1 - \alloc\primed(\val)\right)\cdot \d\buyercdf(\val)
\end{align*}
The existence and uniqueness of value $\val\primed$ hold, since mechanism $\mech\primed$ is BIC and thus interim allocation $\alloc\primed(\val)$ of the buyer is weakly increasing, and buyer's valuation distribution $\buyerdist$ has no point mass on $[0, e)$. Now consider the {\FixPrice} $\mech\doubleprimed$ with trading price $\price\doubleprimed \triangleq \val\primed$. We claim that $\revratio(\mech\doubleprimed) \geq \revratio(\mech\primed)$ and $\residuesurplusratio(\mech\doubleprimed) = \residuesurplusratio(\mech\primed)$.
To see this, note that 
\begin{align*}
    \revratio(\mech\doubleprimed) &\overset{(a)}{=} \frac{1}{\sellerbenchmark}\cdot \expect[\val]{\virtualval(\val)\cdot \indicator{\val\geq \price\doubleprimed}}
    =
    \frac{1}{\sellerbenchmark}\cdot 
    \left(\expect[\val]{\virtualval(\val)\cdot \indicator{\val=e}}
    +
    \displaystyle\int_{\val\primed}^{e} \virtualval(\val)\cdot \d\buyercdf(\val)
    \right)
    \\
    &\overset{(b)}{\geq}
    \frac{1}{\sellerbenchmark}\cdot 
    \left(\expect[\val]{\virtualval(\val)\cdot \indicator{\val= e}}
    +
    \displaystyle\int_{\val\primed}^{e} \virtualval(\val)\alloc\primed(\val)\cdot \d\buyercdf(\val)
    +
    \displaystyle\int_0^{\val\primed} \virtualval(\val)(1-\alloc\primed(\val))\cdot \d\buyercdf(\val)
    \right)
    \\
    &\overset{(c)}{\geq}
    \revratio(\mech\primed)
\end{align*}
where equality~(a) holds due to \Cref{prop:revenue equivalence},
inequality~(b) holds since virtual value $\virtualval(\val)$ is strictly increasing in $\val\in[0, e)$ and the construction of value $\val\primed$,
and inequality~(c) holds due to \Cref{prop:revenue equivalence} and the fact that virtual value $\virtualval(e) = e > 0$. Similarly,
\begin{align*}
    \residuesurplusratio(\mech\doubleprimed) &\overset{(a)}{=} \frac{1}{\buyerbenchmark}\cdot \expect[\val]{\buyerhazardrate(\val)\cdot \indicator{\val\geq \price\doubleprimed}}
    =
    \frac{1}{\buyerbenchmark}\cdot 
    \left(\expect[\val]{\buyerhazardrate(\val)\cdot \indicator{\val= e}}
    +
    \displaystyle\int_{\val\primed}^{e} \buyerhazardrate(\val)\cdot \d\buyercdf(\val)
    \right)
    \\
    &\overset{(b)}{=}
    \frac{1}{\buyerbenchmark}\cdot 
    \left(\expect[\val]{\buyerhazardrate(\val)\cdot \indicator{\val= e}}
    +
    \displaystyle\int_{\val\primed}^{e} \buyerhazardrate(\val)\alloc\primed(\val)\cdot \d\buyercdf(\val)
    +
    \displaystyle\int_0^{\val\primed} \buyerhazardrate(\val)(1-\alloc\primed(\val))\cdot \d\buyercdf(\val)
    \right)
    \\
    &\overset{(c)}{=}
    \residuesurplusratio(\mech\primed)
\end{align*}
where equality~(a) holds due to \Cref{prop:buyer surplus equivalence},
equality~(b) holds since hazard rate $\buyerhazardrate(\val)$ is identical for every $\val\in[0, e)$ and the construction of value $\val\primed$,
and equality~(c) holds due to \Cref{prop:buyer surplus equivalence} and the fact that hazard rate $\buyerhazardrate(e) = 0$.

Putting the two pieces together, we know that {\FixPrice} $\mech\doubleprimed$ has weakly higher objective value in program~\eqref{eq:fixed price mech suffices:mhr UB}. Hence, the optimal mechanism of the optimization program~\eqref{eq:fixed price mech suffices:mhr UB} is a {\FixPrice}.

Finally, we evaluate the objective value of optimization program~\eqref{eq:fixed price mech suffices:mhr UB} among {\FixPrice} $\mech$ with every trading price $\price\in[0, e]$ for \Cref{example:all fair:mhr buyer} (see the blue curve in \Cref{fig:all fair:mhr buyer}). In particular, {\FixPrice} with every trading price $\price \in[0, e]$ satisfies
\begin{align*}
    \revratio(\mech) = \price \cdot e^{-\frac{\price}{e}}
    \;\;
    \mbox{and}
    \;\;
    \residuesurplusratio(\mech) = \frac{e^{1-\frac{\price}{e}} - 1}{e - 1} 
\end{align*}
Combining with the fact that $\sellerbenchmark = 1$, $\buyerbenchmark = e-1$, we know the optimal objective value of program~\eqref{eq:fixed price mech suffices:mhr UB} is upper bounded by $\fixedPriceGFTPercentageUBMHR$, with the maximum value obtained at $\price^*\approx 0.80066$. This completes the proof of \Cref{lem:GFT UB:mhr buyer}.
\end{proof}

%% file: Figures/fig-base-model-mhr-cumhazard-analysis.tex
\begin{tikzpicture}[scale=1, transform shape]
\begin{axis}[
axis line style=gray,
axis lines=middle,
xtick={0, 0.4, 0.8, 1, 1.22436, 1.64872, 2.7182818},
ytick={0, 0.1177216, 0.5, 2.1},
xticklabels={0, $\val_0$, $\price$, 1, $\val_1$, $\optreserve$, $e$},
yticklabels={0, $\ln(\frac{\price}{\revratio})$, $\ln(\optreserve)$, $\infty$},
xmin=0,xmax=3.0,ymin=-0.05,ymax=2.9,
width=0.8\textwidth,
height=0.5\textwidth,
samples=500]

% \fill[blue!20] (0,0) -- (0.642857, 0.229591877551) -- (0.725, 0.199375) -- cycle;

\addplot[domain=1:5, dashed, line width=0.5mm] (x, {ln(x)});

\addplot[dotted, gray, line width=0.3mm] (1.64872, 0) -- (1.64872, 0.5) -- (0, 0.5);

\addplot[dotted, gray, line width=0.3mm] (0.8, 0) -- (0.8, 2.4);

\addplot[dotted, gray, line width=0.3mm] (1.64872, 2.1) -- (1.64872, 2.4);
\addplot[dotted, gray, line width=0.3mm] (0.8, 0.1177216) -- (0., 0.1177216);

\addplot[dotted, gray, line width=0.3mm] (1.22436, 0) -- (1.22436, 0.242612);

\addplot[dotted, gray, line width=0.3mm] (1.64872, 2.1) -- (0, 2.1);

\addplot[red, line width=0.5mm] (0, 0) -- (0.4, 0) -- (1.22436, 0.242612);
\addplot[domain=1.22436:5, red, line width=0.5mm] (x, {0.5 - 1 / 1.64872 * (1.64872 - x)});

\addplot[blue, line width=0.5mm] (0, 0) -- (0.8, 0.1177216) -- (1.64872, 0.5);
\addplot[blue, dotted, line width=0.5mm] (1.64872, 0.5) --(1.64872, 2.1);
\addplot[blue, line width=0.5mm] (1.64872, 2.1) --(5, 2.1);
\draw[blue, fill=white, line width=0.5mm] (axis cs:1.64872, 2.1) circle[radius=0.05cm];

\addplot[domain=0:5, black!100!white, line width=0.5mm] (x, {x * x * 0.18394});

\draw[decorate,decoration={brace,amplitude=8pt}] (0.0,2.4) -- (0.8,2.4) node[midway,above=6pt] {$L$/{{\color{red}$\LOverBar$}}/{{\color{blue}$\LUnderBar$}}};
\draw[decorate,decoration={brace,amplitude=8pt}] (0.8,2.4) -- (1.64872,2.4) node[midway,above=6pt] {$M$/{{\color{red}$\MOverBar$}}/{{\color{blue}$\MUnderBar$}}};
\draw[decorate,decoration={brace,amplitude=8pt}] (1.64872,2.4) -- (3,2.4) node[midway,above=6pt] {$H$/{{\color{red}$\HOverBar$}}/{{\color{blue}$\HUnderBar$}}};

\end{axis}

\end{tikzpicture}

%% file: Figures/fig-GFT-mhr-buyer.tex
\begin{tikzpicture}[scale=1, transform shape]
\begin{axis}[
axis line style=gray,
axis lines=middle,
xlabel = {$\val$},
xtick={0, 0.800662, 2.7182818284},
ytick={0, 0.943472496484, 1},
xticklabels={0, $\fprice$, ~~~~~~$\optreserve = e$},
yticklabels={0, 0.944, $1$},
xmin=0,xmax=3.0,ymin=-0.05,ymax=1.05,
width=0.8\textwidth,
height=0.5\textwidth,
samples=500]

\addplot[dotted, gray, line width=0.3mm] (2.7182818284, 0) -- (2.7182818284, 1) -- (0, 1);
\addplot[dotted, gray, line width=0.3mm] (0.800662, 0) -- (0.800662, 0.943472496484) -- (0, 0.943472496484);

\addplot[domain=0:2.7182818284, black!100!white, line width=0.5mm] (x, {x*(exp(-x/exp(1)))});

\addplot[domain=0:2.7182818284, dashed, black!100!white, line width=0.5mm] (x, {(1+(exp(-x/exp(1))*(x+exp(1)) - 2)-x*(exp(-x/exp(1))))/1.71828182846});

\addplot[domain=0.800662:2.7182818284, blue, line width=1.5mm] (x, {(1+(exp(-x/exp(1))*(x+exp(1)) - 2)-0*x*(exp(-x/exp(1))))/1.71828182846});
\addplot[domain=0:0.800662, blue, line width=0.5mm] (x, {((1+1.71828182846)*x*(exp(-x/exp(1))))/1.71828182846});

\addplot[domain=0:2.7182818284, red, line width=0.5mm] (x, {(1+(exp(-x/exp(1))*(x+exp(1)) - 2)-0*x*(exp(-x/exp(1))))/1.71828182846});

\end{axis}

\end{tikzpicture}

%% file: Paper/NSWM-result.tex
In this section we present tight results regarding the fraction of the {\SecondBest} that is obtained when a mechanism (that is  BIC, IIR, and ex ante WBB) maximizes the ex ante \emph{Nash social welfare (NSW)}. The NSW is defined as the product of the two traders' ex-ante utilities. We refer to such a mechanism as a \emph{\NashSocialWelfareMaximizer} ({\NSWM}).

\begin{definition}
    For bilateral trade instance $\btinstance$, mechanism $\mech\in \mechfam$ is called an {\NashSocialWelfareMaximizer} if it maximizes the ex ante Nash social welfare among all BIC, IIR, ex-ante WBB mechanisms, i.e.,
    \begin{align*}
         \mech
         \in \argmax\limits_{\mech\primed\in\mechfam}~\sellerexanteutil(\mech\primed)\cdot \buyerexanteutil(\mech\primed)
    \end{align*}
\end{definition}

There might be different {\NashSocialWelfareMaximizers} (depending on how ties are broken) --- our results will hold for all of them. 

\begin{theorem}
\label{thm:NSWM GFT}
For every bilateral trade instance, the GFT of every {\NashSocialWelfareMaximizer} is at least $\frac{1}{2}$ fraction of the {\SecondBest} $\OPTSB$.
    
Moreover, for any $\eps>0$ there exists a distribution $\buyerdist$ (that is not regular) such that in the bilateral setting with a zero-value seller and a buyer with valuation distribution $\buyerdist$, every {\NashSocialWelfareMaximizer} obtains at most $\frac{1}{2}+\eps$ fraction of the {\SecondBest} $\OPTSB$.
\end{theorem}

\Cref{thm:NSWM GFT} directly follows \Cref{lem:NSWM trader approx,lem:NSWM GFT UB:general instance} which prove the positive result and negative result in the theorem statement, respectively. 

\subsection{GFT-Approximation by {\NashSocialWelfareMaximizer}}

We characterize the GFT approximation of the {\NashSocialWelfareMaximizer} as follows.

\begin{lemma}
\label{lem:NSWM trader approx}
    For every bilateral trade instance, the ex ante utility of each trader in every {\NashSocialWelfareMaximizer} $\mech$ is at least $\frac{1}{2}$ fraction of her benchmark, i.e.,
    \begin{align*}
        \sellerexanteutil(\mech) \geq \frac{1}{2} \cdot \sellerbenchmark
        \;\;\mbox{and}\;\;
        \buyerexanteutil(\mech) \geq \frac{1}{2} \cdot \buyerbenchmark  
    \end{align*}
    Therefore, the GFT of {\NashSocialWelfareMaximizer} $\mech$ is at least $\frac{1}{2}$ fraction of th {\SecondBest}. 
\end{lemma}
\begin{proof}
    Fix any {\NashSocialWelfareMaximizer} $\mech$. 
    We prove the lemma statement by a contradiction argument. 
    Suppose $\frac{\sellerexanteutil(\mech)}{\sellerbenchmark} < \frac{1}{2}$ for the seller. (The other case for the buyer follows a symmetric argument.) Define auxiliary notation $\revratio \triangleq \frac{\sellerexanteutil(\mech)}{\sellerbenchmark} < \frac{1}{2}$.

    Consider a new mechanism $\mech\primed$ that runs the {\SellerOffer} with probability $\frac{1}{2} - \frac{\revratio}{2 - 2 \revratio}$, and runs mechanism $\mech$ with probability $\frac{1}{2} + \frac{\revratio}{2 - 2 \revratio}$. Since $\revratio < \frac{1}{2}$, mechanism $\mech\primed$ is well-defined. By construction, the ex ante utility of each trader in mechanism $\mech\primed$ can be computed as  
    \begin{align*}
        \sellerexanteutil(\mech\primed) &{} = \left(\frac{1}{2} - \frac{\revratio}{2 - 2 \revratio}\right) \cdot \sellerbenchmark + \left(\frac{1}{2} + \frac{\revratio}{2 - 2 \revratio}\right) \cdot \sellerexanteutil(\mech)
        \\
        &{} \overset{(a)}{=} 
        \left(\frac{1}{2} - \frac{\revratio}{2 - 2 \revratio}\right) \cdot \frac{1}{\revratio}\cdot \sellerexanteutil(\mech) + \left(\frac{1}{2} + \frac{\revratio}{2 - 2 \revratio}\right) \cdot \sellerexanteutil(\mech)    
        \\
        \buyerexanteutil(\mech\primed) &{} \geq 
        \left(\frac{1}{2} + \frac{\revratio}{2 - 2 \revratio}\right) \cdot \buyerexanteutil(\mech)   
    \end{align*}
    where equality~(a) holds since $\revratio = \frac{\sellerexanteutil(\mech)}{\sellerbenchmark}$. Thus, the Nash social welfare of mechanism $\mech\primed$ is 
    \begin{align*}
        &\left(\left(\frac{1}{2} - \frac{\revratio}{2 - 2 \revratio}\right) \cdot \frac{1}{\revratio} + \left(\frac{1}{2} + \frac{\revratio}{2 - 2 \revratio}\right)
        \right)
        \cdot 
        \sellerexanteutil(\mech)
        \cdot 
        \left(\frac{1}{2} + \frac{\revratio}{2 - 2 \revratio}\right) \cdot \buyerexanteutil(\mech)
        \\
        ={} &
        \frac{1}{4(1-\revratio)\revratio}
        \cdot \sellerexanteutil(\mech) \cdot 
        \buyerexanteutil(\mech)
        > {} 
        \sellerexanteutil(\mech) \cdot 
        \buyerexanteutil(\mech)
    \end{align*}
    where the equality holds by algebra and the strict inequality holds since $\revratio < \frac{1}{2}$. The Nash social welfare of mechanism $\mech\primed$ is strictly higher than the {\NashSocialWelfareMaximizer} $\mech$, which is a contradiction. This completes the proof of the first claim about traders' ex ante utility in the lemma statement. Invoking the fact that the {\SecondBest} $\OPTSB$ is upper bounded by $\sellerbenchmark + \buyerbenchmark$, we finish the whole analysis of \Cref{lem:NSWM trader approx} as desired.
\end{proof}

\subsection{Tight Negative Result of {\NashSocialWelfareMaximizer} for Zero-Value Seller}
We next show that even for the simple setting where the seller has zero value,
there exists an instance (\Cref{example:all fair mech:irregular}) in which no {\NashSocialWelfareMaximizer} can obtain more than half of the {\SecondBest} $\OPTSB$. We note that this is the same example which shows the tight negative result for all BIC, IIR, ex ante WBB, {\ksfair} mechanisms in \Cref{thm:optimal GFT:general instance} and \Cref{lem:optimal GFT upper bound:irregular}.

\begin{lemma}
\label{lem:NSWM GFT UB:general instance}
    Fix any $\eps > 0$.
    In \Cref{example:all fair mech:irregular} with sufficiently large $\constantH$ (as a function of $\eps$), every {\NashSocialWelfareMaximizer} $\mech$ obtains at most $(\frac{1}{2} + \eps)$ fraction of the {\SecondBest} $\OPTSB$, i.e., $\GFT{\mech} \leq (\frac{1}{2} + \eps) \cdot \OPTSB$. 
\end{lemma}

\begin{proof}
    In \Cref{example:all fair mech:irregular}, the {\SecondBest} $\OPTSB$ can be computed as
    \begin{align*}
        \OPTSB = \expect[\val]{\val} =
        \displaystyle\int_1^{\constantH} \val\cdot \d \buyercdf(\val)
        + \constantH\cdot (1 - \buyercdf(\constantH))
        =
        (1 - o_{}(1))\cdot \ln \constantH
    \end{align*}
    Moreover, the seller's benchmark can be computed as $
        \sellerbenchmark = \sqrt{\ln \constantH}$.
    Fix any {\NashSocialWelfareMaximizer} $\mech = (\alloc,\price,\sellerprice)$. Invoking \Cref{lem:NSWM trader approx}, we know $
        \frac{\sellerexanteutil(\mech)}{\sellerbenchmark} \geq \frac{1}{2}$.
    We next upper bound the seller's ex ante utility $\sellerexanteutil(\mech)$ as follows
    \begin{align*}
        \sellerexanteutil(\mech) &{} = 
        \sellerprice(0) \overset{(a)}{\leq} 
        \expect[\val]{\price(\val)}
        \overset{(b)}{=}
        \expect[\val]{\virtualval(\val)\cdot \alloc(\val)}
        \\
        &{} 
        \overset{(c)}{\leq} 
        \virtualval(\constantH)\cdot (1 - \buyercdf(\constantH))
        +
        \virtualval(\val\primed)\cdot (\buyercdf(\constantH) - \buyercdf(\val\primed)) \cdot \alloc(\val\primed)
        % \\
        % &{}
        =
        \left(1 - \alloc(\val\primed) + o_{}(1)\right)\cdot \sqrt{\ln\constantH}
        % \\
        % &{} 
    \end{align*}
    where inequality~(a) holds since mechanism $\mech$ is ex ante WBB, 
    equality~(b) holds due to \Cref{prop:revenue equivalence}, and
    inequality~(c) holds since the virtual value function $\virtualval$ satisfies $\virtualval(\constantH) = \constantH \geq 0$, $\virtualval(\val) = \virtualval(\val\primed) < 0$ for $\val\in[\val\primed,\constantH)$, $\virtualval(\val) = 0$ for $\val\in[1, \val\primed)$, and the buyer's interim allocation $\alloc$ is weakly increasing due to BIC. Putting the two pieces together, we know 
    \begin{align*}
        \frac{1}{2} \leq \frac{\sellerexanteutil(\mech)}{\sellerbenchmark}
        \leq \frac{        \left(1 - \alloc(\val\primed) + o_{}(1)\right)\cdot \sqrt{\ln\constantH}}{\sqrt{\ln \constantH}}
    \end{align*}
    which implies $\alloc(\val\primed) \leq \frac{1}{2} + o_{}(1)$. Thus, the GFT of mechanism $\mech$ can be upper bounded as 
    \begin{align*}
        \GFT{\mech} & = \expect[\val]{\val\cdot \alloc(\val)} 
        \leq
        \displaystyle\int_1^{\val\primed} \val\cdot \alloc(\val)\cdot \d \buyercdf(\val)
        +
        \displaystyle\int_{\val\primed}^{\constantH} \val\cdot \d \buyercdf(\val)
        + \constantH\cdot (1 - \buyercdf(\constantH))
        \\
        &{}=
        \displaystyle\int_1^{\val\primed} \val\alloc(\val)\cdot \d \buyercdf(\val)
        +
        o_{}(\ln\constantH)
        \overset{(a)}{\leq} 
        \displaystyle\int_1^{\val\primed} \val\cdot \alloc(\val\primed)\cdot \d \buyercdf(\val)
        +
        o_{}(\ln\constantH)
        \\
        & {} =
        (\alloc(\val\primed) - o_{}(1))\cdot \ln\constantH
        \leq 
        \left(\frac{1}{2} + o_{}(1)\right)\cdot \ln\constantH
    \end{align*}
    where inequality~(a) holds since mechanism $\mech$ is BIC and thus interim allocation $\alloc(\val)$ is weakly increasing in $\val$.
    Combining with the fact that the {\SecondBest} $\OPTSB = (1 - o_{}(1))\cdot \ln\constantH$ argued above, we finish the proof of \Cref{lem:NSWM GFT UB:general instance}.
\end{proof}

%% file: Paper/conclusion.tex
In this work, we study the impact of fairness requirements on the social efficiency of the Bayesian bilateral trade problem. In spirit, the ex-ante determination of a mechanism for bilateral trade can be viewed as a cooperative bargaining problem. Motivated by the Kalai-Smorodinsky solution to the bargaining problem, we introduce the notion of {\ksfairness}. Although the {\SecondBest} cannot be achieved by mechanisms that are both {\ksfair} and truthful (that is, BIC, IIR, and ex ante WBB), approximating it is possible. We characterize the best, or almost the best, fraction of the {\SecondBest} obtainable by any {\ksfair} mechanism in various settings. We first present results for general bilateral trade instances, and then results for zero-value seller instances where the buyer value distribution is either regular or MHR.

There are many interesting directions for future research. First, it would be interesting to understand whether the GFT-optimal {\ksfair} mechanism is indeed the Kalai-Smorodinsky solution. One possible approach is to systematically develop a framework that converts any ex ante WBB, {\ksfair} mechanism to an ex ante (or ex post) SBB, {\ksfair} mechanism without negative transfer. \citet{BCW-22} develops a framework to convert ex ante WBB mechanisms to ex post SBB mechanisms with the same allocation, but it does not preserve {\ksfairness}.
Second, for zero-value seller instances, it would also be worthwhile to improve our results by developing a pure analytical analysis (without the final numerical evaluations of the GFT approximation ratios). Third, for general bilateral trade instances, can one obtain a further improved GFT approximation that beats $\frac{1}{e - 1}$ when both traders have MHR distributions? 
Note this cannot be achieved by a direct generalization of \Cref{thm:improved GFT:mhr buyer} (which analyzes a {\ksfair} {\FixPrice}, which is DSIC). As shown in \citet{BD-21} (Proposition B.1), for any $\varepsilon > 0$  there exist bilateral instances where both traders have MHR distributions, for which no DSIC, IIR, ex ante WBB mechanism can achieve $\varepsilon$ fraction of the {\SecondBest}.
More generally, is it possible to obtain a GFT approximation that beats $\frac{1}{2}$ when both traders have regular distributions? As we have shown in \Cref{lem:BROM:regular}, this cannot be achieved using the {\ksfair} {\BiasedRandomOffer}. Hence, new implementations of {\ksfair} mechanisms might be needed. Finally, it would be valuable to extend the definitions and results to multi-buyer settings, and then also to double-auction settings.

%% file: Paper/bargaining-and-bilateral-trade.tex
In this section, we discuss the relationship between fair bilateral trade and cooperative bargaining. First, we discuss an abstract model of cooperative bargaining, and discuss a reduction from picking a truthful mechanism for a bilateral trade instance, to picking a solution to a bargaining problem. Next, we discuss several prominent solutions for the cooperative bargaining problem and reexamine them through the lens of bilateral trade. We show a connection between fairness notions discussed in this paper for bilateral trade, and 
solutions to bargaining problems.

Consider a bilateral instance $\btinstance$. Recall that we denote by $\mechfamily$ the family of all BIC, IIR  and ex ante WBB mechanisms for $\btinstance$. For each such mechanism $\mech\in \mechfamily$, we can compute the ex ante utilities of the buyer and the seller, which are  $\buyerexanteutil(\mech)$ and $ \sellerexanteutil(\mech)$, respectively. We are interested in studying fairness notions for bilateral trade mechanisms. In essence, this amounts to selecting which pairs of utility values $(\buyerexanteutil(\mech), \sellerexanteutil(\mech))$ are considered fair (ex ante). Thus, instead of thinking of a trading problem over an item, one can  take an ex ante view and think of the problem of picking a fair truthful mechanism as a cooperative bargaining problem over ex ante utilities of truthful mechanisms.

The general study of cooperative bargaining problems has a rich history, originating from the seminal work of Nash \cite{N-50}. We begin by presenting an abstract model of cooperative bargaining.

\begin{definition} [Bargaining Problem]
    A \emph{bargaining problem} for two agents is a pair $(\bargainprob,\dispnt)$ where $\bargainprob$ is a bounded, closed and convex subset of $\RR\times\RR$ and $\dispnt\in \bargainprob$ is the disagreement point.

    A \emph{bargaining solution} to the bargaining problem is a function $\bargainsolution$ that given a pair $(\bargainprob,\dispnt)$ outputs an agreement point $\bargainsolution(\bargainprob,\dispnt) \in \bargainprob$.
\end{definition}

Intuitively, the coordinates of each point in $\bargainprob$ represent the utilities of the agents when agreeing on that outcome. {Throughout the paper, for a given point $\pnta\in \RR\times\RR$ we will denote the utilities of agents by $\buyercoord{\pnta}=x_1$ and $\sellercoord{\pnta}=x_2$. The disagreement point $\dispnt$ represents the utilities that each agent gets when no agreement is reached. Note that the agents will never agree on a solution $\pnta\in \bargainprob$ if $x_i < d_i$ for some agent $i$. For our purposes we will always assume that $\dispnt=(0,0)$ and omit $\dispnt$ entirely.

For a bilateral trade instance $\btinstance$ we will define a bargaining problem $\bargainprobBT$ where each point corresponds to the ex ante utilities of the agents for some BIC, IIR and ex ante WBB mechanism. Formally:
\begin{definition} [Bilateral Trade Bargaining] 
The bargaining set that corresponds to the bilateral trade instance $\btinstance$ is:
\begin{align*}
    \bargainprobBT=\left\{(\buyerexanteutil(\mech),\sellerexanteutil(\mech)) \mid \mech \in \mechfamily\right\}
\end{align*}
\end{definition}

Note that the set $\bargainprobBT$ is indeed convex (and therefore compact): for any two mechanisms $\mech_1,\mech_2\in \mechfamily$ and constant probability $\mixprob$, we can define a mechanism $\mech$ which runs $\mech_1$ with probability $\mixprob$ and $\mech_2$ with probability $1-\mixprob$, and the ex ante utilities of this new mechanism will be the convex combination of the utilities of each mechanism. Note that this new mechanism is also in $\mechfamily$ since randomizing between two BIC, IIR and WBB mechanisms maintains these properties. Additionally, the set $\bargainprobBT$ is closed since the set of possible utility outcomes from the set of BIC, IIR and ex ante WBB mechanisms is closed, as these requirements all correspond to weak inequalities. Finally, the disagreement point for the bilateral trade bargaining problem is always $\dispnt=(0,0)$, since if the agents cannot reach an agreement then no trade occurs and neither gains or loses utility.

An important aspect of the bilateral trade problem that is not captured by the corresponding bargaining problem is the profit of the mechanism. WBB mechanisms allow the buyer to pay a higher price than the seller receives, and this additional payment becomes the profit of the mechanism. Importantly, this payment is considered as part of the GFT of the mechanism, yet it is not visible in the bargaining problem. Thus, given a bilateral trade instance $\btinstance$, each point $\pnta\in\bargainprobBT$ corresponds to a BIC, IIR and ex post WBB mechanism $\mech\in\mechfamily$ with $\GFT{\mech} \geq \buyercoord{\pnta}+\sellercoord{\pnta}$, {and the inequality might be strict}. 

A consequence of this difference is that given a bilateral trade instance $\btinstance$, the set of Second-Best mechanisms (those that maximize the GFT) is not necessarily on the Pareto front of the bargaining problem $\bargainprobBT$ (the set of Pareto efficient points). In \cite{BCWZ-17} it is shown that any BIC, IIR and ex ante WBB mechanism can be converted to a BIC, IIR and ex post SBB mechanism with the same GFT, which implies that at least one of the Second-Best mechanisms should be on the Pareto front, but this reduction may change the relative utilities of the traders. Thus the KS solution (see \Cref{def:ks-solution}) to the bargaining problem $\bargainprobBT$ may not correspond to the {\ksfair} mechanism with the highest GFT. To show that this is not the case, it is sufficient to convert a BIC, IIR, ex ante WBB and {\ksfair} mechanism into a BIC, IIR, ex ante \emph{SBB} and {\ksfair} mechanism with the same GFT. A natural attempt to accomplish this is to proportionally ex ante split the ex ante gains of the mechanism between the two traders (an ex ante split does not affect incentive compatibility or participation). However, using this method there might be ex post positive transfer to the buyer, which is impractical and unnatural. We leave this challenge for future work.

\subsection{Bargaining Solutions}

There are many proposed solutions to the bargaining problem, each corresponding to different notions of fairness\footnote{{Each of the solutions can be proven to match a different set of axiomatic assumptions on the behavior of the solution. We will not discuss this axiomatic analysis in this paper. See \cite{T-94} for a general overview.}}. In this paper we will discuss the three most prominent solutions:
\begin{enumerate}
    \item \emph{The Egalitarian solution:} the maximal\footnote{{Throughout this paper we will use the term "maximal" to refer to the "north-east" point on a closed line with positive slope. Formally, this will be the point $\pnta$ on the line which maximizes $\buyercoord{\pnta} + \sellercoord{\pnta}$. Since we will only consider lines with positive slope, this will be equivalent to maximizing $\buyercoord{\pnta}$ or $\sellercoord{\pnta}$ separately.}} outcome in which both agents receive the same utility.
    \item \emph{The Kalai-Smorodinsky solution:} the maximal outcome in which both agents receive the same portion of their optimal utility.
    \item \emph{The Nash solution:} the outcome which maximizes the product of both agents' utilities.
\end{enumerate}

We will now use the reduction from bilateral-trade to the bargaining problem defined above to adapt these solutions into fairness notions for bilateral trade.

\subsubsection{Egalitarian Solution}
One solution to the bargaining problem is the egalitarian solution, originally presented in \citet{Kal-77,Mye-77}:

\begin{definition}
    The \emph{egalitarian solution} $\pnta=\soleg(\bargainprob)$ is the maximal point $\pnta\in\bargainprob$ of equal coordinates $\buyercoord{x}=\sellercoord{x}$. 
\end{definition}

A useful intuition is to think of the egalitarian solution as intersection between the {north-east} (Pareto) boundary of $\bargainprob$ and the 45-degree ray from the origin.

Consider now a bilateral-trade instance $\btinstance$ and the corresponding bargaining problem $\bargainprobBT$. The egalitarian solution $\soleg(\bargainprobBT)$ corresponds to the mechanism $\mech\in \mechfamily$ that satisfies $\buyerexanteutil(\mech) = \sellerexanteutil(\mech)$ {with the highest ex ante utility for both agent's}. This incentivizes the following fairness notion:

\begin{restatable}[Equitability]{definition}{defequitability}
\label{def:equitable}
    {For a given bilateral trade instance $\btinstance$, mechanism $\mech\in \mechfamily$} is \emph{\equitable} if two players achieve the same ex ante utilities, i.e., $\sellerexanteutil(\mech) = \buyerexanteutil(\mech)$.
\end{restatable}

Thus, in bilateral trade the egalitarian solution corresponds to the equitable mechanism $\mech$ {that maximizes $\buyerexanteutil(\mech) + \sellerexanteutil(\mech)$} (or equivalently, either $\buyerexanteutil(\mech)$ or $\sellerexanteutil(\mech)$ individually). Equitability is certainly a desirable fairness notion, but, unfortunately, it also too stringent a requirement in many cases: equitable mechanisms cannot guarantee any constant fraction of the {\SecondBest}, even when limited to a zero-cost seller any buyer with regular distribution. An analysis of equitable mechanisms can be found in \Cref{appendix:equitablity}.

\subsubsection{Kalai-Smorodinsky Solution}

\label{subsec:ks-solution}

Throughout this paper we study the fairness notion of {\ksfairness} for bilateral trade.\footnote{See \Cref{sec:ksfairness} for a formal definition.} This notion is related to the Kalai-Smorodinsky solution to the bargaining problem, originally presented in \citep{KS-75}.
This solution picks the Pareto (maximal) point that equalizes the ratio of each agent's utility to her ideal utility.

\begin{definition}
\label{def:ks-solution}
    For a given bargaining problem $\bargainprob$, the \emph{ideal point} $\idealp{\bargainprob}$ of $\bargainprob$ is defined $\idealbuyer{\bargainprob}=\max\left\{\buyercoord{\pnta}\mid\pnta\in \bargainprob\right\}$ and $\idealseller{\bargainprob}=\max\left\{\sellercoord{\pnta}\mid \pnta\in \bargainprob\right\}$.
    
    The \emph{Kalai-Smorodinsky solution} $\pnta=\solks(\bargainprob)$ is the point $\pnta\in \bargainprob$ satisfying $\frac{\buyercoord{\pnta}}{\idealbuyer{\bargainprob}}=\frac{\sellercoord{\pnta}}{\idealseller{\bargainprob}}$ which maximizes $\buyercoord{\pnta}+\sellercoord{\pnta}$. Alternatively, it is the maximal point of $\bargainprob$ on the line connecting the origin to $\idealp{\bargainprob}$. {We refer to the line connecting the ideal point and the origin as the \emph{\ksline}}.
\end{definition}

For a bilateral trade instance $\btinstance$ we consider the corresponding bargaining problem $\bargainprobBT$. The ideal point of this problem is
\begin{align*}
\idealp{\bargainprobBT}= \left(\max_{\mech\in\mechfamily}\buyerexanteutil(\mech),\max_{\mech\in\mechfamily}\sellerexanteutil(\mech)\right)   =  (\buyerbenchmark, \sellerbenchmark)
\end{align*}

Thus, the Kalai-Smorodinsky KS solution corresponds to the mechanism $\mech\in \mechfamily$ {which maximizes $\buyerexanteutil(\mech)$ among all those satisfying} {\ksfairness}.

We note that given a bilateral trade instance $\btinstance$, it is possible to explicitly cast the problem of finding a KS solution mechanism $\mech$ as a linear optimization problem  (See \Cref{sec:ksfairness} for a detailed description of this program). Moreover, \Cref{lem:SBB implication} implies that all the positive results of this paper (\Cref{thm:optimal GFT:general instance,thm:improved GFT:mhr traders,thm:improved GFT:regular buyer,thm:improved GFT:mhr buyer}) for the approximation to the {\SecondBest} achievable by a {\ksfair} and truthful mechanisms also apply to the KS solution. Additionally, the upper bounds presented in these results also hold for the KS solution in their corresponding settings.

\subsubsection{Nash Solution}
The Nash solution is a prominent solution to the bargaining problem, originally presented in \cite{N-50}. 

\begin{definition}
    The \emph{Nash solution} $x=\solnash(\bargainprob)$ is the maximizer of the product $x_1 \cdot x_2$.
\end{definition}

For the bilateral trade mechanism, the Nash solution will be the mechanism $
\mech \in \mechfamily$ that maximizes $\buyerexanteutil(\mech)\cdot \sellerexanteutil(\mech)$. An examination of Nash Social Welfare Maximization for bilateral trade can be found in \Cref{sec:nash fairness}. In particular, we show a tight bound of $\frac{1}{2}$ on the approximation of the {\SecondBest} obtainable by a mechanism which maximizes the (ex ante) Nash Social Welfare.

%% file: Paper/general-bargaining-results.tex
In this section we extrapolate the results of \Cref{sec:general-gft-approx} to the more general model of cooperative bargaining.
% \footnote{For formal definitions of the bargaining model it's relation to bilateral trade, see \Cref{appendix:bargaining-and-trade}.} 
We show that both \Cref{thm:blackbox reduction} and \Cref{thm:optimal GFT:general instance} can be generalized to the cooperative bargaining model. 

The following known result for bilateral trade helps bridge the gap between the models.

\begin{restatable}[\citealp{BCWZ-17}]{lemma}{SBupperbound}
 \label{lemma:SB-upper-bound}
     For every bilateral trade instance $\btinstance$, the {\SecondBest} $\OPTSB$ is at most the sum of two players' ideal utilities $\sellerbenchmark,\buyerbenchmark$, i.e., $\OPTSB \leq \sellerbenchmark + \buyerbenchmark$.
\end{restatable}

This result motivates the following definition for the bargaining model:

\begin{definition}
    The \emph{ideal {total} utility} of a given bargaining problem $\bargainprob$ is $\idealutility{\bargainprob} = \idealbuyer{\bargainprob} + \idealseller{\bargainprob}$. We write that point $\pnta\in\bargainprob$ has a $\GFTapprox$ fraction of the ideal utility if $\buyercoord{\pnta} + \sellercoord{\pnta} \geq \GFTapprox\cdot\idealutility{\bargainprob}$.
\end{definition}

Recall that {\ksfair} mechanisms in the bilateral trade model correspond to points on the {\ksline} in the bargaining model. Combining this with the above definition and \Cref{lemma:SB-upper-bound} implies a reduction of the bilateral trade model to the bargaining problem: {any positive result shown for the fraction of the ideal {total} utility attainable by a point on the {\ksline} for the general bargaining problem implies that same fraction of the {\SecondBest} can always be obtained by a truthful {\ksfair} mechanism in the bilateral trade model}. This reduction shows that the general results depicted in this section imply their corresponding results from \Cref{sec:general-gft-approx}.

\subsubsection{Generalization of \texorpdfstring{\Cref{thm:blackbox reduction}}{Theorem~\ref{thm:blackbox reduction}}}

%In this section we 
{We first} generalize the black-box reduction framework of \Cref{thm:blackbox reduction} to the bargaining model. Given a bargaining problem $\bargainprob$, the framework converts an arbitrary point $\pnta\in\bargainprob$ into a point $\pntb$ along the {\ksline} of $\bargainprob$, with a provable approximation guarantee of the ideal {total} utility of $\bargainprob$. The proof is near identical to the proof of \Cref{thm:blackbox reduction}, we write the proof explicitly for the sake of formality.

\begin{theorem}[Black-box reduction]
    \label{lemma:bargaining:blackbox reduction}
    Fix some bargaining problem $\bargainprob$ and point $\pnta\in \bargainprob$. Define constant $\GFTapprox\in[0,1]$ as
    \begin{align*}
        \GFTapprox \triangleq \min\left\{
        \frac{\buyercoord{x}}{\idealbuyer{\bargainprob}},
        \frac{\sellercoord{x}}{\idealseller{\bargainprob}}
        \right\}
    \end{align*}
    Then, there exists a point $\pntb\in S$ on the {\ksline} which is has a $\GFTapprox$ fraction of the ideal utility of $\bargainprob$, that is $\buyercoord{\pntb} + \sellercoord{\pntb} \geq \GFTapprox\cdot \idealutility{\bargainprob}$. Specifically, the solution $\pntb$ is a convex combination $\pntb=\mixprob \cdot \pnta + (1-\mixprob) \cdot \pntc$ for some appropriate $\mixprob\in[0, 1]$, where $\pntc$ is chosen as follows:
    \begin{itemize}
        \item If $\frac{\buyercoord{\pnta}}{\idealbuyer{\bargainprob}} \geq \frac{\sellercoord{\pnta}}{\idealseller{\bargainprob}}$ then let $\pntc$ be some arbitrary point satisfying $\sellercoord{\pntc} = \idealseller{\bargainprob}$.
        \item If $\frac{\buyercoord{\pnta}}{\idealbuyer{\bargainprob}} < \frac{\sellercoord{\pnta}}{\idealseller{\bargainprob}}$ then let $\pntc$ be some arbitrary point satisfying $\buyercoord{\pntc} = \idealbuyer{\bargainprob}$.
    \end{itemize}
\end{theorem}

\begin{proof}
    Without loss of generality, we assume $\frac{\buyercoord{\pnta}}{\idealbuyer{\bargainprob}} \geq \frac{\sellercoord{\pnta}}{\idealseller{\bargainprob}} = \GFTapprox$. The other case follows a symmetric argument. Fix an arbitrary point $\pntc\in {\bargainprob}$ such that $\sellercoord{\pntc}=\idealseller{\bargainprob}$. For every $\mixprob\in[0, 1]$, consider the point $\pntb_{\mixprob}$ defined as the convex combination $\pntb_{\mixprob}=\mixprob \cdot \pnta + (1-\mixprob)\cdot \pntc$. By construction, we have
    \begin{align*}
        \frac{\sellercoord{\pntb_{\mixprob}}}{\idealseller{\bargainprob}} 
        =
        \frac{\mixprob\cdot \sellercoord{\pnta} + (1-\mixprob)\cdot \sellercoord{z}}{\idealseller{\bargainprob}} 
        =
        \frac{\mixprob\cdot \sellercoord{\pnta} + (1-\mixprob)\cdot \idealseller{\bargainprob}}{\idealseller{\bargainprob}} 
    \end{align*}
    which is weakly decreasing linearly in $\mixprob\in[0, 1]$, since $\sellercoord{\pnta} \leq \idealseller{\bargainprob}$ by definition. Similarly, 
    \begin{align*}
        \frac{\buyercoord{\pntb_{\mixprob}}}{\idealbuyer{\bargainprob}} 
        =
        \frac{\mixprob\cdot \buyercoord{\pnta} + (1-\mixprob)\cdot \buyercoord{z}}{\idealbuyer{\bargainprob}} 
    \end{align*}
    which is also linear in $\mixprob\in[0, 1]$. Moreover, due to the case assumption, we know 
    \begin{align*}
        \frac{\sellercoord{\pnta}}{\idealseller{\bargainprob}} 
        \leq 
        \frac{\buyercoord{\pnta}}{\idealbuyer{\bargainprob}}
        \;\;
        \mbox{and}
        \;\;
        \frac{\pntb^{0}}{\idealseller{\bargainprob}} = 1
        \geq 
        \frac{\pntb^{0}}{\idealbuyer{\bargainprob}}
    \end{align*}
    Thus, there exists $\mixprob^*\in[0, 1]$ such that 
    \begin{align*}
        \frac{\buyercoord{\pntb_{\mixprob^*}}}{\idealbuyer{\bargainprob}} 
        \overset{(a)}{=} 
        \frac{\sellercoord{\pntb_{\mixprob^*}}}{\idealseller{\bargainprob}} 
        \overset{(b)}{\geq}
        \frac{\sellercoord{\pntb_{1}}}{\idealseller{\bargainprob}} 
        \overset{(c)}{\geq}
        \GFTapprox
    \end{align*}
    which implies that the point $\pntb_{\mixprob^*}$ is on the {\ksline}. Here, equality~(a) holds due to the intermediate value theorem, inequality~(b) holds due to the monotonicity of $
    {\pntb_{\mixprob}}/{\idealseller{\bargainprob}}$ as a function of $\mixprob$ argued above,
    and 
    inequality~(c) holds since the point $\pntb^{1}$ is equivalent to original point $\pnta$. Finally,
    \begin{align*}
        \pntb_{\mixprob^*} + \pntb_{\mixprob^*} \geq \GFTapprox \cdot (\idealbuyer{\bargainprob} + \idealseller{\bargainprob}) = \GFTapprox \cdot \idealutility{\bargainprob}
    \end{align*}
    which completes the proof. 
\end{proof}

% I deleted this figure because I feel it is too complicated.
%    \usefigures{        \begin{figure}            \centering            \subfloat[]{        \input{Figures/fig-bargaining-blackbox}        \label{fig:bargaining:black-box}        }\caption{            A visual proof of \Cref{lemma:bargaining:blackbox reduction}: We begin from the point $\pnta=(1.2,0.6)$. This point has $\frac{4}{5}$ of the ideal value on the horizontal axis and $\frac{2}{3}$ of the ideal value on the vertical axis, so $\GFTapprox=\min(\frac{4}{5},\frac{2}{3})=\frac{2}{3}$. We connect our point to an arbitrary optimal point of the vertical axis, and find the intersection with the {\ksline}. The green line represents all points where the total utility is $\frac{2}{3}$ of the ideal {total} utility: we wish to show that the red point is above this line.  \Cref{fig:bargaining:half-approx} shows a visual proof of \Cref{thm:bargaining:half-approx-ksline}. Explanation: the red line is the {\ksline}. The black points are the optimal solutions, and the black line connects them. The green line is the line of all points with... \textbf{TBD}}\end{figure}    }

\subsubsection{Generalization of \texorpdfstring{\Cref{thm:optimal GFT:general instance}}{Theorem~\ref{thm:optimal GFT:general instance}}}

%In this section we 
{We now present a generalization of} \Cref{thm:optimal GFT:general instance} to the bargaining model. We show that for any bargaining problem $\bargainprob$, there always exists a point $\pnta\in \bargainprob$ on the {\ksline} such that $\pnta$ is a $\frac{1}{2}$ approximation of the ideal {total} utility of $\bargainprob$. The proof is near identical to the proof of \Cref{thm:optimal GFT:general instance}, we write the proof explicitly for the sake of formality. We make use of the generalized black-box reduction  \Cref{lemma:bargaining:blackbox reduction} proven in the previous section.

\begin{theorem}[Ideal {Total} Utility Approximation on the {\ksline}]
\label{lemma:bargaining:half-approx-ksline}
    Fix a bargaining problem $\bargainprob$ and points $\buyerbenchmarkpoint,\sellerbenchmarkpoint\in\bargainprob$ such that $\buyerbenchmarkpoint$ is an ideal point for the first agent ($\buyercoord{\buyerbenchmarkpoint}=\idealbuyer{\bargainprob}$ and $\sellerbenchmarkpoint$ is an ideal point for the second agent ($\sellercoord{\sellerbenchmarkpoint}=\idealseller{\bargainprob}$). Then the line connecting $\buyerbenchmarkpoint$ and $\sellerbenchmarkpoint$ intersects the {\ksline} at a point $\pntb\in \bargainprob$ which satisfies $\buyercoord{\pntb}+\sellercoord{\pntb}\geq \frac{1}{2} \cdot \idealutility{\bargainprob}$.
\end{theorem}

\begin{proof}
    Consider the point $\pnta=\frac{1}{2}\cdot \buyerbenchmarkpoint + \frac{1}{2}\cdot \sellerbenchmarkpoint$. Notice that both players' ex ante utilities satisfy
    \begin{align*}
        \frac{\buyercoord{\pnta}}{\idealbuyer{\bargainprob}}
        =
        \frac{\buyercoord{\buyerbenchmarkpoint}+\buyercoord{\sellerbenchmarkpoint}}{2\cdot \idealbuyer{\bargainprob}}
        %=\frac{\idealbuyer{\bargainprob}+\buyercoord{\sellerbenchmarkpoint}}{2\cdot \idealbuyer{\bargainprob}}
        \geq 
        \frac{1}{2}
        \;
        \mbox{ and }
        \;
        \frac{\sellercoord{\pnta}}{\idealbuyer{\bargainprob}}
        =
        \frac{\sellercoord{\buyerbenchmarkpoint}+\sellercoord{\sellerbenchmarkpoint}}{2\cdot \idealseller{\bargainprob}}
        %=\frac{\sellercoord{\buyerbenchmarkpoint}+\idealseller{\bargainprob}}{2\cdot \idealseller{\bargainprob}}
        \geq 
        \frac{1}{2}
    \end{align*}    
    Invoking \Cref{lemma:bargaining:blackbox reduction} with $\GFTapprox=\frac{1}{2}$ for the point $\pnta$, we obtain a point $\pntb\in S$ on the intersection between the {\ksline} and the line connecting $\buyerbenchmarkpoint, \sellerbenchmarkpoint$ such that $\buyercoord{\pntb}+\sellercoord{\pntb}\geq \frac{1}{2} \cdot (\idealbuyer{\bargainprob} +\idealseller{\bargainprob})$, completing the proof.
\end{proof}

% I decided to remove this figure as well since it is strange to try to present an intuitive proof this far into the appendices.
%\usefigures{    \begin{figure}        \centering        \subfloat[]{    \input{Figures/fig-bargaining-half-approx}    \label{fig:bargaining:half-approx}    }    \caption{ A visual proof of \Cref{thm:bargaining:half-approx}: The black points are the optimal solutions of each axis, and the black line connects them. The red line is the {\ksline}, and the red point on the intersection between these lines is $\pnta$. The green line is the of all points $\pntc$ with $\buyercoord{\pntc}+\sellercoord{\pntc}=\frac{1}{2}\cdot\idealutility{\bargainprob}$ ($=8$). Finally, the orange line is the secondary axis of the rectangle. The visual proof is as follows: note that the red, green and orange lines must intersect at the center of the rectangle. Since the optimal points (in black) must be above the orange line, the black line must intersect the red line above the green line.}    \end{figure}}

%% file: Paper/equitable-mechanisms.tex
In this section we study another fairness notion, {\equitability} (formally defined below), which requires the ex-ante utilities of the agents to be equalized by a mechanism that is BIC, IIR, ex ante WBB. We show that even with zero-value seller, the seller ex-ante utility in an {\equitable} mechanism (that is BIC, IIR, ex ante WBB) might be arbitrary lower than the maximum ex ante utility of the seller (which is at most the maximum utility of the buyer). Additionally, for every $\eps > 0$, the GFT of an {\equitable} mechanism might be only an $\eps$ fraction of the {\SecondBest} $\OPTSB$.

\defequitability*

\begin{figure}
    \centering
    \subfloat[]{
\input{Figures/fig-equitable-rev-curve}
\label{fig:equitable:revenue curve}
}~~~~
    \subfloat[]{
\input{Figures/fig-equitable-util-pair}
\label{fig:equitable:ex ante utility pair}
}
\caption{Graphical illustration of \Cref{example:equitable}. In \Cref{fig:equitable:revenue curve}, the black curve is the revenue curve of the buyer. In \Cref{fig:equitable:ex ante utility pair}, the shaded region corresponds to the pair of buyer and seller's ex ante utilities $(\buyerexanteutil,\sellerexanteutil)$ that is achievable from some BIC, IIR, ex ante WBB mechanism. The two black points represent the traders' ex ante utility pair induced by the {\BuyerOffer} and {\SellerOffer}, respectively. The red dashed line represents equitable utility pairs (i.e., ones with $\buyerexanteutil = \sellerexanteutil$). The red square corresponds to the GFT-optimal {\equitable} mechanism, whose GFT and two traders' ex ante utilities are significantly worse than the {\SecondBest} $\OPTSB$ and two traders' own benchmarks $\buyerbenchmark,\sellerbenchmark$.}
    \label{fig:equitable}
\end{figure}
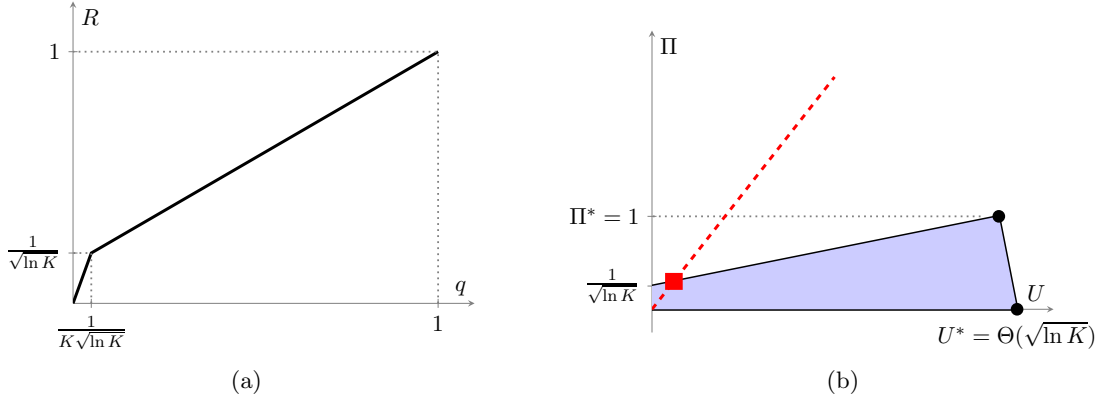

We consider the following instance where the seller has a deterministic value of zero, and the buyer has a regular valuation distribution.
\begin{example}
    \label{example:equitable}
    Fix any $\constantH \geq e$. The buyer has a regular valuation distribution $\buyerdist$, which has support $[1, \constantH]$ and cumulative density function $\buyercdf(\val) = \frac{(\constantH\sqrt{\ln\constantH} - 1)(\val - 1)}{(\constantH\sqrt{\ln\constantH} - 1)(\val - 1) + \constantH - 1}$ for $\val \in[1,\constantH]$
    and $\buyercdf(\val) = 1$ for $\val \in(\constantH, \infty)$.
    (Namely, there is an atom at $\constantH$ with probability mass of $\frac{1}{\constantH\sqrt{\ln\constantH}}$.)
    The seller has a deterministic value of 0.
    See \Cref{fig:equitable} for an illustration of the revenue
curve and the bargaining set of $\buyerdist$.
\end{example}

\begin{lemma}
    \label{lem:equitable GFT UB}
    Fix any $\eps > 0$. In \Cref{example:equitable} with sufficiently large $\constantH$ (as a function of $\eps$),
    every BIC, IIR, ex ante WBB, and {\equitable} $\mech$ obtains at most $\eps^2$ fraction of the {\SecondBest} $\OPTSB$, i.e., $\GFT{\mech} \leq \eps^2 \cdot \OPTSB$. Moreover, the seller and buyer's ex ante utilities is at most $\eps$ and $\eps^2$ fraction of their benchmarks, i.e., $\sellerexanteutil(\mech) \leq \eps \cdot \sellerbenchmark$ and $\buyerexanteutil(\mech) \leq \eps^2 \cdot \buyerbenchmark$.
\end{lemma}
Recall that the (unbiased) {\RandomOffer} ({\RO}) ensures that each trader receives at least $\frac{1}{2}$ fraction of her own benchmark, i.e., $\sellerexanteutil(\RO) \geq \frac{1}{2}\cdot \sellerbenchmark$ and $\buyerexanteutil(\RO)\geq \frac{1}{2}\cdot \buyerbenchmark$. In contrast, \Cref{lem:equitable GFT UB} indicates that by imposing the {\equitability}, it may not only limit the GFT but also significantly reduce \emph{both} traders' ex ante utilities (compared with their benchmarks). In this sense, the {\equitability} is not be an appropriate fairness definition, since it may be disliked by both traders (due to low ex ante utilities) and the social planner (due to low GFT).

\begin{proof}[Proof of \Cref{lem:equitable GFT UB}]
    Let $\constantH = \max\{e^{4/\eps^2},e^{49}\}$. In \Cref{example:equitable}, the {\SecondBest} $\OPTSB$ can be computed as 
    \begin{align*}
        \OPTSB = \expect[\val]{\val} = 
        \displaystyle\int_1^{\constantH}\val\cdot \d \buyercdf(\val) + \constantH\cdot (1 - \buyercdf(\constantH)) \geq  
        \sqrt{\ln\constantH}
        =
        \frac{2}{\eps}
    \end{align*}
    Meanwhile, two traders' benchmarks are $\sellerbenchmark = 1$ and $\buyerbenchmark = \expect[\val]{\val} \geq \frac{2}{\eps}$. 
    
    Fix any BIC, IIR, ex ante WBB, and {\equitable} mechanism $\mech = (\alloc, \price, \sellerprice)$. We consider two cases depending on expected payment $\expect[\val]{\price(\val)}$ of the buyer.

    \xhdr{Case (a) $\expect[\val]{\price(\val)} \leq 2 \sqrt{1 / \ln\constantH}$:} In this case, since mechanism $\mech$ is ex ante WBB, the seller's ex ante utility can be upper bounded as 
    \begin{align*}
        \sellerexanteutil(\mech) = \sellerprice(0) \leq \expect[\val]{\price(\val)} \leq 2 \sqrt{1 / \ln\constantH} = \eps
    \end{align*}
    Moreover, since mechanism is {\equitable}, the buyer's ex ante utility can also be upper bounded as 
    \begin{align*}
        \buyerexanteutil(\mech) = \sellerexanteutil(\mech) \leq 2 \sqrt{1 / \ln\constantH} = \eps
    \end{align*}
    Putting all the pieces together, we show all three inequalities in the lemma statement, i.e.,
    \begin{align*}
        \GFT{\mech} = \expect[\val]{\price(\val)} + \buyerexanteutil(\mech)
        \leq 2\eps \leq \eps^2 \cdot \OPTSB~,
        \\
        \sellerexanteutil(\mech) \leq \eps \leq \eps\cdot \sellerbenchmark
        \;\;\mbox{and}\;\;
        \buyerexanteutil(\mech) \leq \eps \leq \eps^2\cdot \buyerbenchmark
    \end{align*}
    We finish the analysis of case~(a).
    
    \xhdr{Case (b) $\expect[\val]{\price(\val)} \geq 2\sqrt{1/\ln\constantH}$:} In this case, we argue that mechanism $\mech$ violates the {\equitability}. (In other words, this case never happens.) Invoking \Cref{prop:revenue equivalence}, we obtain
    \begin{align*}
        \expect[\val]{\price(\val)} = \expect[\val]{\virtualval(\val)\cdot \alloc(\val)} = \virtualval(\constantH)\cdot (1 - \buyercdf(\constantH))\cdot \alloc(\constantH) 
        +
        \displaystyle\int_{1}^\constantH \virtualval(\val)\cdot \alloc(\val)\cdot\d \buyercdf(\val)
    \end{align*}
    Combining with the facts that $\virtualval(\constantH)\cdot (1 - \buyercdf(\constantH))\cdot \alloc(\constantH)  \leq \sqrt{1/\ln\constantH}$ and $\expect[\val]{\price} \geq 2\sqrt{1/\ln\constantH}$, we obtain
    \begin{align*}
        \displaystyle\int_{1}^\constantH \virtualval(\val)\cdot \alloc(\val)\cdot\d \buyercdf(\val)
        \geq \frac{1}{2}\cdot \expect[\val]{\price(\val)}
    \end{align*}
    Next, we upper bound the ex ante utility of the buyer as follows
    \begin{align*}
        \buyerexanteutil(\mech) &{} = \expect[\val]{\val\cdot \alloc(\val) - \price(\val)}
        \geq 
        \int_1^\constantH \val\cdot \alloc(\val)\cdot \d\buyercdf(\val) - 
        \expect[\val]{\price(\val)}
        \\
        &{}\overset{(a)}{\geq}
        \frac{1}{2}\left(1 - \sqrt{\frac{1}{\ln\constantH}}\right)\cdot \expect[\val]{\price(\val)}\cdot \sqrt{\ln\constantH} - \expect[\val]{\price(\val)} \overset{(b)}{\geq} 2 \cdot \expect[\val]{\price(\val)}
    \end{align*}
    where inequality~(a) holds since $\int_1^\constantH\virtualval(\val)\alloc(\val)\cdot\d\buyercdf(\val) \geq \expect[\val]{\price(\val)}/2$ argued above and the buyer's interim allocation $\alloc(\val)$ is weakly increasing in $\val$ (due to BIC), and inequality~(b) holds since $\constantH\geq e^{49}$. Meanwhile, since mechanism $\mech$ is ex ante WBB, the seller's ex ante utility is at most 
    \begin{align*}
        \sellerexanteutil(\mech) = \sellerprice(0) \leq \expect[\val]{\price(\val)} \leq \frac{1}{2}\cdot \buyerexanteutil(\mech) < \buyerexanteutil(\mech)
    \end{align*}
    which implies mechanism $\mech$ violates the {\equitability}. This is a contradiction. Therefore, we finish the analysis of case (b) and the proof of \Cref{lem:equitable GFT UB}.
\end{proof}

%% file: Figures/fig-equitable-rev-curve.tex
\begin{tikzpicture}[scale=0.8, transform shape]
\begin{axis}[
axis line style=gray,
axis lines=middle,
xlabel = $\quant$,
ylabel = $\revcurve$,
xtick={0, 0.05, 1},
ytick={0, 0.2, 1},
xticklabels={0, $\frac{1}{\constantH\sqrt{\ln\constantH}}$, 1},
yticklabels={0, $\frac{1}{\sqrt{\ln\constantH}}$, 1},
xmin=0,xmax=1.1,ymin=-0.0,ymax=1.2,
width=0.5\textwidth,
height=0.4\textwidth,
samples=1000]

\addplot[black!100!white, line width=0.5mm] (0, 0) -- (0.05, 0.2) -- (1, 1);

\addplot[dotted, gray, line width=0.3mm] (1, 0) -- (1, 1) -- (0, 1);
\addplot[dotted, gray, line width=0.3mm] (0.05, 0) -- (0.05,0.2) -- (0, 0.2);

\end{axis}

\end{tikzpicture}

%% file: Figures/fig-equitable-util-pair.tex
\begin{tikzpicture}[scale=0.8, transform shape]
\begin{axis}[
axis line style=gray,
axis lines=middle,
xlabel = $\buyerexanteutil$,
ylabel = $\sellerexanteutil$,
xtick={0, 1},
ytick={0, 0.1, 0.4},
xticklabels={0, $\buyerbenchmark=\Theta(\sqrt{\ln\constantH})$},
yticklabels={0, $\frac{1}{\sqrt{\ln\constantH}}$, $\sellerbenchmark = 1$},
xmin=0,xmax=1.1,ymin=-0.1,ymax=1.2,
width=0.5\textwidth,
height=0.4\textwidth,
samples=1000]

\addplot[black!100!white, line width=0.5mm] (0, 0) -- (0, 0.1) -- (0.95, 0.4) -- (1, 0) -- (0, 0);

\fill[blue!20] (0, 0) -- (0, 0.1) -- (0.95, 0.4) -- (1, 0) -- (0, 0);

\addplot[dotted, gray, line width=0.3mm] (0.9, 0.4) -- (0, 0.4);

\addplot[dashed, red, line width = 0.5mm] (0, 0) -- (0.5, 1);

\draw[black, fill=black, line width=0.5mm] (axis cs:1, 0) circle[radius=0.08cm];
\draw[black, fill=black, line width=0.5mm] (axis cs:0.95, 0.4) circle[radius=0.08cm];

% \draw[red, fill=red, line width=0.5mm] (axis cs:0.06, 0.12) circle[radius=0.08cm];

\draw[red, fill=red, line width=0.5mm] (0.04, 0.09) rectangle (0.08, 0.15);

\end{axis}

\end{tikzpicture}

%% file: Paper/interim-ex-post-ks-fairness.tex
Recall that the definition of {\ksfairness} (\Cref{def:ks fairness}) imposes a condition on the ex-ante utilities of the two traders. Specifically, it requires that the ex-ante utilities of both traders achieve the same fraction of their respective benchmarks. In this section, we explore two variants of {\ksfairness}, where the condition is modified to apply at the interim stage (\Cref{def:interim ks fairness}) or the ex post stage (\Cref{def:ex post ks fairness}), as defined below.

As demonstrated in \Cref{prop:interim/ex post fairness:no trade}, there exist simple instances in which any BIC, IIR, and ex-ante WBB mechanism that satisfies either interim {\ksfairness} or ex post {\ksfairness} results in the {\NoTrade} mechanism (i.e., one where trade never occurs). 
Thus, these definitions are too stringent, justifying our focus on the ex ante definition of fairness ({\ksfairness}, as defined in \Cref{def:ks fairness}).

\begin{definition}[Interim {\ksfairness}]
\label{def:interim ks fairness}
    For bilateral trade instance $\btinstance$, mechanism $\mech\in \mechfam$ is \emph{interim {\ksfair}}, if the two traders' interim utilities for any pair of value realizations $(\val, \cost)$ achieve the same {fraction of} each trader's own interim benchmark $\sellerbenchmark(\cost)$ and $\buyerbenchmark(\val)$, i.e.,
    \begin{align*}
        \frac{\sellerutil(\cost)}{\sellerbenchmark(\cost)}
        = \frac{\buyerutil(\val)}{\buyerbenchmark(\val)}
    \end{align*}
    where $\sellerutil(\cost)$ (resp.\ $\buyerutil(\val)$) is the seller's (resp.\ buyer's) interim utility given value realization $\cost$ (resp.\ $\val)$ in mechanism $\mech$. Benchmark $\sellerbenchmark(\cost)$ is defined as the seller's interim utility in the BIC, IIR, and ex ante WBB mechanism that maximizes that utility (when she has a deterministic value $\cost$), when facing a buyer with distribution $\buyerdist$. Similarly, benchmark $\buyerbenchmark(\val)$ is defined as the buyer's interim utility in the BIC, IIR, and ex ante WBB mechanism that maximizes that utility (when she has a deterministic value $\val$), when facing a seller with distribution $\sellerdist$.
\end{definition}

\begin{definition}[Ex post {\ksfairness}]
\label{def:ex post ks fairness}
    For bilateral trade instance $\btinstance$, mechanism $\mech\in \mechfam$ is \emph{ex post {\ksfair}}, if the two traders' ex post utilities for any pair of value realizations $(\val, \cost)$ achieve the same {fraction of} each trader's own ex post benchmark $\sellerbenchmark(\val, \cost)$ and $\buyerbenchmark(\val, \cost)$, i.e.,
    \begin{align*}
        \frac{\sellerutil(\val, \cost)}{\sellerbenchmark(\val, \cost)}
        = \frac{\buyerutil(\val, \cost)}{\buyerbenchmark(\val,\cost)}
    \end{align*}
    where $\sellerutil(\val, \cost)$ (resp.\ $\buyerutil(\val, \cost)$) is the seller's (resp.\ buyer's) ex post utility given value realization $\val$ and $\cost$ in mechanism $\mech$. Benchmark $\sellerbenchmark(\val, \cost)$ is defined as the seller's ex post in the BIC, IIR, and ex ante WBB mechanism that maximizes that utility (when she has a deterministic value $\cost$), when facing a buyer with a deterministic value $\val$. Similarly, benchmark $\buyerbenchmark(\val,\cost)$ is defined as the buyer's ex post utility in the BIC, IIR, and ex ante WBB mechanism that maximizes that utility (when she has a deterministic value $\val$), when facing a seller with a deterministic value $\cost$.
\end{definition}

We remark that the interim {\ksfairness} (\Cref{def:interim ks fairness}) implies the ex ante {\ksfairness} (\Cref{def:ks fairness}). To see this, note that in both fairness definition, the seller's ex ante benchmark $\sellerbenchmark$ and interim benchmark $\sellerbenchmark(\cost)$ are consistent, i.e., $\sellerbenchmark = \expect[\cost]{\sellerbenchmark(\cost)}$, and the same consistency holds for the buyer's benchmarks, i.e., $\buyerbenchmark = \expect[\val]{\buyerbenchmark(\val)}$. Following the same reason, we observe that the ex post {\ksfairness} does not imply interim or ex ante {\ksfairness}, since its benchmark may be different, i.e., $\sellerbenchmark \not= \expect[\val,\cost]{\sellerbenchmark(\val,\cost)}$. 

We next present our result for these fairness notions, showing that for non-degenerate instances, only the {\NoTrade} satisfy the interim or ex post {\ksfairness}. The proof is based on studying the following instance.

\begin{example}
\label{example:interim fairness:zero-value seller uniform buyer}
The buyer has valuation distribution $\buyerdist$ with support $\supp(\buyerdist) = [\lval, \hval]$, where $0\leq \lval < \hval$. The seller has a deterministic value of zero.
\end{example}

\begin{lemma}
\label{prop:interim/ex post fairness:no trade}
    In \Cref{example:interim fairness:zero-value seller uniform buyer}, if a BIC, IIR, ex ante WBB mechanism $\mech$ is interim {\ksfair} or ex post {\ksfair}, then it must be the {\NoTrade} (which has zero GFT). 
\end{lemma}
\begin{proof}
    We first consider the case where mechanism $\mech = (\alloc,\price,\sellerprice)$ is interim {\ksfair}. 
    Since mechanism $\mech$ is BIC and IIR, the buyer's interim utility $\buyerutil(\val)$ for every value $\val\in[\lval,\hval]$ can be expressed as 
    $\buyerutil(\val) = \buyerutil(\lval) + \int_{\lval}^\val \alloc(t)\cdot \d t$.
    Meanwhile, since the seller has deterministic value of zero, the buyer's interim benchmark $\buyerbenchmark(\val) = \val$.
    Invoking the assumption that mechanism $\mech$ is interim {\ksfair}, for every value $\val\in[\lval,\hval]$ of the buyer, we know 
    $\frac{\buyerutil(\val)}{\buyerbenchmark(\val)}
        =
        \frac{\sellerutil(0)}{\sellerbenchmark(0)}$.
    Putting all the three pieces together, we know that for every value $\val\in[\lval,\hval]$,
    \begin{align*}
        \frac{1}{\val}\cdot \left(\buyerutil(\lval) + \displaystyle\int_{\lval}^\val \alloc(t)\cdot \d t \right)= 
        \frac{\sellerutil(0)}{\sellerbenchmark(0)} 
        \;\;
        \Longleftrightarrow
        \;\;
        \left(\buyerutil(\lval) + \displaystyle\int_{\lval}^\val \alloc(t)\cdot \d t \right)= 
        \frac{\sellerutil(0)}{\sellerbenchmark(0)} 
        \cdot \val
    \end{align*}
    Taking the derivative on both sides, we obtain that 
    $\alloc(\val) = \frac{\sellerutil(0)}{\sellerbenchmark(0)}$ for every value $\val\in[\lval,\hval]$ and $\buyerutil(\lval) = \frac{\sellerutil(0)}{\sellerbenchmark(0)}\cdot \lval$.
    Consequently, we obtain $\buyerutil(\val) = \frac{\sellerutil(0)}{\sellerbenchmark(0)}\cdot \val$ and $\price(\val) = 0$ for every value $\val\in[\lval,\hval]$.
    Now consider the seller's (interim) utility $\sellerutil(0)$, which can be upper bounded by
    \begin{align*}
        \sellerutil(0) = \sellerprice(0)\leq \expect[\val]{\price(\val)} 
        = 0
    \end{align*}
    where the inequality holds due to ex ante WBB, the second equality holds due to $\price(\val) = 0$ for every value $\val$ argued above.
    Finally, since $\hval > \lval \geq 0$, the seller's (interim) benchmark $\sellerbenchmark(0) > 0$, and thus for every value $\val \in[\lval, \hval]$:
    \begin{align*}
        \alloc(\val) = \frac{\sellerutil(0)}{\sellerbenchmark(0)} = 0 
    \end{align*}
    which implies mechanism $\mech$ is the {\NoTrade}.
    
    Next we assume mechanism $\mech$ is ex post {\ksfair}. Since the seller has a deterministic value of zero, the buyer's interim allocation is equivalent to her ex post allocation. Hence, following the same argument, we obtain $\alloc(\val, 0) = \frac{\sellerutil(\hval,0)}{\sellerbenchmark(\hval,0)} = 0$ (due to the ex post {\ksfairness}), and thus mechanism $\mech$ is the {\NoTrade}. This completes the proof of \Cref{prop:interim/ex post fairness:no trade}.
\end{proof}

%% file: Paper/apx-lemgftprogrammhrbuyer.tex
\label{apx:lemgftprogrammhrbuyer}
In this section, we prove \Cref{lem:GFT program:mhr buyer}. Its analysis is similar to the one for \Cref{lem:GFT program:regular buyer} in \Cref{subsec:improved GFT:regular buyer}. Specifically, we search over all MHR distributions and argue that the worst GFT can only be induced by a subclass of them, whose \emph{cumulative hazard rate function} (rather than the revenue curve, as was the case for regular buyer distributions) can be characterized by finite many parameters (see \Cref{fig:GFT program:mhr buyer}).

\lemgftprogrammhrbuyer*

\begin{proof}
    Let $\cumhazard$ be the cumulative hazard rate function of the buyer and $\optreserve$ be the monopoly reserve for $\buyerdist$. (Since the buyer's valuation distribution is MHR, the monopoly reserve is unique.) Without loss of generality, we normalize the monopoly revenue to be equal to one. Since the buyer's valuation distribution is MHR, cumulative hazard rate function $\cumhazard$ is convex.
    
    \xhdr{Step 0- Introducing necessary notations.} We introduce $\revratio\in(0, 1)$ as a constant whose value will be pined down at the end of this analysis. Given constant $\revratio$, let $\price\in[0, \optreserve]$ be the smallest price such that $\price\cdot (1 - \buyercdf(\price)) \geq \revratio$. 
    Moreover, let $\val_0 \triangleq \price - \cumhazard(\price) / \cumhazard'(\price)$ and $\val_1 \triangleq (\cumhazard(\price) - \cumhazard(\optreserve) + \optreserve\cdot \cumhazard'(\optreserve) - \price\cdot \cumhazard'(\price))/(\cumhazard'(\optreserve) - \cumhazard'(\price))$. Finally, we also define 
    \begin{align*}
        H &\triangleq \expect[\val\sim\buyerdist]{\val\cdot \indicator{\val \geq \optreserve}}~,
        \;\;
        M \triangleq \expect[\val\sim\buyerdist]{\val\cdot \indicator{\price \leq \val < \optreserve}}~,
        \;\;
        L \triangleq \expect[\val\sim\buyerdist]{\val\cdot \indicator{ \val < \price}}~.
    \end{align*}
    which partition the {\SecondBest} into three pieces, i.e., $\OPTSB = \expect[\val]{\val} = H + M + L$.
    All notations are illustrated in \Cref{fig:GFT program:mhr buyer}.

    \xhdr{Step 1- Characterizing of a (possibly) not {\ksfair} {\FixPrice}.}
    First, we consider a (possibly not {\ksfair}) {\FixPrice} ({$\FPM_{\price}$}) with trading price $\price$:
    \begin{itemize}
        \item For the seller, her ex ante utility (aka., revenue) is $\sellerexanteutil(\FPM_{\price}) = \price\cdot (1 - \buyercdf(\price)) = \revratio$. Since her benchmark $\sellerbenchmark$ (aka., monopoly revenue) is normalized to one, her ex ante utility is an $\revratio$ fraction of her benchmark $\sellerbenchmark$.
        \item For the buyer, her ex ante utility can be computed as 
        \begin{align*}
            \buyerexanteutil(\FPM_{\price}) &= 
            \expect[\val]{\plus{\val - \price}} = 
            H + M - \revratio
        \end{align*}
        Meanwhile, the buyer's benchmark $\buyerbenchmark$ is 
        \begin{align*}
            \buyerbenchmark = H + M + L
        \end{align*}
    \end{itemize}
    Putting the two pieces together, we conclude that in this (possibly not {\ksfair}) {\FixPrice} ({$\FPM_{\price}$}), both traders' ex ante utilities satisfy
    \begin{align*}
        \frac{\sellerexanteutil(\FPM_{\price})}{\sellerbenchmark} = \revratio
        \;\;
        \mbox{and}
        \;\;
        \frac{\buyerexanteutil(\FPM_{\price})}{\buyerbenchmark} = \frac{H + M - \revratio}{H + M + L}
    \end{align*}
    
    \xhdr{Step 2- Characterizing of {\ksfair} {\FixPrice}.} Define auxiliary notation $\exanteutilratio\in[0, 1]$ as 
    \begin{align*}
        \exanteutilratio \triangleq \displaystyle
        \revratio - \plus{\frac{H + M + L}{H + M + L + 1}\left(\revratio - \frac{H + M - \revratio}{H + M + L}\right)}
    \end{align*}
    We next show that there exists price $\fprice\in[0,\optreserve]$ such that the {\FixPrice} with trading price $\fprice$ is {\ksfair} and both traders' ex ante utilities are at least $\exanteutilratio$ fraction of their benchmarks $\sellerbenchmark,\buyerbenchmark$, respectively. To see this, consider the following two cases separately.
    \begin{itemize}
        \item Suppose that $\revratio < \frac{H + M - \revratio}{H + M + L}$ and thus $\exanteutilratio = \revratio$. In this case, by increasing trading price $\price$ in the {\FixPrice}, the seller's ex ante utility increases continuously (due to the concavity of revenue curve $\revcurve$) and the buyer's ex ante utility decreases continuously. Invoking the intermediate value theorem, there exists price $\fprice \in (\price, \optreserve)$ such that both traders' ex ante utilities is at least $\exanteutilratio$ fraction of their benchmarks $\sellerbenchmark,\buyerbenchmark$, respectively.
        \item Suppose that $\revratio \geq \frac{H + M - \revratio}{H + M + L}$ and thus $\exanteutilratio = \revratio - {\frac{H + M + L}{H + M + L + 1}\left(\revratio - \frac{H + M - \revratio}{H + M + L}\right)}$. In this case, let $\Delta \triangleq {\frac{H + M + L}{H + M + L + 1}\left(\revratio - \frac{H + M - \revratio}{H + M + L}\right)}\geq 0$. By decreases trading price $\price$ in the {\FixPrice}, the seller's ex ante utility decreases continuously (due to the concavity of revenue curve $\revcurve$). Let $\price\primed < \price$ be the trading price such that the seller's ex ante   utility (aka., revenue) is equal to $\revratio - \Delta$. Under trading price $\price\primed$, the buyer's ex ante utility is at least $H + M - \revratio + \Delta$. (To see this, note that by decreasing trading price from $\price$ to $\price\primed$, the GFT weakly increases and the seller's ex ante utility decreases by $\Delta$. Thus, the buyer's ex ante utility increases by at least $\Delta$.) Due to the definition of $\Delta$, two traders' ex ante utilities in the {\FixPrice} ($\FPM_{\price\primed}$) with trading price $\price\primed$ satisfy
        \begin{align*}
            \frac{\sellerexanteutil(\FPM_{\price\primed})}{\sellerbenchmark} = \revratio - \Delta = \frac{H + M - \revratio + \Delta}{H + M + L}
            \leq 
            \frac{\buyerexanteutil(\FPM_{\price\primed})}{\buyerbenchmark}
        \end{align*}
        If the inequality above holds with equality, the {\FixPrice} with trading price $\price\primed$ is {\ksfair} and both traders' ex ante utilities are at least $\exanteutilratio$ fraction of their benchmarks $\sellerbenchmark,\buyerbenchmark$, respectively. Otherwise, we can invoke argument in the previous case.
    \end{itemize}
    Summarizing the analysis above, the GFT-approximation of the {\ksfair} {\FixPrice} ($\FPM_{\fprice}$) with trading price $\fprice$ can be computed as 
    \begin{align*}
        \frac{\GFT{\FPM_{\fprice}}}{\OPTSB}
        \geq
        \frac{\exanteutilratio \cdot (\sellerbenchmark + \buyerbenchmark)}{\expect[\val]{\val}}
        =
        \exanteutilratio + \frac{\exanteutilratio}{H + M + L}
    \end{align*}

    \xhdr{Step 3- Formulating GFT approximation as two-player game.} Putting all the pieces together, the optimization program in the lemma statement can be viewed as a two-player zero-sum game between a min player (adversary) and a max player (ourself as GFT-approximation prover). The payoff in this game is the GFT approximation lower bound $\exanteutilratio + \frac{\exanteutilratio}{H + M + L}$ shown above. As a reminder
    , quantities $\exanteutilratio, H, M, L$ depend on both the buyer's valuation distribution and constant $\revratio$ used in the analysis. The min player chooses the worst MHR distribution (equivalently, convex cumulative hazard rate function) of the buyer, and the max player chooses constant $\revratio$. Importantly, the choice of constant $\revratio$ can depend on the buyer's valuation distribution. To capture this, we formulate this two-player zero-sum game in three stages: 
    \begin{itemize}
        \item (Stage 1) The min player (adversary) chooses monopoly quantile $\optreserve$ and $H$.
        \item (Stage 2) The max player (ourself as GFT-approximation prover) chooses constant $\revratio$ for the analysis.
        \item (Stage 3) The min player (adversary) chooses $\val_0, \price, M, L$.
    \end{itemize}
    It remains to verify that all constraints in the optimization program capture the feasibility condition for the both min player and max player's actions. We next verify the non-trivial constraints individually. 
    \begin{itemize}
        \item (Bounds for monopoly reserve $\optreserve$) Recall that we normalize the monopoly revenue to be one. Hence, the lower bound that $\optreserve \geq 1$ is trivial. Meanwhile, the upper bound that $\optreserve \leq e$ is due to the MHR condition \citep[Lemma~1 in][]{AGM-09}.
        \item (Bounds for price $\price$) Recall that $\price\in[0, \optreserve]$ is the smallest price such that $\price\cdot (1 - \buyercdf(\price)) \geq \revratio$. Hence, $\price \geq \revratio$.
        Moreover, the convexity of cumulative hazard rate function $\cumhazard$ implies that 
        \begin{align*}
            \cumhazard(\optreserve) + (\price - \optreserve)\cdot \cumhazard'(\optreserve) \leq \cumhazard(\price) \leq \cumhazard(\optreserve) \cdot \frac{\price}{\optreserve}
        \end{align*}
        By the definition of cumulative hazard rate function $\cumhazard$ and our normalization of the monopoly revenue to be one, we have $\cumhazard(\optreserve) = \ln(\optreserve)$, $\cumhazard(\price) = \ln(\price/\revratio)$, and $\cumhazard'(\optreserve) = 1/{\optreserve}$.\footnote{To see $\cumhazard'(\optreserve) = 1/{\optreserve}$, consider the equal-revenue distribution (with monopoly revenue equal to one). The cumulative hazard rate function of the equal-revenue distribution is $\ln(\val)$ for $\val\in[1, \infty)$. Since buyer's valuation distribution $\buyerdist$ is MHR and we normalize its monopoly revenue to one, its induced cumulative hazard rate function $\cumhazard$ is convex and has at most one intersection with $\ln(\val)$ at $\val=\optreserve$. This implies that the derivatives of two curves are identical at $\optreserve$, i.e., $\cumhazard'(\optreserve) = 1/\optreserve$. See the black dashed line in \Cref{fig:GFT program:mhr buyer} for an illustration.} Hence, the inequality above is equivalent to  
        \begin{align*}
            -\optreserve\cdot \LambertFunc\left(-\frac{\revratio}{e}\right) 
            \leq \price \leq  
            - \optreserve\cdot \LambertFunc\left(-\frac{\revratio\ln(\optreserve)}{\optreserve}\right)\cdot \frac{1}{\ln(\optreserve)}
        \end{align*}
        \item (Upper bound for value $\val_0$) Recall that $\val_0 \triangleq \price - \cumhazard(\price) / \cumhazard'(\price)$. The convexity of cumulative hazard rate function $\cumhazard$ implies that 
        \begin{align*}
            \cumhazard'(\quant) \leq \frac{\cumhazard(\optreserve) - \cumhazard(\price)}{\optreserve - \price} = \frac{\ln(\optreserve) - \ln(\price/\revratio)}{\optreserve - \price}
        \end{align*}
        After rearranging, we obtain 
        \begin{align*}
            \val_0 \leq \optreserve - \frac{\ln(\optreserve)}{\ln(\optreserve) - \ln(\frac{\price}{\revratio})}\cdot (\optreserve - \price)
        \end{align*}
        as stated in the constraint. 
        
        \item (Bounds for truncated GFT $L$) Recall that $L \triangleq \expect[\val]{\val\cdot \indicator{\val < \price}}$. The convexity of the cumulative hazard rate function implies that for every value $\val \in [0, \price]$,
        \begin{align*}
            \plus{
            \frac{\val - \val_0}{\price - \val_0}
            \cdot \cumhazard\left(\price\right)
            }
            \leq 
            \cumhazard(\val) 
            \leq
            \frac{\val}{\price}\cdot \cumhazard\left(\price\right)
        \end{align*}
        where both inequalities bind at $\val = \price$ and $\val = 0$. Moreover, $\cumhazard(\price) = \ln(\price/\revratio)$. The left-hand side and right-hand side can be viewed as two cumulative hazard rate functions $\cumhazard_1, \cumhazard_2$ that sandwich the original cumulative hazard rate function $\cumhazard$ (see blue and red cumulative hazard rate functions illustrated in \Cref{fig:GFT program:mhr buyer}). It can be verified that $\LOverBar$ and $\LUnderBar$ are $\expect[\val]{\val\cdot \indicator{ \val < \price}}$ where the random value $\val$ is realized from valuation distribution induced by those two cumulative hazard rate functions $\cumhazard_1, \cumhazard_2$, respectively. Invoking \Cref{lem:cum hazard rate monotonicity}, we obtain $\LUnderBar\leq L \leq \LOverBar$ as stated in the constraint, since these two cumulative hazard rate functions $\cumhazard_1, \cumhazard_2$ sandwich the original cumulative hazard rate function $\cumhazard$.
        \item (Bounds for truncated GFT $M$) The argument is similar to the argument above for truncated GFT $L$. Recall that $M \triangleq \expect[\val]{\val\cdot \indicator{ \price \leq \val <\optreserve }}$. The convexity of the cumulative hazard rate function implies that for every quantile $\val \in [\price, \optreserve]$,
        \begin{align*}
            \max\left\{
            \frac{\val - \val_0}{\price - \val_0}
            \cdot 
            \cumhazard\left(\price\right)
            ,
            \cumhazard(\optreserve) - (\optreserve - \val)\cdot \cumhazard'(\optreserve) 
            \right\}
            \leq 
            \cumhazard(\val) 
            \leq
            \frac{\optreserve - \val}{\optreserve - \price}
            \cdot \left(\cumhazard(\optreserve) - \cumhazard\left(\price\right)\right)
            +
            \cumhazard\left(\price\right)
        \end{align*}
        where both inequalities bind at $\val = \price$ and $\val = \optreserve$. Moreover, $\cumhazard(\price) = \ln(\price/\revratio)$, $\cumhazard(\optreserve) = \ln(\optreserve)$, and $\cumhazard'(\optreserve) = 1/\optreserve$. The left-hand side and right-hand side can be viewed as two cumulative hazard rate functions $\cumhazard_1, \cumhazard_2$ that sandwich the original cumulative hazard rate function $\cumhazard$ (see blue and red cumulative hazard rate functions illustrated in \Cref{fig:GFT program:mhr buyer}). It can be verified that $\MOverBar$ and $\MUnderBar$ are $\expect[\val]{\val\cdot \indicator{ \price \leq \val <\optreserve }}$ where the random value $\val$ is realized from valuation distribution induced by those two cumulative hazard rate functions $\cumhazard_1, \cumhazard_2$, respectively. Invoking \Cref{lem:cum hazard rate monotonicity}, we obtain $\MUnderBar\leq M \leq \MOverBar$ as stated in the constraint, since these two cumulative hazard rate functions $\cumhazard_1, \cumhazard_2$ sandwich the original cumulative hazard rate function $\cumhazard$.
        \item (Bounds for truncated GFT $H$) The argument is similar to the argument above for truncated GFT $L$. Recall that $H \triangleq \expect[\val]{\val\cdot \indicator{\val \geq \optreserve}}$. The convexity of the cumulative hazard rate function implies that for every value $\val \in [\optreserve, \infty)$,
        \begin{align*}
            \cumhazard(\optreserve) + (\val - \optreserve) \cdot 
            \cumhazard'(\optreserve)
            \leq 
            \cumhazard(\val) 
        \end{align*}
        where inequality binds at $\val = \optreserve$. The left-hand side can be viewed as cumulative hazard rate functions $\cumhazard_1$. Moreover, consider another cumulative hazard function $\cumhazard_2$ such that $\cumhazard_2(\optreserve) = \cumhazard(\optreserve)$ and $\cumhazard_2(\val) = \infty$ for every $\val > \optreserve$. Note that $\cumhazard_1$ and $\cumhazard_2$ sandwich the original cumulative hazard rate function $\cumhazard$ (see blue and red cumulative hazard rate functions illustrated in \Cref{fig:GFT program:mhr buyer}). It can be verified that $\HOverBar$ and $\HUnderBar$ are $\expect[\val]{\val\cdot \indicator{  \val \geq\optreserve }}$ where the random value $\val$ is realized from valuation distribution induced by those two cumulative hazard rate functions $\cumhazard_1, \cumhazard_2$, respectively. Invoking \Cref{lem:cum hazard rate monotonicity}, we obtain $\HUnderBar\leq H \leq \HOverBar$ as stated in the constraint, since these two cumulative hazard rate functions $\cumhazard_1, \cumhazard_2$ sandwich the original cumulative hazard rate function $\cumhazard$.
        \item (Equation for value $\val_1$) Recall that  $\val_0 \triangleq \price - \cumhazard(\price) / \cumhazard'(\price)$ and $\val_1 \triangleq (\cumhazard(\price) - \cumhazard(\optreserve) + \optreserve\cdot \cumhazard'(\optreserve) - \price\cdot \cumhazard'(\price))/(\cumhazard'(\optreserve) - \cumhazard'(\price))$. Combining both equations with $\cumhazard(\price) = \ln(\price/\revratio)$, $\cumhazard(\optreserve) = \ln(\optreserve)$, and $\cumhazard'(\optreserve) = 1/\optreserve$, we obtain 
        \begin{align*}
            \val_1 = \left(\frac{1}{\optreserve} - \frac{1}{\price - \val_0}\ln\left(\frac{\price}{\revratio}\right)\right)^{-1}\left(
        \ln\left(\frac{\price}{\revratio}\right) - \ln(\optreserve)
        + 1 
        - \frac{\price}{\price-\val_0}\ln\left(\frac{\price}{\revratio}\right) 
        \right)~,
        \end{align*}
        as stated in the optimization program.
    \end{itemize}
    Finally, we numerically evaluation the optimization program and obtain $\fixedPriceGFTPercentageMHR$. We present more details of this numerical evaluation in \Cref{apx:numerical evaluation:mhr buyer}.
    This completes the proof of \Cref{lem:GFT program:mhr buyer}.
\end{proof}

\begin{lemma}
\label{lem:cum hazard rate monotonicity}
    Given any two distributions $\buyerdist_1,\buyerdist_2$ and any two value $\val\primed, \val\doubleprimed$ with $\val\primed \leq \val\doubleprimed$. 
    Suppose $\buyercdf_1(\val\primed) = \buyercdf_2(\val\primed)$ and $\buyercdf_1(\val\doubleprimed) = \buyercdf_2(\val\doubleprimed)$. If the induced cumulative hazard rate functions $\cumhazard_1,\cumhazard_2$ satisfy that for every value $\val \in[\val\primed,\val\doubleprimed]$,
    $\cumhazard_1(\val) \geq \cumhazard_2(\val)$,
    then 
    \begin{align*}
        \expect[\val\sim\buyerdist_1]{\val\cdot \indicator{\val\primed \leq \val\leq \val\doubleprimed}}
        \leq 
        \expect[\val\sim\buyerdist_2]{\val\cdot \indicator{\val\primed \leq \val\leq \val\doubleprimed}}
    \end{align*}
\end{lemma}
\begin{proof}
    Since for every value $\val \in[\val\primed,\val\doubleprimed]$,
    $\cumhazard_1(\val) \geq \cumhazard_2(\val)$ and the inequality is binding at $\val = \val\primed$ and $\val = \val\doubleprimed$, it is guaranteed that for every value $\val\in[\val\primed,\val\doubleprimed]$, $\buyercdf_1(\val) \geq \buyercdf_2(\val)$
    and the inequality is binding at $\val = \val\primed$ and $\val = \val\doubleprimed$. This is sufficient to prove the lemma statement:
    \begin{align*}
        &\expect[\val\sim\buyerdist_1]{\val\cdot \indicator{\val\primed \leq \val\leq \val\doubleprimed}} 
        =
        \displaystyle\int_{\val\primed}^{\val\doubleprimed}
        \val\cdot \d \buyercdf_1(\val) 
        \overset{(a)}{=}
        \val\doubleprimed \buyercdf_1(\val\doubleprimed)
        -
        \val\primed\buyercdf_1(\val\primed)
        -
        \displaystyle\int_{\val\primed}^{\val\doubleprimed}
        \buyercdf_1(\val) \cdot \d \val
        \\
        &\qquad\overset{(b)}{\leq}
        \val\doubleprimed \buyercdf_2(\val\doubleprimed)
        -
        \val\primed\buyercdf_2(\val\primed)
        -
        \displaystyle\int_{\val\primed}^{\val\doubleprimed}
        \buyercdf_2(\val) \cdot \d \val
        \overset{(c)}{=}  
        \displaystyle\int_{\val\primed}^{\val\doubleprimed}
        \val\cdot \d \buyercdf_2(\val) 
        =
        \expect[\val\sim\buyerdist_2]{\val\cdot \indicator{\val\primed \leq \val\leq \val\doubleprimed}} 
    \end{align*}
    where equalities~(a) (c) hold due to the integration by parts, and inequality~(b) holds as we argued above. This complete the proof of \Cref{lem:cum hazard rate monotonicity}.
\end{proof}

%% file: Paper/apx-numerical-evaluation-regular-buyer.tex
In this section, we provide additional details for our numerical evaluation of program~\ref{program:GFT:regular buyer}.
As the first step, we prove that it suffices to assume $H = 1$ and $L = \LOverBar$.
\begin{lemma}
\label{lem:regular buyer program:reducing variables}
    For program~\ref{program:GFT:regular buyer}, there exists an optimal solution with $H = 1$ and $L = \LOverBar$. Namely, the optimal objective value of program~\ref{program:GFT:regular buyer} is the same as the optimal objective value of program~\ref{program:GFT:regular buyer} with additional constraints (for the minimization) that $H = 1$ and $L = \LOverBar$
\end{lemma}
\begin{proof}
    We first argue that forcing $L = \LOverBar$ does not change the optimal objective value of the program. To see this, note that $L$ only appears in the objective function, where auxiliary variable $\exanteutilratio$ depends on $L$. By algebra, it can be verified that auxiliary variable $\exanteutilratio$ decreases as $L$ increases. Hence, the objective function is decreasing in $L$ and thus there exists an optimal solution such that $L = \LOverBar$.

    We next argue that forcing $H = 1$ does not change the optimal objective value of the program. Suppose there exists an optimal solution $(\optquant, H, \revratio, \quant, \val_0, M, L)$ of the program with $H > 1$ and $L = 1$. Consider an alternative solution $(\optquant\primed, H\primed, \revratio\primed, \quant\primed, \val_0\primed, M\primed, L\primed)$ where we set 
    \begin{align*}
        \optquant\primed \gets e^{1-H}\cdot \optquant,~
        H\primed \gets 1,~
        \revratio\primed \gets \revratio,~
        \quant\primed \gets \quant,~
        \val_0\primed \gets \val_0,~
        M\primed \gets M + (H - 1),~
        L\primed \gets L
    \end{align*} 
    It can be verified that the objective values of two solutions are identical. Moreover, the alternative solution is feasible. To see this, note that $\optquant\primed < \optquant$ since $H > 1$. Consequently, by algebra, the feasible regions of $\quant\primed, \val\primed$, and $M\primed$ becomes larger. This also ensures that there does not exists $\revratio\doubleprimed$ which can improve the objective value while fixing $\optquant\primed$ and $H\primed$. This finishes the proof the lemma.
\end{proof}
After simplifying program~\ref{program:GFT:regular buyer} by fixing $H = 1$ and $L = \LOverBar$, there are still five variables $(\optquant, \revratio,\quant,\val_0, M)$, where $\optquant$ appears in the outer minimization, $\revratio$ appears in the middle maximization, and $\quant,\val_0, M$ appear in the inner minimization. Next, we partition the feasible region of $\optquant$ (i.e., $[0, 1]$) into subintervals. For each subinterval, we manually set the value of $\revratio$. In this way, for each subinterval $[s, \ell]$, we obtain a pure minimization program~\ref{program:GFT:regular buyer:decompose} parameterized by $(s,\ell, \revratio)$ with variables $(\optquant,\quant,\val_0, M)$ defined as follows:
\begin{align}
\label{program:GFT:regular buyer:decompose}
\tag{$\mathcal{P}_{\mathrm{REG}}[s,\ell,\revratio]$}
\arraycolsep=5.4pt\def\arraystretch{1}
    \begin{array}{llll}
     &\min\limits_{\optquant, \quant, \val_0, M}   & 
      \displaystyle\exanteutilratio + \frac{\exanteutilratio}{1 + M + \LOverBar} &
      \vspace{10pt}
      \\
      \vspace{10pt}
      &\text{s.t.}
      & \optquant\in[s, \ell],~  
      % & 
      % \\
      % \vspace{10pt}
      % && 
      \quant\in \left[\optquant + (1 - \revratio)(1 - \optquant), 1\right],~ 
      & 
      \\
      \vspace{10pt}
      && 
      \val_0\in \left[0, 1 - \displaystyle\frac{1 - \revratio}{\quant - \optquant}(1 - \optquant)\right],~ 
      % & 
      % \\
      % \vspace{10pt}
      % && 
      M\in\left[\MUnderBar, \MOverBar\right]& 
    \end{array}
\end{align}
where auxiliary variables $\exanteutilratio$, $\MUnderBar, \MOverBar, \LOverBar$ are follows the same definitions as program~\ref{program:GFT:regular buyer}. 
By construction, program~\ref{program:GFT:regular buyer} can be lower bounded by program~\ref{program:GFT:regular buyer:decompose} as follows
\begin{lemma}
    Fix any $K\in\naturals$, any partition $\{[s_k, \ell_k]\}_{k\in[K]}$ of $[0, 1]$, and any assignment $\{\revratio_k\}_{k\in[K]} \in[0, 1]^K$. The optimal objective value $\Obj{\text{\ref{program:GFT:regular buyer}}}$ of program~\ref{program:GFT:regular buyer} can be lower bounded as
    \begin{align*}
        \Obj{\text{\ref{program:GFT:regular buyer}}} 
        \geq 
        \min_{k\in[K]} 
        \Obj{\text{\hyperref[program:GFT:regular buyer:decompose]{$\mathcal{P}_{\mathrm{REG}}[s_k,\ell_k,\revratio_k]$}}}
    \end{align*}
\end{lemma}
\begin{proof}
    Suppose $(\optquant, H, \revratio, \quant, \val_0, M, L)$ is an optimal solution of program~\ref{program:GFT:regular buyer}. Invoking \Cref{lem:regular buyer program:reducing variables}, it is without loss of generality to assume $H = 1$ and $L = \LOverBar$. Let $k\primed$ be the index such that $\optquant\in [s_{k\primed}, \ell_{k\primed}]$. By definition, we have 
    \begin{align*}
        \Obj{\text{\ref{program:GFT:regular buyer}}} 
        =
        \Obj{\text{\hyperref[program:GFT:regular buyer:decompose]{$\mathcal{P}_{\mathrm{REG}}[s_{k\primed},\ell_{k\primed},\revratio]$}}}
        \geq 
        \Obj{\text{\hyperref[program:GFT:regular buyer:decompose]{$\mathcal{P}_{\mathrm{REG}}[s_{k\primed},\ell_{k\primed},\revratio_k]$}}}
        \geq 
        \min_{k\in[K]} 
        \Obj{\text{\hyperref[program:GFT:regular buyer:decompose]{$\mathcal{P}_{\mathrm{REG}}[s_k,\ell_k,\revratio_k]$}}}
    \end{align*}
    where the first inequality holds since $\revratio$ belong to the optimal solution of program~\ref{program:GFT:regular buyer}.
\end{proof}
We report the empirical choice of partition $\{[s_k, \ell_k]\}$ and assignment $\{\revratio_k\}$ in \Cref{tab:GFT:regular buyer:decompose}. For each $k$, we uniformly discretize the feasible region of each variable into $500$ points and then numerically evaluate program~\ref{program:GFT:regular buyer:decompose}, which leads to the numerical lower bound of $\fixedPriceGFTPercentageRegular$.

\begin{table}
    \centering
    \begin{tabular}{|c|c|c|c|c|}
    \hline
    $[s, \ell]$ &
       [0, 0.002] & [0.002, 0.008]  & [0.008, 0.018]  & [0.018, 0.034] \\
       \hline
       $\revratio$ & 0.8 & 0.78 & 0.76 & 0.74
       \\
       \hline
       \multicolumn{5}{c}{}
       \\
       \hline
      $[s, \ell]$  & [0.034, 0.044]  & [0.044, 0.078] & [0.078, 0.1] & [0.1, 1] \\
       \hline
       $\revratio$ & 0.72 & 0.7 & 0.68 & 0.66
       \\
       \hline
    \end{tabular}
    \caption{Parameters used in program~\ref{program:GFT:regular buyer:decompose}.}
    \label{tab:GFT:regular buyer:decompose}
\end{table}

%% file: Paper/apx-numerical-evaluation-mhr-buyer.tex
In this section, we provide additional details for our numerical evaluation of program~\ref{program:GFT:mhr buyer}. It follows a similar approach as the one illustrated in \Cref{apx:numerical evaluation:regular buyer} for program~\ref{program:GFT:regular buyer}.
As the first step, we prove that it suffices to assume $L = \LOverBar$.
\begin{lemma}
\label{lem:mhr buyer program:reducing variables}
    For program~\ref{program:GFT:mhr buyer}, there exists an optimal solution with $L = \LOverBar$. Namely, the optimal objective value of program~\ref{program:GFT:mhr buyer} is the same as the optimal objective value of program~\ref{program:GFT:mhr buyer} with additional constraints (for the minimization) that $L = \LOverBar$
\end{lemma}
\begin{proof}
    We first argue that forcing $L = \LOverBar$ does not change the optimal objective value of the program. To see this, note that $L$ only appears in the objective function, where auxiliary variable $\exanteutilratio$ depends on $L$. By algebra, it can be verified that auxiliary variable $\exanteutilratio$ decreases as $L$ increases. Hence, the objective function is decreasing in $L$ and thus there exists an optimal solution such that $L = \LOverBar$.
\end{proof}
After simplifying program~\ref{program:GFT:mhr buyer} by fixing $L = \LOverBar$, there are still six variables $(\optreserve, H, \revratio,\price,\val_0, M)$, where $\optreserve, H$ appears in the outer minimization, $\revratio$ appears in the middle maximization, and $\price,\val_0, M$ appear in the inner minimization. Next, we partition both the feasible region of $\optreserve$ (i.e., $[1, e]$) and the feasible region of $H$ (i.e., $[1, 2]$) into subintervals. For each partition (specified by the product of two subintervals), we manually set the value of $\revratio$. In this way, we obtain a pure minimization program~\ref{program:GFT:mhr buyer:decompose} parameterized by $(s,\ell, a, b, \revratio)$ with variables $(\optreserve,H, \price,\val_0, M)$ defined as follows:
\begin{align}
\label{program:GFT:mhr buyer:decompose}
\tag{$\mathcal{P}_{\mathrm{MHR}}[s,\ell,a, b,\revratio]$}
\arraycolsep=5.4pt\def\arraystretch{1}
    \begin{array}{llll}
     &\min\limits_{\optreserve, H, \price, \val_0, M}   & 
      \displaystyle\exanteutilratio + \frac{\exanteutilratio}{H + M + \LOverBar} &
      \vspace{10pt}
      \\
      \vspace{10pt}
      &\text{s.t.}
      & \optreserve\in[s, \ell],~
      % & 
      % \\
      % \vspace{10pt}
      % && 
      \price\in \left[\max\left\{
    -\optreserve\cdot \LambertFunc\left(-\frac{\revratio}{e}\right), \revratio\right\},
    - \optreserve\cdot \LambertFunc\left(-\displaystyle\frac{\revratio\ln(\optreserve)}{\optreserve}\right)\cdot \frac{1}{\ln(\optreserve)} \right],~ 
      & 
      \\
      \vspace{10pt}
      && 
      \val_0\in \left[0, \optreserve - \frac{\ln(\optreserve)}{\ln(\optreserve) - \ln\left(\frac{\price}{\revratio}\right)}\cdot (\optreserve - \price)\right],~ 
      % & 
      % \\
      % \vspace{10pt}
      % && 
      H\in[a, b],~
      M\in\left[\MUnderBar, \MOverBar\right]& 
    \end{array}
\end{align}
where auxiliary variables $\exanteutilratio$, $\MUnderBar, \MOverBar, \LOverBar$ are follows the same definitions as program~\ref{program:GFT:mhr buyer}. 
By construction, program~\ref{program:GFT:mhr buyer} can be lower bounded by program~\ref{program:GFT:mhr buyer:decompose} as follows
\begin{lemma}
    Fix any $K\in\naturals$, any partition $\{[s_k, \ell_k]\times[a_k,b_k]\}_{k\in[K]}$ of $[1, e]\times[1,2]$, and any assignment $\{\revratio_k\}_{k\in[K]} \in[0, 1]^K$. The optimal objective value $\Obj{\text{\ref{program:GFT:mhr buyer}}}$ of program~\ref{program:GFT:mhr buyer} can be lower bounded as
    \begin{align*}
        \Obj{\text{\ref{program:GFT:mhr buyer}}} 
        \geq 
        \min_{k\in[K]} 
        \Obj{\text{\hyperref[program:GFT:mhr buyer:decompose]{$\mathcal{P}_{\mathrm{MHR}}[s_k,\ell_k,a_k,b_k,\revratio_k]$}}}
    \end{align*}
\end{lemma}
\begin{proof}
    Suppose $(\optreserve, H, \revratio, \price, \val_0, M, L)$ is an optimal solution of program~\ref{program:GFT:mhr buyer}. Invoking \Cref{lem:mhr buyer program:reducing variables}, it is without loss of generality to assume $H = 1$ and $L = \LOverBar$. Let $k\primed$ be the index such that $\optreserve\in [s_{k\primed}, \ell_{k\primed}]$. By definition, we have 
    \begin{align*}
        \Obj{\text{\ref{program:GFT:mhr buyer}}} 
        &=
        \Obj{\text{\hyperref[program:GFT:mhr buyer:decompose]{$\mathcal{P}_{\mathrm{MHR}}[s_{k\primed},\ell_{k\primed},a_{k\primed},b_{k\primed},\revratio]$}}}
        \geq 
        \Obj{\text{\hyperref[program:GFT:mhr buyer:decompose]{$\mathcal{P}_{\mathrm{MHR}}[s_{k\primed},\ell_{k\primed},a_{k\primed},b_{k\primed},\revratio_k]$}}}
        \\
        &\geq 
        \min_{k\in[K]} 
        \Obj{\text{\hyperref[program:GFT:mhr buyer:decompose]{$\mathcal{P}_{\mathrm{MHR}}[s_k,\ell_k,a_k,b_k,\revratio_k]$}}}
    \end{align*}
    where the first inequality holds since $\revratio$ belong to the optimal solution of program~\ref{program:GFT:mhr buyer}.
\end{proof}
We omit the empirical choice of partition $\{[s_k, \ell_k]\times [a_k, b_k]\}$ and assignment $\{\revratio_k\}$ since the partition size is larger. For each $k$, we uniformly discretize the feasible region of each variable into $100$ points and then numerically evaluate program~\ref{program:GFT:mhr buyer:decompose}, which leads to the numerical lower bound of $\fixedPriceGFTPercentageMHR$.